\documentclass[12pt]{article}
\usepackage{amsmath,graphicx,enumerate,natbib,url,setspace,lscape,longtable}
\usepackage{epsfig,mathrsfs,amsthm,amssymb,color, verbatim,multirow,makecell} 
\usepackage{authblk,float,array,lscape,bm,comment,setspace}
\usepackage{amsfonts,indentfirst,booktabs,blindtext,comment,placeins,mathabx}
\usepackage[colorlinks=true, urlcolor=blue]{hyperref}
\hypersetup{
	linkcolor= blue, 
	citecolor= blue 
}

\usepackage{algorithm}
\usepackage{etoolbox}
\usepackage{subcaption}


\def\1{\mathbf{1}}

\def\sN{\mathsf N}

\def\yb{\mathbf y}

\def\zb{\mathbf z}
\def\wb{\mathbf w}

\def\E{\mathbb{E}}

\def\bv{\mathbf{v}}

\def \tbe_t{\tilde{\bm{\epsilon}}_t}

\theoremstyle{plain}
\newtheorem{theorem}{Theorem}[section]
\newtheorem{lemma}{Lemma}[section]
\newtheorem{remark}{Remark}[section]
\newtheorem{proposition}{Proposition}[section]
\newtheorem{corollary}{Corollary}[section]
\newtheorem{example}{Example}[section]
\newtheorem{definition}{Definition}[section]
\newtheorem{assumption}{Assumption}[section]

\def\beq{\begin{equation}}
	\def\eeq{\end{equation}}
\def\ben{\begin{equation*}}
	\def\een{\end{equation*}}
\def\bea{\begin{eqnarray}}
	\def\eea{\end{eqnarray}}
\def\bda{\begin{eqnarray*}}
	\def\eda{\end{eqnarray*}}

\def\bet{\begin{theorem}}
	\def\eet{\end{theorem}}
\def\bel{\begin{lemma}}
	\def\eel{\end{lemma}}
\def\bep{\begin{proposition}}
	\def\eep{\end{proposition}}
\def\bec{\begin{corollary}}
	\def\eec{\end{corollary}}
\def\bc{\begin{center}}
	\def\ec{\end{center}}

\graphicspath{{Fig/}}

\newcommand{\blind}{1}

\numberwithin{equation}{section}

\begin{document}

\def\spacingset#1{\renewcommand{\baselinestretch}%
	{#1}\small\normalsize} \spacingset{1}


\if1\blind
{
	\title{\bf {Bias-Corrected Multiplier Bootstrap Inference for Spectral Edges of Large Covariance Matrices}}
		\author[1]{Xiucai Ding\thanks{Email: \texttt{xcading@ucdavis.edu.}  XCD is partially supported by NSF DMS-2515104.}}
	\author[1]{Yichen Hu\thanks{Email: \texttt{ethhu@ucdavis.edu.}  YCH is partially supported by NSF DMS-2515104. }}
		\author[1]{Jiahui Xie\thanks{Email: \texttt{jihxie@ucdavis.edu.} JHX is partially supported by a research grant from UC Davis (PI: XCD).}}
	\affil[1]{Department of Statistics, University of California, Davis}
	\renewcommand*{\Affilfont}{\small}  
	\renewcommand*{\Authands}{ and }  
	\date{}  
	\maketitle
} \fi

\if0\blind
{
	\bigskip
	\bigskip
	\begin{center}
		{\LARGE\bf A spectral inference method for determining the number of communities in networks}
	\end{center}
	\medskip
} \fi

\vspace*{-8pt}
\begin{abstract}
Inference for spectral edges of large covariance matrices is a fundamental problem in high-dimensional statistics. A natural source of information is provided by the largest non-spiked sample eigenvalues, since they lie near the boundary of the bulk spectrum. A major difficulty is that these eigenvalues asymptotically follow the Tracy--Widom distribution, so direct confidence intervals based on the sample eigenvalues require precise model-specific quantities. These quantities are difficult to estimate in practice under general unknown population covariance structures, making direct Tracy--Widom-based inference practically inconvenient. In this paper, we propose a bias-corrected multiplier bootstrap procedure for inference on the spectral edges. The key idea is to introduce a carefully calibrated multiplier perturbation that regularizes the edge fluctuation to a slightly larger scale at which Gaussian approximation becomes tractable. The confidence interval is constructed directly from bootstrap eigenvalues, together with a data-driven recentering step that corrects the bootstrap-induced shift of the edge. On the theoretical side, we show that, after bias correction, the largest few non-spiked bootstrap eigenvalues are asymptotically Gaussian conditionally on the data. We also establish the asymptotic validity of the proposed confidence interval and show that it provides a data-driven, theoretically justified cutoff for the scree plot. 
\end{abstract}

	\spacingset{1.3} 


\section{Introduction}

Inference for extreme eigenvalues of large covariance matrices is a fundamental problem in high-dimensional statistics \citep{Johnstone2007,YaoZhengBai2015}. It arises naturally in principal component analysis \citep{anderson2003introduction}, factor models \citep{bai2002determining,onatski2009formal,Dobriban2020}, and signal detection \citep{9779233,Nadler2008}, where the spectral edge separates bulk variation from informative low-rank structure. In many modern applications, such as genomics \citep{patterson2006population,Price2006,lappalainen2013transcriptome,1000Genomes2015}, finance \citep{Markowitz1952,LedoitWolf2003,FanLiaoMincheva2013}, and signal processing \citep{nadakuditi2014optshrink,CouilletHachem2011}, one is interested not only in locating the edge of the bulk spectrum, but also in using it to assess whether additional eigenvalues have separated from the bulk and, ultimately, to determine the number of spikes \citep{Johnstone2001}.

A natural approach to infer the spectral edge is to use the sample eigenvalues, since the leading non-spiked eigenvalues concentrate near the boundary of the bulk spectrum and thus provide the most direct information about the edge location \citep{BaiSilverstein2010, knowles2017anisotropic,ErdosYau2017,bloemendal2016principal}. The main difficulty, however, is that these eigenvalues fluctuate on the Tracy--Widom scale \citep{tracy1994level,tracy1996orthogonal,Johnstone2001,el2007tracy,lee2016tracy,bloemendal2016principal}. As a result, valid inference for the spectral edge requires accurate centering by the deterministic edge together with a precise scaling constant. Although the corresponding asymptotic theory is well understood \citep{Bao2015,Ding&Yang2018,9779233,Fan2019separable,knowles2017anisotropic}, the associated centering and scaling quantities are often difficult to estimate efficiently in practice for Tracy--Widom-based statistical inference, especially under general population covariance structures \citep{Karoui2009,paul2009no,lee2016tracy}. This makes direct Tracy--Widom-based inference both theoretically delicate and practically inconvenient.

In this paper, we develop a practical bootstrap-based approach for inference on the spectral edge of large covariance matrices. Our primary goal is to construct a confidence interval for the deterministic edge of the bulk spectrum. Once such a confidence interval is available, it can be used directly to test whether a candidate eigenvalue still belongs to the bulk, and hence to detect the presence of additional spikes. As a further consequence, the same construction also yields a data-driven estimator of the number of spikes.

More specifically, suppose we observe $n$ i.i.d.\ random vectors $\yb_i\in\mathbb R^p$ generated from
\vspace*{-10pt}
\begin{equation}
	\vspace*{-10pt}
	\yb_i=\Sigma^{1/2}\zb_i,
	\label{eq_populationfirstdefinition}
\end{equation}
where $\Sigma$ is a population covariance matrix of general structure that may contain $r$ spikes, and $\zb_i$ has i.i.d.\ entries with unit variance; see Section~\ref{sec_theory} for a more detailed discussion. Our main inferential target is the deterministic right edge of the limiting non-spiked spectrum of the sample covariance matrix \vspace*{-10pt}
\begin{equation}\label{eq_matrixsample}
	\vspace*{-10pt}
	Q= n^{-1}\sum_{i=1}^n (\yb_i-\bar{\yb})(\yb_i-\bar{\yb})^\top,
\end{equation}
where $\bar{\yb}=n^{-1}\sum_{i=1}^n \yb_i.$ More precisely, this edge, denoted by $\mathsf{E}$ (cf.\ \eqref{eq_edgeonehaa}), is the right endpoint of the support of the generalized Marchenko--Pastur law \citep{BaiSilverstein2010}, which describes the limiting spectral distribution of $Q$.

To perform inference, we introduce the multiplier-bootstrap covariance matrix
\vspace*{-10pt}
\begin{equation}
	\vspace*{-10pt}
	Q_{\mathrm{MB}}
	:=
	\frac1n\sum_{i=1}^n \xi_i^2(\yb_i-\bar{\yb})(\yb_i-\bar{\yb})^\top,
	\label{eq_boostrap}
\end{equation}
where the multipliers $\{\xi_i^2\}$ are chosen so that the bootstrap perturbation regularizes the edge fluctuation while preserving the relevant spectral structure.

The main novelty of our approach in Section \ref{sec_proposedapproach} is that the confidence interval is constructed directly from bootstrap eigenvalues, together with a bias-correction step, rather than through direct estimation of the Tracy--Widom centering and scaling constants (see Figure \ref{illustrationintro} for an illustration). This is nonstandard and novel because the multiplier bootstrap changes not only the fluctuation scale but also the deterministic edge itself. Our analysis in Section \ref{sec_theory} shows that, with a suitable choice of multipliers, the bootstrap perturbation enlarges the edge fluctuation to a slightly larger scale where Gaussian approximation becomes tractable, while the induced edge bias can still be consistently corrected. This leads to a practical and fully data-driven procedure for edge inference, and, in turn, a threshold-free estimator for the number of spikes without a sequential testing.


The remainder of the introduction is organized as follows. In Section~\ref{sec_literature}, we review the most relevant literature on spectral edge inference, spike detection, and related bootstrap methods. In Section~\ref{sec_overview}, we summarize our main results and highlight the novelties of the proposed method.

\subsection{Summary of related literature}
\label{sec_literature}

The present paper is related to several strands of literature in high-dimensional statistics and random matrix theory. The first strand concerns the asymptotic behavior of extreme eigenvalues of large covariance matrices. The Tracy--Widom fluctuation of the largest eigenvalue for covariance matrices without spikes was established in the seminal work of \citep{Johnstone2001}, building on the original Tracy--Widom laws in \citep{tracy1994level,tracy1996orthogonal}. This theory has since been extended to broad covariance settings; see, for example, \citep{el2007tracy,FanJohnstone2022,lee2016tracy,knowles2017anisotropic,Bao2015,Ding&Yang2018,Fan2019separable}, and to the leading non-spiked eigenvalues in spiked covariance models \citep{9779233}. These works provide sharp asymptotic descriptions of the edge behavior of the largest non-spiked eigenvalues and therefore form the mathematical foundation for edge inference. However, as mentioned earlier, direct use of these results for statistical inference typically requires accurate estimation of the deterministic edge and the corresponding scaling constant, which can be difficult in practice under general population covariance structures.

A second strand studies spiked covariance models \citep{Johnstone2001,ding2021spiked} and the statistical problem of detecting spikes or estimating their number. In such models,  strong spikes separate from the bulk and generate outlier sample eigenvalues, whereas weaker spikes remain buried at the edge, a phenomenon commonly referred to as the \emph{BBP transition} \citep{BaikBenArousPeche2005}; see, for example, \citep{bloemendal2016principal,ding2021spiked,Paul2007Spiked}. Motivated by this phenomenon, many methods have been developed for estimating the number of spikes, including \citep{bai2002determining}, \citep{onatski2009formal}, \citep{passemier2014estimation}, \citep{braeken2017empirical}, \citep{dobriban2019deterministic}, \citep{9779233}, \citep{fan2022estimating}, and \citep{ke2023estimation}. Most of these procedures rely on eigenvalue gaps, thresholding rules, or sequential testing schemes. By contrast, our approach starts from a confidence interval for the deterministic spectral edge and then uses this interval to induce a threshold-free estimator for the spike number. In this sense, spike detection and spike-number estimation in our framework arise as consequences of edge inference.

A third relevant strand concerns bootstrap and resampling methods for high-dimensional spectral statistics. Such methods are attractive because they can, in principle, avoid direct estimation of delicate centering and scaling quantities. However, for edge eigenvalues, bootstrap approximation is particularly challenging because the native fluctuation occurs on the Tracy--Widom scale and is highly sensitive to the location of the deterministic edge. Existing works in related directions illustrate both the promise and the difficulty of resampling-based spectral inference for global spectral quantities in high dimensions; see, for example, \citep{Dobriban2020,LopesBlandinoAue2019,el2019non,DetteRohde2024}. For individual eigenvalues, the available results are much more limited. In particular, the works of \citep{YuZhaoZhou2025,el2019non} study the standard bootstrap methods for sample covariance matrices and establish useful asymptotic results for sufficiently large spiked eigenvalues. To the best of our knowledge, however, our paper is the first to develop a resampling-based inferential procedure for a non-spiked spectral edge in the high-dimensional regime. Moreover, our method differs from standard multiplier-bootstrap approaches in a more essential way: rather than attempting to directly reproduce the original Tracy--Widom law, we introduce a carefully chosen multiplier perturbation that deliberately regularizes the edge fluctuation to a slightly larger scale where Gaussian approximation becomes available. This  creates a new bias issue, since the multiplier perturbation shifts the deterministic edge itself, and a main novelty of our procedure is to correct this bias through a data-driven recentering step.

Overall, our work connects these three lines of literature. It is rooted in the edge theory of large covariance matrices, is motivated by inferential problems for spikes, and develops a bootstrap-based procedure tailored specifically to edge inference. Compared with existing methods, the proposed approach avoids direct estimation of Tracy--Widom centering and scaling constants, does not require distinct and very large spikes, and provides a unified framework for confidence intervals, testing, and spike-number estimation.
\subsection{Main results and novelties}
\label{sec_overview}

Our main contribution is to propose a new multiplier-bootstrap procedure for inference on the deterministic right edge $\mathsf E$ of the bulk spectrum as in Section \ref{sec_algorithm}. The method is designed to construct a confidence interval for $\mathsf E$ (cf. (\ref{eq_CI})) directly from bootstrap eigenvalues. The key idea is to introduce a carefully calibrated multiplier perturbation (cf. Definition \ref{defn_feasiblemultiplier} and Example \ref{example_one}) that regularizes the original edge fluctuation: the perturbation is chosen to be asymptotically larger than the native Tracy--Widom fluctuation, while still small enough to preserve the relevant spectral structure near the edge. This produces a bootstrap edge fluctuation at a slightly larger scale where Gaussian approximation becomes tractable.

A second novelty is that the proposed procedure is inherently bias-corrected. Unlike standard bootstrap settings, the multiplier perturbation changes not only the fluctuation scale but also the deterministic location of the edge itself. As a result, the bootstrap eigenvalues are centered around a perturbed edge $\mathsf E_{\mathrm{MB}}$ rather than the original edge $\mathsf E$. To address this issue, we introduce a data-driven bias-correction step (cf. (\ref{eq_correctionterm})) based on the difference between the sample eigenvalue and the average of the bootstrap eigenvalues. This correction allows the bootstrap eigenvalues to be properly recentered and leads to a feasible confidence interval for $\mathsf E$. For illustration, we provide a schematic in Figure~\ref{illustrationintro}.

\begin{figure}[ht]
	\includegraphics[width=16.5cm, height=6cm]{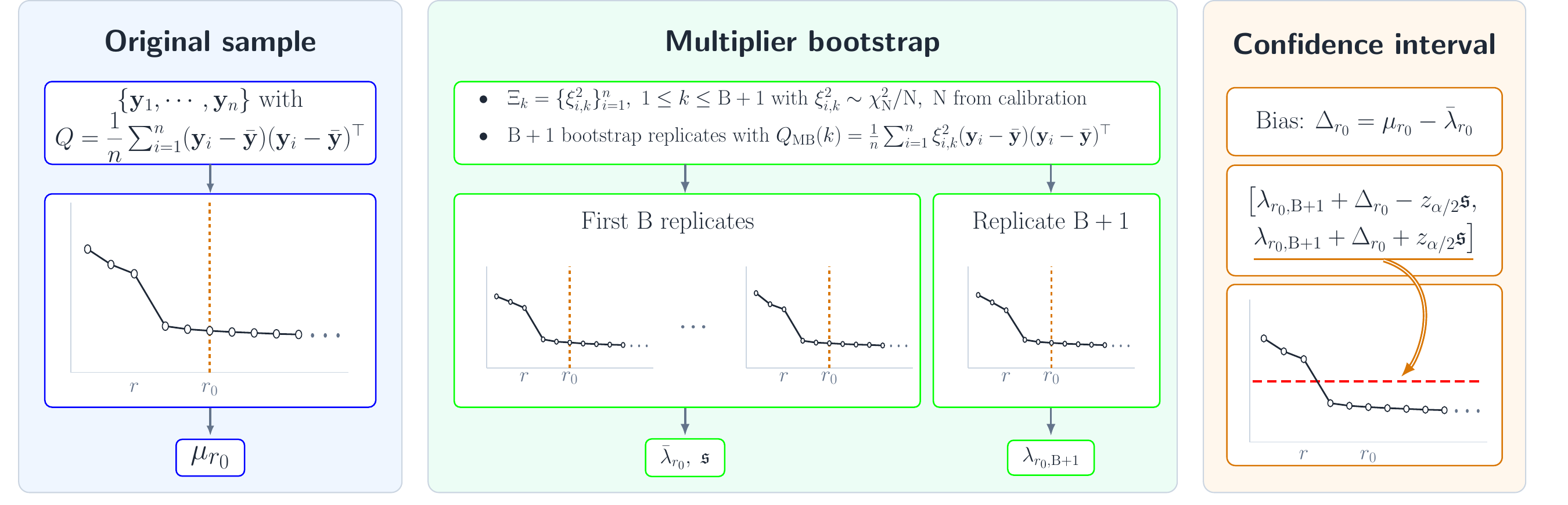} 
	\caption{{\footnotesize Schematic illustration of the proposed procedure. In the original sample, $\mu_{r_0}$ denotes the $r_0$-th largest eigenvalue of the sample covariance matrix, where $r$ is the true number of spikes and $r_0$ is a candidate upper bound chosen so that, the $r_0$-th sample eigenvalue is non-spiked. In the bootstrap step, $\xi_{i,k}^2$ are i.i.d.\ multipliers constructed as in Example~\ref{example_one} and used to form the bootstrap covariance matrices $Q_{\mathrm{MB}}(k)$, with their distribution calibrated through the tuning parameter $\sN$. The first $\mathrm{B}$ bootstrap replicates are used to estimate the bootstrap center $\bar{\lambda}_{r_0}$ and fluctuation scale $\mathfrak{s}$, while the additional replicate $\lambda_{r_0,\mathrm{B}+1}$ is used for the final inferential step. The bias-correction term $\Delta_{r_0}=\mu_{r_0}-\bar{\lambda}_{r_0}$ recenters the bootstrap eigenvalue, leading to the confidence interval for the deterministic edge $\mathsf E$. The upper endpoint of the confidence interval provides a data-driven and theoretically justified cutoff for the scree plot, and can therefore be used to estimate the number of spikes.}}
	\label{illustrationintro}
\end{figure}

The above methodological construction is justified by our first main theoretical result in Theorem \ref{thm_boostrapmain}, which shows that, after bias correction and rescaling, the largest few non-spiked bootstrap eigenvalues are asymptotically Gaussian conditionally on the data. In this sense, the multiplier bootstrap replaces the original Tracy--Widom edge fluctuation by a Gaussian fluctuation at a slightly enlarged scale $n^{-2/3}$ up to a polynomial factor. This conditional Gaussian approximation forms the main theoretical foundation of the proposed procedure; see Section \ref{sec_strate} for more details. 

We then establish the validity of the resulting confidence interval for $\mathsf E$ in Theorem~\ref{thm_power}. Under the null hypothesis corresponding to the case where the interval is constructed from non-spiked eigenvalues, we show that the interval achieves the desired asymptotic coverage probability. Under the alternative, we prove that the interval fails to cover the true edge with probability tending to one, provided that the additional spike separates from the bulk at the scale specified in our assumptions. Notably, this separation requirement is local and does not require the spikes to be distinct; in particular, the procedure continues to apply to possibly degenerate spikes as long as they exceed the edge at a scale much larger than the order of $n^{-1/6}$ (cf. (\ref{eq_spiketruedefinition})). Therefore, the same edge confidence interval yields not only valid inference for $\mathsf E$, but also a  procedure for spike detection.

As a further consequence, the proposed edge confidence interval leads naturally to a threshold-free estimator of the number of spikes (cf.\ \eqref{eq_estimator}). Since this estimator is constructed from the same edge-inference procedure, its theoretical guarantee in Corollary~\ref{cor_spikeestimation} follows from the same comparison between sample eigenvalues and the confidence interval. Equivalently, the procedure yields a data-driven and theoretically justified cutoff for the scree plot, separating eigenvalues associated with spikes from those belonging to the bulk. In this way, confidence interval construction, hypothesis testing, and spike-number estimation are unified within a single framework.

In addition, we develop a practical calibration procedure for selecting the multiplier parameter $\sN$ in Section \ref{sec_choiceofSN}. The rule is based on a Gaussian reference model with explicitly known edge behavior and is motivated by the universality principle in random matrix theory. This yields a simple data-independent choice depending only on the dimension pair $(p,n)$, making the procedure convenient to implement in practice.

Finally, numerical studies in Section \ref{sec_numerical} and real-data analyses in Section \ref{sec:real_data} demonstrate that the proposed procedure performs accurately in finite samples and compares favorably with several existing methods, especially in settings with relatively weak spikes or heterogeneous bulk spectra.

The paper is organized as follows. In Section~\ref{sec_proposedapproach}, we introduce the proposed bias-corrected multiplier-bootstrap procedure for spectral edge inference. In Section~\ref{sec_theory}, we provide the theoretical justification. Numerical studies and real-data analyses are presented in Sections~\ref{sec_numerical} and \ref{sec:real_data}, respectively, and the corresponding code can be found at \href{https://github.com/EthanHuStat/MultiplierBootstrap}{GitHub}. The supplementary file contains additional numerical results, detailed technical proofs, and further discussion of the results.


\section{Our proposed procedure}\label{sec_proposedapproach}

\subsection{Feasible multipliers and their constructions}
One of the key ingredients in high-dimensional multiplier bootstrap is the proper choice of the multipliers. Unlike in low-dimensional settings, the size of the multiplier perturbation plays a crucial role in determining the edge behavior of the bootstrapped covariance matrix. We therefore introduce a class of feasible multipliers as in Definition \ref{defn_feasiblemultiplier} below, designed to produce a perturbation that is asymptotically strong enough to regularize the edge law while still preserving the relevant spectral structure.

\begin{definition}\label{defn_feasiblemultiplier}
	\normalfont
	Throughout the paper, we call a multiplier $\xi^2$ \emph{feasible} if it satisfies the following conditions:
	\begin{enumerate}
		\item Its mean is normalized as
		\begin{equation}\label{eq_meancondition}
			\mathbb{E}\xi^2 = 1.
		\end{equation}
		
		\item  Its variance satisfies
		\begin{equation}\label{eq_variancecondition}
			\operatorname{Var}(\xi^2) \asymp n^{-1/3+\delta},
		\end{equation}
		for some constant $\delta$ such that $2\delta_* < \delta < 1/3$, where $\delta_*>0$ is a sufficiently small constant to be specified later in Definition~\ref{def_OmegaX} of the supplement.
	\end{enumerate}
\end{definition}

\begin{remark}
	The assumptions in Definition~\ref{defn_feasiblemultiplier} guarantee that the multiplier bootstrap remains within the desired perturbative regime required for high-dimensional statistical inference. The mean condition (\ref{eq_meancondition}) preserves the first-order location of the covariance structure. The variance condition (\ref{eq_variancecondition}) ensures that the multiplier-induced fluctuation is asymptotically larger than the native edge fluctuation, which follows a Tracy--Widom law on the scale $n^{-2/3}$. More precisely, the constant $\delta_*$ is tied to the edge rigidity bound of order $n^{-2/3+\delta_*}$ appearing in Definition~\ref{def_OmegaX} of the supplement. Thus, requiring $\delta>2\delta_*$ guarantees that the bootstrap fluctuation dominates the intrinsic Tracy--Widom-scale variation at the edge and yields a Gaussian limit. At the same time, the restriction $\delta<1/3$ ensures that the perturbation remains sufficiently mild so that the relevant spectral structure of the original covariance matrix is preserved.
\end{remark}

To construct feasible multipliers satisfying Definition~\ref{defn_feasiblemultiplier}, one may rescale standard probability distributions so that the resulting multiplier has mean one and variance of the desired order. A convenient way to do this is to introduce an integer-valued tuning parameter $\sN \equiv \sN(p,n)\in\mathbb{Z}$ controlling the magnitude of the perturbation. We next present the following example as concrete constructions of feasible multipliers.

\begin{example}\label{example_one}
	For Chi-squared distribution with $\sN$ degrees of freedom $\chi^2_{\sN},$ let $\xi^2=\chi^2_{\sN}/\sN.$ Then $\E\xi^2=1$ and $\operatorname{Var}(\xi^2)=2/\sN$. Hence, by choosing $\sN \asymp n^{1/3-\delta}$ appropriately for some $\delta$ as in (\ref{eq_variancecondition}), this construction yields a feasible multiplier. 
	
\end{example}

\begin{remark}
	Other feasible multiplier families can also be constructed, for instance from suitably rescaled Beta distributions. Since such alternatives require additional shape parameters to be tuned, we restrict attention to the location-rescaled Chi-squared multiplier as in Example \ref{example_one} in this paper.
\end{remark}

The practical selection of the tuning parameter $\sN$ is deferred to Section \ref{sec_choiceofSN}, where we calibrate it under a Gaussian reference model. This is motivated by the universality principle in random matrix theory, which suggests that the edge behavior of large covariance matrices is largely independent of the precise entry distribution. Thus, Gaussian calibration provides a natural and tractable rule for choosing $\sN$.

\subsection{Edge inference via bias-corrected multiplier-bootstrap eigenvalues}\label{sec_algorithm}

We now introduce our inferential procedure based on multiplier bootstrap as illustrated in Figure \ref{illustrationintro}. The procedure relies on a calibrated multiplier perturbation to regularize the edge fluctuation and produce a Gaussian approximation, together with a bias-correction step that compensates for the systematic shift between the bootstrap and sample eigenvalues. This yields a fully data-driven confidence interval and testing procedure for the spectral edge $\mathsf{E}$. The practical selection of the tuning parameter $\sN$ will be discussed in Section~\ref{sec_choiceofSN}. 

Our inferential procedure can be described in the following four steps. 

\begin{enumerate}
	\item Choose $\sN$ as in Section \ref{sec_choiceofSN} and construct multipliers as in Example \ref{example_one}. 
	\item For a given large integer $\mathrm{B}$, we independently generate $\mathrm{B}+1$ sets of multipliers \vspace*{-10pt}
	\begin{equation}\label{eq_multipliernumberofB}
		\vspace*{-10pt}
		\Xi_k = \{\xi^2_{i,k}\}_{i=1}^n, \qquad 1 \leq k \leq \mathrm{B}+1.
	\end{equation}
	Based on these multipliers, we construct $\mathrm{B}+1$ bootstrap matrices as in \eqref{eq_boostrap} by \vspace*{-10pt}
	\begin{equation*}
		\vspace*{-10pt}
		Q_{\mathrm{MB}}(k) := \frac{1}{n}\sum_{i=1}^n \xi^2_{i,k} (\yb_i-\bar{\yb}) (\yb_i-\bar{\yb})^\top, 
		\qquad 1 \leq k \leq \mathrm{B}+1.
	\end{equation*}
	Let $\lambda_{r_0,k}$ be the $r_0$-th largest eigenvalue of $Q_{\mathrm{MB}}(k)$ for each $1 \leq k \leq \mathrm{B}+1$, where $r_0$ is a  large integer serving as an upper bound for the true number of spikes $r$.
	\item Let $\bar{\lambda}_{r_0}$ denote the sample mean of $\{\lambda_{r_0,k}\}_{k=1}^{\mathrm{B}}$, and define \vspace*{-10pt}
	\begin{equation*}
		\vspace*{-10pt}
		\mathfrak{s}:=\sqrt{\frac{1}{\mathrm{B}}\sum_{k=1}^{\mathrm{B}}(\lambda_{r_0,k}-\bar{\lambda}_{r_0})^2},
	\end{equation*}
	to be their sample standard deviation. Recall $\mu_{r_0}$ is the $r_0$-th largest eigenvalue of the sample covariance matrix  $Q$ in (\ref{eq_matrixsample}). Estimate the bias of $\bar{\lambda}_{r_0}$ to $\mu_{r_0}$ by \vspace*{-10pt} \begin{equation}\label{eq_correctionterm}
		\vspace*{-10pt}
		\Delta_{r_0}:=\mu_{r_0}-\bar{\lambda}_{r_0}.
	\end{equation} 
	\item Given a nominal level $\alpha$, the confidence interval for the spectral edge $\mathsf{E}$ is given by \vspace*{-10pt}
	\begin{equation}\label{eq_CI}
		\vspace*{-10pt}
		\left[ \widehat{\mathsf{E}}^-,   \widehat{\mathsf{E}}^+\right]:=\left[ 
		\lambda_{r_0,\mathrm{B}+1} +\Delta_{r_0}- z_{\alpha/2}\,{\mathfrak{s}}, \;
		\lambda_{r_0,\mathrm{B}+1} 
		+\Delta_{r_0}+ z_{\alpha/2}\,\mathfrak{s}
		\right],
	\end{equation}
	where $z_{\alpha/2}$ denotes the $\alpha/2$ upper quantile of a standard Gaussian random variable.
\end{enumerate}

\begin{remark}\label{rmk_rmkaftermaininferenceprocedure}
	The procedure separates the bootstrap replicates into two distinct roles. The first $\mathrm{B}$ bootstrap eigenvalues are used to estimate the center and fluctuation scale of the bootstrap distribution through $\bar\lambda_{r_0}$ and $\mathfrak{s}$, while the additional replicate $\lambda_{r_0,\mathrm{B}+1}$ is reserved for the final inferential step. In particular, the correction term (\ref{eq_correctionterm}) serves as a data-driven recentering, or bias-correction, term that aligns the bootstrap eigenvalues with the observed sample eigenvalue. Consequently, the quantity $\lambda_{r_0,\mathrm{B}+1}+\Delta_{r_0}$ can be regarded as a recentered bootstrap counterpart of $\mu_{r_0}$, while the width of the interval is governed by the estimated bootstrap fluctuation scale $\mathfrak{s}$, which is of order $\mathrm{O}(n^{-2/3+\delta/2})$ under condition (\ref{eq_variancecondition}). The requirement $\mathrm{B}\gg 1$ ensures that the Monte Carlo error arising from the estimation of $\bar\lambda_{r_0}$ and $\mathfrak{s}$ is asymptotically negligible relative to the edge fluctuation scale. 
	
	Finally, for practical implementation, we take $r_0=\lfloor \mathsf{C}\log n\rfloor,$ for some constant $\mathsf{C}$, for example, $\mathsf{C}=3$. We emphasize that our theory assumes the true number of spikes $r$ is bounded, so this choice of $r_0$ is introduced purely as a practical device for automation. In the simulation studies and real-data analyses below, it leads to good empirical performance.
\end{remark}

We note that the confidence interval in \eqref{eq_CI} can also be used to estimate the number of spikes, since its upper endpoint provides a data-driven inferential estimate of the bulk edge. This removes the need for a fully sequential testing scheme as in \citep{9779233,passemier2014estimation,onatski2009formal}, and makes the procedure computationally efficient. Motivated by this observation, we estimate the number of spikes by \vspace*{-10pt}
\begin{equation}\label{eq_estimator}
	\vspace*{-10pt}
	\widehat r
	:=
	\#\left\{
	1\le i\le r_0:
	\mu_i > \lambda_{r_0,\mathrm{B}+1}+\Delta_{r_0}
	+ z_{\alpha/2}\mathfrak{s}
	\right\}.
	%
\end{equation}

That is, $\widehat r$ counts the number of sample eigenvalues that exceed the upper endpoint of the estimated confidence interval for the spectral edge. This is consistent with the interpretation that spiked eigenvalues are separated from the bulk, whereas non-spiked eigenvalues fluctuate around the edge. Therefore, the proposed thresholding rule yields a natural and unified estimator of the number of spikes, and can also be viewed as providing a data-driven and theoretically justified cutoff for the scree plot.

%

\subsection{Calibration-based selection of $\sN$}\label{sec_choiceofSN}

In this section, we propose a practical way to choose the multiplier parameter $\sN$ by calibrating it under a reference model whose edge can be explicitly calculated. This strategy is motivated by the universality phenomenon in random matrix theory \citep{Akemann2011OxfordRMT,ErdosYau2017}: after appropriate normalization, the edge behavior of large covariance matrices is often largely insensitive to the fine details of the underlying entry distribution, and depends primarily on the dimension pair $(p,n)$ and the local spectral structure near the edge. Therefore, rather than tuning $\sN$ directly from the observed spectrum, we calibrate it under a simple reference model and then transfer the resulting choice to the target problem. Similar ideas have been widely explored in the high-dimensional inference literature; see, for example, \citep{Bao2015,9779233,passemier2014estimation,onatski2009formal}.
More specifically, for each fixed dimension pair $(p,n)$, we select a data-independent value, denoted by $\widehat{\sN}\equiv \widehat{\sN}(p,n)$, using a Wishart matrix \citep{anderson2003introduction} associated with a $p\times n$ data matrix whose entries are i.i.d.\ $\mathcal{N}(0,1)$. For this reference model, the rightmost spectral edge of its limiting spectral distribution is given explicitly by \citep{Johnstone2001} \vspace*{-10pt}
\begin{equation}\label{eq_edge}
	\vspace*{-10pt}
	\mathsf E_{\rm I}
	=
	\left(1+\sqrt{p/n}\right)^2.
\end{equation}

This calibrated value is then used in all subsequent tests and confidence intervals with the same dimension pair. The calibration procedure is summarized as follows.
\begin{enumerate}
	\item For some large $\mathrm{R}_{\rm cal}$ and each dimension pair $(p,n)$, generate $\mathrm{R}_{\rm cal}$ Gaussian matrices \vspace*{-10pt}
	$$
	\vspace*{-10pt}
	\mathbf W_{\ell}
	=
	(\wb_{\ell,1},\ldots,\wb_{\ell,n})
	\in \mathbb R^{p \times n},
	\qquad 1\le \ell\le \mathrm{R}_{\rm cal},
	$$
	where $
	\wb_{\ell,i}
	\stackrel{\mathrm{i.i.d.}}{\sim}
	\mathcal N(0,\mathbf I_p),$ for
	$1\le i\le n, \ 1\le \ell\le \mathrm{R}_{\rm cal}.$
	
	\item Let $\mathcal G$ be a grid of candidate values for $\sN$. To ensure compatibility with \eqref{eq_variancecondition}, we consider
	\begin{equation*}
		\mathcal G
		=
		\left\{
		\left\lfloor \frac{n^{1/3}}{5} \right\rfloor+1,\,
		\left\lfloor \frac{n^{1/3}}{5} \right\rfloor+2,\,
		\ldots,\,
		\left\lfloor 5n^{1/3} \right\rfloor
		\right\},
	\end{equation*}
	and let $\mathcal A$ denote a set of nominal levels used for calibration, for example, $\mathcal A=\{0.01, 0.05,0.10\}.$ For each calibration replication $1\le \ell\le \mathrm{R}_{\rm cal}$, each candidate $\sN'\in\mathcal G$, and each nominal level $\alpha\in\mathcal A$, we apply the multiplier bootstrap procedure in Section~\ref{sec_algorithm} to the sample $\mathbf W_\ell$ and construct the associated confidence interval for $\mathsf E_{\rm I}$ in \eqref{eq_edge} according to \eqref{eq_CI} with $r_0=1$. We denote this interval by\vspace*{-10pt}
	\begin{equation*}
		\vspace*{-10pt}
		\left[
		\widehat{\mathsf E}^{-}_{\ell,\alpha}(\sN'),
		\widehat{\mathsf E}^{+}_{\ell,\alpha}(\sN')
		\right].
	\end{equation*}

	\item For each $\sN'\in\mathcal G$ and $\alpha\in\mathcal A$, we calculate the empirical non-coverage probability by \vspace*{-10pt}
	$$
	\vspace*{-10pt}
	\widehat\alpha(\sN')
	=
	\frac{1}{\mathrm{R}_{\rm cal}}
	\sum_{\ell=1}^{\mathrm{R}_{\rm cal}}
	\mathbf 1
	\left\{
	\mathsf E_{\rm I}
	\notin
	\left[
	\widehat{\mathsf E}^{-}_{\ell,\alpha}(\sN'),
	\widehat{\mathsf E}^{+}_{\ell,\alpha}(\sN')
	\right]
	\right\}.
	$$
	\item We choose the calibrated value $\widehat\sN$ by minimizing the aggregate calibration error over the levels in $\mathcal A$: \vspace*{-10pt}
	\begin{equation*}
		\vspace*{-10pt}
		\widehat\sN
		=
		\arg\min_{\sN'\in\mathcal G}
		\sum_{\alpha\in\mathcal A}
		\left|
		\widehat\alpha(\sN')-\alpha
		\right|.
	\end{equation*}
	If the minimizer is not unique, we select the smallest one.
\end{enumerate}

The resulting value $\widehat\sN$ is therefore a lookup value depending only on the pair $(p,n)$. In all subsequent applications with the same $(p,n)$, we set $\sN=\widehat\sN$. Thus, the calibration is performed once for each $(p,n)$ and does not need to be repeated for every new dataset. This makes the procedure computationally convenient and avoids tuning $\sN$ directly from the observed spectrum, which could otherwise introduce additional instability.



\section{Theoretical justification}\label{sec_theory}
In this section, we provide the theoretical justification of our results. Recall that for any symmetric matrix $H \in \mathbb{R}^{n \times n},$ its empirical spectral distribution (ESD) is defined as \vspace*{-10pt}
\begin{equation*}
	\vspace*{-10pt}
	\mu_n(H):=\frac{1}{n} \sum_{i=1}^n \delta_{\lambda_i(H)},
\end{equation*}
where $\delta_x$ is the Dirac-Delta function.

As throughout the paper we focus on the sample covariance matrix and its bootstrap counterpart, for the purpose of theoretical development we may, without loss of generality, assume that the random vectors in (\ref{eq_populationfirstdefinition}) have mean zero and therefore analyze the corresponding uncentered covariance matrices\vspace*{-10pt}
\begin{equation}\label{eq_rmtmodel}
	\vspace*{-10pt}
	Q_{\text{Sam}}=n^{-1}\Sigma^{1/2}ZZ^\top \Sigma^{1/2}, \ \  Q_{\text{MB}}=n^{-1}\Sigma^{1/2}Z \Xi^2 Z^\top \Sigma^{1/2},
\end{equation}
where $Z=(\zb_i) \in \mathbb{R}^{p \times n}$ and $\Xi^2$ is a $n \times n$ diagonal matrix contains the multipliers $\{\xi_i^2\}.$ When the random vectors have nonzero means, one may first center the data matrix as in (\ref{eq_boostrap}), after which all of the subsequent analysis carries over directly, as discussed in Section 9 of \citep{bloemendal2016principal}. To avoid repetition, we state the following assumption once and use it throughout the paper. This assumption is standard and has been widely adopted in the literature \citep{YaoZhengBai2015}.

\begin{assumption}\label{assum_X}
	For the matrices defined in (\ref{eq_rmtmodel}), we introduce the notation $X=Z/\sqrt{n}.$ For notational simplicity, we assume that the entries of $X=(x_{ij})$ are centered i.i.d. random variables satisfying, for $1\le i\le p$ and $1\le j\le n$, \vspace*{-10pt}
	\begin{equation}\label{eq_standard1n}
		\vspace*{-10pt}
		\mathbb{E}x_{ij}=0,\qquad \mathbb{E}x_{ij}^{2}=\frac{1}{n}.
	\end{equation}
	In addition, we assume that for every $k\in\mathbb{N}$, there exists a constant $C_k>0$, independent of $n$, such that
	$\mathbb{E}\bigl|\sqrt{n}x_{ij}\bigr|^{k}\le C_k.$
\end{assumption}

\subsection{Some background in random matrix theory}\label{sec_background}
In this section, we introduce some background in random matrix theory. For the possibly spiked population covariance matrix $\Sigma$ in (\ref{eq_populationfirstdefinition}), we assume that there exists some deterministic matrix $\Sigma_0$ which admits the spectral decomposition that \vspace*{-10pt}
\begin{equation}\label{eq_basepopulationcovariancematrix}
	\vspace*{-10pt}
	\Sigma_0=\sum_{i=1}^p \sigma_i \bv_i \bv_i^\top. 
\end{equation}
Based on $\Sigma_0,$ we follow \citep{ding2021spiked,ding2024eigenvector} and suppose that $\Sigma$ adapts the following spectral decomposition  \vspace*{-10pt}
\begin{equation}\label{eq_spikedmatrix}
	\vspace*{-10pt}
	\Sigma=\sum_{i=1}^p \widetilde{\sigma}_i \bv_i \bv_i^\top,
\end{equation}
where $\widetilde{\sigma}_i \equiv \sigma_i$ for $r+1 \leq i \leq p$, and $\widetilde{\sigma}_i$, $1 \leq i \leq r$, are the possible spikes when $r\geq 1$. We also allow the case $r=0$, in which $\Sigma=\Sigma_0$ has no spikes.

On the one hand, it is well known that the limiting  spectral distribution (LSD) of  $\Sigma_0^{1/2}XX^\top \Sigma_0^{1/2}$ is characterized by the generalized Marchenko–Pastur law \citep{BaiSilverstein2010}, whose density function is denoted by $\varrho$. We define\vspace*{-10pt}
\begin{equation}\label{eq_edgeonehaa}
	\vspace*{-10pt}
	\mathsf{E} := \sup\{ x \in \mathbb{R}_+ : x \in \operatorname{supp}(\varrho) \}.
\end{equation}
Note that $\mathsf{E}$ admits an explicit characterization. Define \vspace*{-10pt}
\begin{equation}\label{eq_mp}
	\vspace*{-10pt}
	f(x) := -\frac{1}{x} + \frac{1}{n} \sum_{i=1}^p \frac{1}{x + \sigma_i^{-1}}.
\end{equation}
Let $\mathsf{b}$ denote the largest critical point of $f$, that is, the largest solution to $f'(x)=0$. It is well known (see, e.g., Section 2.2 of \citep{knowles2017anisotropic}) that $-\sigma_1^{-1} < \mathsf{b} < 0$ and \vspace*{-10pt}
\begin{equation}\label{eq:edge}
	\vspace*{-10pt}
	\mathsf{E} = f(\mathsf{b}).
\end{equation}

On the other hand, the bootstrap counterpart introduces a small perturbation to the spectrum of $\Sigma_0^{1/2}XX^\top\Sigma_0^{1/2}$ through the multiplier matrix $\Xi^2$ (cf. \eqref{eq_rmtmodel}). More precisely, as shown in Theorem~\ref{lem_solutionsystem1} of the supplement, the LSD of $\Sigma_0^{1/2}X \Xi^2 X^{\top}\Sigma_0^{1/2},$ whose density is denoted by $\varrho_{\mathrm{MB}}$, is characterized by the multiplicative convolution of the generalized Marchenko--Pastur law and the multiplier distribution. We therefore define \vspace*{-10pt}
\begin{equation}\label{eq_edgetwohaa}
	\vspace*{-10pt}
	\mathsf{E}_{\mathrm{MB}}
	:=
	\sup\{x\in\mathbb{R}_+ : x\in \operatorname{supp}(\varrho_{\mathrm{MB}})\}.
\end{equation}

In general, the bootstrap perturbation induces a small discrepancy between the original spectral edge $\mathsf E$ in \eqref{eq_edgeonehaa} and the bootstrap edge $\mathsf{E}_{\mathrm{MB}}$ in \eqref{eq_edgetwohaa}. Denote the resulting edge bias by \vspace*{-10pt}
\begin{equation}\label{eq_bias}
	\vspace*{-10pt}
	\Delta_{\mathrm{edge}}
	:=
	\mathsf E-\mathsf E_{\mathrm{MB}}.
\end{equation}

\begin{remark}\label{eq_biastermrmk}
	As will be shown in the technical proof in Section~\ref{appendix_proofoftheorem32} of the supplement, under (\ref{eq_variancecondition}), the bias term satisfies $\Delta_{\mathrm{edge}} \asymp \operatorname{Var}(\xi^2) \asymp n^{-1/3+\delta}.$ As discussed in Remark~\ref{rmk_rmkaftermaininferenceprocedure} and will be made explicit in Theorem~\ref{thm_boostrapmain}, this bias is not negligible relative to the fluctuation scale $n^{-2/3+\delta/2}$. Consequently, it must be corrected using \eqref{eq_correctionterm}. The resulting correction error is of smaller order than $n^{-2/3+\delta/2}$, as justified in (\ref{eq_biasestimation}) of the supplement. 
\end{remark}

Finally, the following assumptions are adopted throughout this paper and are commonly used in the high-dimensional data analysis literature. 

\begin{assumption}\label{assu_modelassumption}
	We assume that for some constant $0<\tau<1:$
	\begin{enumerate}[(i)]
		\item For high-dimensionality, we assume that $\tau \leq p/n \leq \tau^{-1}.$
		\item For non-spiked population covariance matrix $\Sigma_0$ in (\ref{eq_basepopulationcovariancematrix}), we assume  $\tau \leq \sigma_p \leq \sigma_{p-1} \leq \cdots \leq \sigma_1 \leq \tau^{-1}. $ We also assume Assumption \ref{assum_technique} of the supplement holds. 
		\item For the spikes in (\ref{eq_spikedmatrix}), we assume that $r \geq 0$ is fixed, and when $r \geq 1,$ for all $1 \leq i \leq r$ \vspace*{-10pt}
		\begin{equation}\label{eq_spiketruedefinition}
			\vspace*{-10pt}
			\widetilde{\sigma}_i>-\mathsf{b}^{-1}+\mathfrak{t}, \ \text{where} \ \mathfrak{t} \gtrsim n^{-1/6+\kappa}, 
		\end{equation}
		where $\kappa$ is some positive constant satisfying $\kappa>\delta/2$, with $\delta$ as in \eqref{eq_variancecondition}.
	\end{enumerate}
\end{assumption}

\begin{remark}
	A few remarks on Assumption~\ref{assu_modelassumption} are in order. Condition~(i) is a standard assumption used in the literature to characterize the high-dimensional regime; see, for example, \citep{YaoZhengBai2015,Johnstone2007}. Condition~(ii) imposes regularity assumptions on the non-spiked part of the population covariance matrix. Assumption~\ref{assum_technique}, which will be introduced in the supplement after some necessary notation is developed, is also standard in random matrix theory and is satisfied by many commonly used covariance models \citep{knowles2017anisotropic,el2007tracy, BaiSilverstein2010,Karoui2009,lee2016tracy,Fan2019separable}. At a high level, these assumptions ensure the regular square-root behavior of the density of the limiting spectral distribution associated with the matrices in \eqref{eq_rmtmodel}.
	
	Condition~(iii) characterizes what constitutes a spike in the multiplier-bootstrap setting. Roughly speaking, the proposed procedure is able to identify a spike once \eqref{eq_spiketruedefinition} holds. Although the separation scale $n^{-1/6+\kappa}$ is larger than the classical BBP scale $n^{-1/3+\kappa}$ for sample covariance matrices; see \citep{bloemendal2016principal,ding2021spiked,DingYang2021}, it represents the price paid for using the multiplier bootstrap. Nevertheless, this requirement is still substantially weaker than those appearing in the multiplier bootstrap literature \citep{el2019non,boostsanalysis}, where the corresponding quantity $\mathfrak t$ is required to diverge at the rates $n^{1/2+\upsilon}$ and $n^{1/4+\upsilon}$, respectively, for some constant $\upsilon>0$. Moreover, in much of the existing literature on spike detection, $\mathfrak t$ is typically required to be at least of constant order, or even divergent. In this sense, the proposed method is able to handle weaker spikes, as will be illustrated numerically in Section~\ref{sec_simudetectionofspike}.
\end{remark}


\subsection{Asymptotic normality of the largest non-spiked bootstrap eigenvalues}
In this section, we establish the asymptotic distribution of the largest few non-spiked bootstrap eigenvalues. This result provides the  theoretical foundation for  the proposed bootstrap procedure, and in particular for the confidence interval for the edge $\mathsf E$ in \eqref{eq_CI}.

We first introduce some notation. Recall that $\varrho$ denotes the limiting spectral density associated with \eqref{eq_edgeonehaa}. Let
\begin{equation*}
	m(z):=\int \frac{1}{x-z}\varrho(x)\,\mathrm{d}x,
\end{equation*}
be its Stieltjes transform, where $z\in\mathbb C_+$ lies in the upper half-plane. Throughout the paper, we adopt the convention $m(\mathsf E):=\lim_{\eta\downarrow 0}m(\mathsf E+\mathrm{i}\eta).$ Using this notation together with the eigenvalues $\{\sigma_i\}$ in \eqref{eq_basepopulationcovariancematrix}, we further define
\begin{equation*}
	\mathfrak C_k
	=
	\frac{1}{n}\sum_{i=1}^p
	\frac{\sigma_i^k}{\mathsf E^2\bigl(1+\sigma_i m(\mathsf E)\bigr)^2},
	\qquad k=1,2.
\end{equation*}
We then set \vspace*{-10pt}
\begin{equation}\label{eq_def_v}
	\vspace*{-10pt}
	\mathsf v
	=
	\left(\frac{\mathfrak C_1}{\mathfrak C_2}\right)^2
	\bigl(\mathsf E \times m(\mathsf E)\bigr)^2
	\operatorname{Var}(\xi^2),
\end{equation}
where $\operatorname{Var}(\xi^2)$ denotes the variance of the multipliers in Definition~\ref{defn_feasiblemultiplier}.

The main result of this section is summarized in the following theorem. Recall the edge bias term $\Delta_{\mathrm{edge}}$ defined in \eqref{eq_bias}, and let $\{\lambda_k\}$ denote the eigenvalues of the bootstrap sample covariance matrix in \eqref{eq_rmtmodel}.
\begin{theorem}\label{thm_boostrapmain}
	Consider the multipliers in Example~\ref{example_one}, and assume that they satisfy Definition~\ref{defn_feasiblemultiplier}. Suppose Assumptions \ref{assum_X} and \ref{assu_modelassumption} hold. Then, for any fixed integer $\mathrm{K}>0$ and any $1\le i\le \mathrm{K}$, the following holds with probability at least $1-\mathrm{o}(1)$: \vspace*{-10pt}
	\begin{equation*}
		\vspace*{-10pt}
		\sup_{x\in\mathbb{R}}
		\left|
		\mathbb{P}
		\left(
		\sqrt{n\mathsf{v}^{-1}}
		\bigl(\lambda_{r+i}-\mathsf{E}+\Delta_{\mathrm{edge}}\bigr)
		\le x
		\,\middle|\, X
		\right)
		-\Phi(x)
		\right|
		=
		\mathrm{o}(1),
	\end{equation*}
	where $\mathsf{v}$ is defined in (\ref{eq_def_v}) and $\Phi(x)$ denotes the cumulative distribution function of a standard Gaussian random variable.
\end{theorem}

Theorem~\ref{thm_boostrapmain} shows that, conditionally on the observed data, the largest few non-spiked bootstrap eigenvalues are asymptotically Gaussian after recentering by the deterministic edge $\mathsf E$ together with the bias term $\Delta_{\mathrm{edge}}$, and after scaling by the variance parameter $\mathsf v$. In other words, the multiplier perturbation regularizes the intrinsic edge fluctuation and replaces the Tracy--Widom-type behavior at the original scale by a Gaussian limit at the enlarged scale induced by the multipliers. This Gaussian approximation is the key mechanism underlying the proposed confidence interval in \eqref{eq_CI}.

Several remarks are in order. First, the theorem applies to the largest few non-spiked eigenvalues, rather than only to the single eigenvalue $\lambda_{r+1}$. This is important for subsequent applications to inference on the number of spikes. Second, the appearance of the bias term $\Delta_{\mathrm{edge}}$ reflects the fact that the multiplier bootstrap perturbs not only the fluctuation scale but also the deterministic spectral edge. As discussed earlier in Remark \ref{eq_biastermrmk}, this bias is not negligible at the bootstrap fluctuation scale and therefore must be corrected in the inferential procedure. Third, the variance parameter $\mathsf v$ is proportional to $\operatorname{Var}(\xi^2)$, which makes explicit how the choice of the multiplier distribution governs the size of the Gaussianized fluctuation. 



\subsection{Accuracy and power for the proposed procedure}\label{sec_accuracyandpower}

In this section, we provide the theoretical foundation for the proposed multiplier-bootstrap procedure in Section~\ref{sec_algorithm}. In particular, we establish the accuracy and power of the confidence interval \eqref{eq_CI} for inference on the spectral edge $\mathsf E$ under \vspace*{-10pt}
\begin{equation}
	\vspace*{-10pt}
	\mathbf H_0:r<r_0,
	\qquad
	\mathbf H_a:r\ge r_0,
	\label{eq_nullpypothesis}
\end{equation}
where $r$ denotes the true number of spikes and $r_0$ is a candidate upper bound. This formulation is used to determine whether the confidence interval is constructed from non-spiked eigenvalues, and is therefore mainly introduced for theoretical justification. In practice, as discussed in Remark \ref{rmk_rmkaftermaininferenceprocedure}, one can always choose $r_0$ sufficiently large so that $\mathbf H_0$ holds. As a byproduct, we also establish the asymptotic efficiency of the estimator in \eqref{eq_estimator} for the number of spikes.

\begin{theorem}\label{thm_power}
	Suppose the assumptions of Theorem~\ref{thm_boostrapmain} hold. If the null hypothesis $\mathbf H_0$ in \eqref{eq_nullpypothesis} is true and $\mathrm{B}\gtrsim n^{\mathsf{c}}$ for some constant $\mathsf{c}>0$, then with probability at least $1-\mathrm{o}(1),$ the confidence interval in \eqref{eq_CI} satisfies \vspace*{-10pt}
	\begin{equation}\label{eq_thmpower_one}
		\vspace*{-10pt}
		\mathbb{P}\left(
		\mathsf E \in \left[\widehat{\mathsf E}^-,\,\widehat{\mathsf E}^+\right] \,\middle|\, X
		\right)
		=
		1-\alpha+\mathrm{o}(1).
	\end{equation}
	On the other hand, if the alternative hypothesis $\mathbf H_a$ in \eqref{eq_nullpypothesis} holds, then with probability at least $1-\mathrm{o}(1)$ \vspace*{-10pt}
	\begin{equation}\label{eq_thmpower_two}
		\vspace*{-10pt}
		\mathbb{P}\left(
		\mathsf E \in \left[\widehat{\mathsf E}^-,\,\widehat{\mathsf E}^+\right] \,\middle|\, X
		\right)
		=
		\mathrm{o}(1).
	\end{equation}
\end{theorem}

Theorem~\ref{thm_power} summarizes the basic inferential consequence of the proposed confidence interval. Under the null hypothesis, the interval achieves the desired asymptotic coverage level $1-\alpha$, thereby providing a valid procedure for inference on the spectral edge $\mathsf E$. Under the alternative, the coverage probability vanishes asymptotically, which means that the interval is increasingly unlikely to contain the true edge once an additional spike separates from the bulk. In this sense, the same interval simultaneously yields both valid edge inference under the null and a powerful detection mechanism under the alternative.

From a testing perspective, the theorem implies that the confidence interval in \eqref{eq_CI} induces a consistent decision rule for \eqref{eq_nullpypothesis}. Indeed, accepting $\mathbf H_0$ when $\mathsf E$ is contained in the interval and rejecting otherwise yields asymptotic type-I error $\alpha$ and asymptotic power tending to one. This dual interpretation is particularly useful in our setting, because it connects edge inference directly with spike detection and lays the groundwork for estimating the number of spikes.

The intuition behind the theorem is as follows. Under the null hypothesis, the eigenvalue $\mu_{r_0}$ of the sample covariance matrix behaves like a non-spiked edge eigenvalue, and Theorem~\ref{thm_boostrapmain} shows that its bootstrap analogue admits an asymptotically Gaussian approximation after bias correction and rescaling. In Section~\ref{sec_algorithm}, we use $\bar{\lambda}_{r_0}$ to estimate the bootstrap edge $\mathsf E_{\mathrm{MB}}$ and $\mu_{r_0}$ to estimate the original edge $\mathsf E$. Consequently, the edge bias $\Delta_{\mathrm{edge}}$ is estimated by $\Delta_{r_0}$, while $\mathfrak{s}^2$ estimates the asymptotic variance parameter $\mathsf v/n$. As will be shown in the proof in Section \ref{appendxi_secthreoem32} of the supplement, the corresponding plug-in errors are asymptotically negligible. Under the alternative, however, the eigenvalue $\lambda_{r_0}$ is no longer governed by the non-spiked edge law; instead, it is pushed away from the bulk edge by the additional spike. Consequently, the interval centered at the bootstrap edge analogue no longer covers the true edge, leading to vanishing coverage probability.

Based on Theorem~\ref{thm_power}, we can examine the effectiveness of the estimator in \eqref{eq_estimator} for the number of spikes.

\begin{corollary}\label{cor_spikeestimation}
	Consider the estimator in \eqref{eq_estimator}, and suppose that the candidate value $r_0$ is large such that the null hypothesis $\mathbf H_0$ in \eqref{eq_nullpypothesis} holds. Under the assumptions of Theorem~\ref{thm_boostrapmain}, we have with probability at least $1-\mathrm{o}(1)$ \vspace*{-10pt}
	\begin{equation*}
		\vspace*{-10pt}
		\mathbb{P}(\widehat{r}=r | X)=1-\alpha/2+\mathrm{o}(1).
	\end{equation*}
\end{corollary}

Corollary~\ref{cor_spikeestimation} shows that the proposed estimator $\widehat r$ recovers the true number of spikes with asymptotic probability $1-\alpha/2$, provided that $r_0$ is chosen sufficiently large so that it serves as an upper bound for $r$. Its performance is therefore directly tied to the validity of the confidence interval in \eqref{eq_CI}. In particular, if one allows the nominal level to vary with the sample size and chooses $\alpha \equiv \alpha_n\downarrow 0$, then the corollary yields the consistency statement $\mathbb{P}(\widehat r=r\mid X)\to 1$. At the same time, for practical implementation, fixing a small value of $\alpha$ is often sufficient to obtain accurate and stable finite-sample performance, as confirmed by our numerical experiments.

\subsection{Proof strategies}\label{sec_strate}
At a high level, the proof of Theorem~\ref{thm_boostrapmain} relies on two complementary ingredients: a deterministic comparison between the original edge $\mathsf E$ and its multiplier bootstrap analogue $\mathsf E_{\mathrm{MB}}$, and a conditional Gaussian approximation for the largest non-spiked bootstrap eigenvalues. The key mechanism is that the multiplier bootstrap introduces a perturbation of size $\operatorname{Var}(\xi^2)$, which is asymptotically larger than the intrinsic Tracy--Widom fluctuation, yet still sufficiently small to preserve the bulk spectral geometry. This separation of scales is precisely what makes Gaussian edge inference possible.

The analysis is based on the decomposition \vspace*{-10pt}
\begin{equation}\label{eq_decompositionproofstrategy}
	\vspace*{-10pt}
	\lambda_{r+i}-\mathsf E
	=
	\bigl(\lambda_{r+i}-\widehat{\mathsf E}_{\mathrm{MB}}\bigr)
	+
	\bigl(\widehat{\mathsf E}_{\mathrm{MB}}-\mathsf E_{\mathrm{MB}}\bigr)
	-
	\Delta_{\mathrm{edge}},
	\qquad 1\le i\le \mathrm K,
\end{equation}
where $\Delta_{\mathrm{edge}}=\mathsf E-\mathsf E_{\mathrm{MB}}$, and $\widehat{\mathsf E}_{\mathrm{MB}}$ is the random counterpart of $\mathsf E_{\mathrm{MB}}$ defined through \eqref{eq_edgeequations1} of the supplement. Conditionally on suitably constructed high-probability events, the first term is of Tracy--Widom order $n^{-2/3}$, the second term provides the leading Gaussian fluctuation, and the last term is a deterministic bias. 

To establish the Gaussian fluctuation and derive the corresponding variance formula, we first identify the self-consistent equations characterizing the deterministic right edge $\mathsf E_{\mathrm{MB}}$ in \eqref{eq_edgeequations2} of the supplement, together with its random counterpart $\widehat{\mathsf E}_{\mathrm{MB}}$ associated with a realized multiplier sample in \eqref{eq_edgeequations1} of the supplement. A stability analysis of these edge equations shows that the leading fluctuation of $\widehat{\mathsf E}_{\mathrm{MB}}-\mathsf E_{\mathrm{MB}}$ is an average of multiplier transforms, and therefore admits a conditional Gaussian approximation via the Berry--Esseen theorem; see, for example, \citet[Section 2.11]{vdV1998}. The corresponding variance can be computed explicitly as in \eqref{eq_def_v}. Consequently, the Gaussian fluctuation occurs on the scale $\sqrt{\mathsf v/n}\gg n^{-2/3}$, so the random part of the right-hand side of \eqref{eq_decompositionproofstrategy} is dominated by the second term on the right-hand side, which yields the Gaussian limit.

For the proof of the first statement of Theorem~\ref{thm_power}, Theorem~\ref{thm_boostrapmain} reduces the problem to showing that the unknown deterministic shift $\Delta_{\mathrm{edge}}$ can be consistently estimated by the empirical bias-correction term $\Delta_{r_0}$ in \eqref{eq_correctionterm}. As shown in (\ref{eq_biasestimation}) of the supplement, this plug-in step introduces only a negligible error relative to the Gaussian fluctuation. Combined with a consistent estimate of the variance from the bootstrap replicates, this yields the confidence interval in \eqref{eq_CI}. 

To establish the second statement of Theorem \ref{thm_power}, it remains to analyze the behavior of spiked eigenvalues. Under the alternative hypothesis in (\ref{eq_nullpypothesis}), the $r_0$-th eigenvalue is no longer an edge eigenvalue but instead becomes an outlier. The key point is to show that the center of the confidence interval is separated from $\mathsf E$ by a deterministic gap that is much larger than the interval width and the bias term. Our analysis is based on the decomposition that  \vspace*{-10pt}
\begin{equation}\label{eq_spikepicture}
	\vspace*{-10pt}
	\lambda_{r_0}-\mathsf{E}=\lambda_{r_0}-\widehat{\vartheta}^{\mathtt{MB}}_{r_0}+\widehat{\vartheta}^{\mathtt{MB}}_{r_0}-\vartheta^{\mathtt{MB}}_{r_0}+\vartheta^{\mathtt{MB}}_{r_0}-\vartheta_{r_0}^{\mathtt S}+\vartheta_{r_0}^{\mathtt S}-\mathsf{E},
\end{equation}
where we compare the locations: the sample outlier location $\vartheta_{r_0}^{\mathtt S}$ in \eqref{eq_def_vartheta_S} of the supplement, the deterministic bootstrap outlier location $\vartheta_{r_0}^{\mathtt{MB}}$ in \eqref{eq_def_vartheta_MB} of the supplement, and the random bootstrap outlier location $\widehat{\vartheta}_{r_0}^{\mathtt{MB}}$ associated with the realized multipliers on some suitably constructed probability event. Under the assumption of (\ref{eq_spiketruedefinition}), the last term on the right-hand side of (\ref{eq_spikepicture}) can be bounded from below by a gap of order $ n^{-1/3+2\kappa}$, as shown in \eqref{eq_power_outlier_lower} of the supplement. This term is much larger than both the interval width and the bias term, the latter being of order $\mathrm{O}(n^{-1/3+\delta})$. Moreover, the last but two term of (\ref{eq_spikepicture}) can be controlled by $\mathrm O(\operatorname{Var}(\xi^2)) \asymp \mathrm O(n^{-1/3+\delta}) $, which is obtained by comparing the sample and bootstrap outlier equations through a one-dimensional root perturbation argument as in Lemma \ref{lem_power_bootstrap_gap} of the supplement. The second term fluctuates on the smaller scale $n^{-1/2}\sqrt{\operatorname{Var}(\xi^2)} \asymp n^{-2/3+\delta/2} \ll n^{-1/3+2 \kappa}$ as in (\ref{eq_power_bootstrap_localization}) of the supplement under the assumption of (\ref{eq_spiketruedefinition}), while the first term fluctuates on the scale $ n^{-2/3+\kappa} \ll n^{-1/3+2\kappa}$ as in (\ref{eq_power_outlier_est1}) of the supplement. Taken together, these estimates imply that, under the quantitative separation condition in \eqref{eq_spiketruedefinition}, the center $\lambda_{r_0,\mathrm{B}+1}+\Delta_{r_0}$ stays asymptotically farther away from $\mathsf E$ than the interval half-width, as $\Delta_{r_0}$ can also be controlled by $\mathrm{O}(n^{-1/3+\delta})$ as shown in (\ref{eq_power_delta_small}) of the supplement. Consequently, the interval misses $\mathsf E$ with probability tending to one.

The proof of Corollary~\ref{cor_spikeestimation} combines the null and alternative results through the rule in \eqref{eq_estimator}, where under the alternative hypothesis, the second statement of Theorem~\ref{thm_power} implies rejection with probability tending to one, while under the null hypothesis, the first statement of  Theorem~\ref{thm_power} yields acceptance with probability {$1-\alpha/2+\mathrm{o}(1)$}.

\section{Numerical simulations}\label{sec_numerical}
In this section, we investigate the empirical performance of the proposed multiplier bootstrap procedure introduced in Section~\ref{sec_algorithm}. Section~\ref{sec:simu_setup} describes the simulation setup. In Section~\ref{sec_simuanalysisofouralgorithm}, we examine the finite-sample accuracy and power of the proposed method. In Section~\ref{sec_simudetectionofspike}, we apply the method to the problem of estimating the number of spikes and compare its performance with several existing approaches.


\subsection{Simulation setup}
\label{sec:simu_setup}

Throughout the simulations, we generate samples according to \eqref{eq_populationfirstdefinition}. The entries of $\zb_i$ are taken to be i.i.d.\ either from the standard Gaussian distribution or from a Student's-$t$ distribution. In the main text, we focus on the Gaussian case and defer the results for the Student's-$t$ case to Section~\ref{appendix_additionalnumericalresult} of the supplement. The qualitative conclusions are similar in both settings.

For the possibly spiked population covariance matrix $\Sigma$ in \eqref{eq_populationfirstdefinition}, we adopt the decomposition in \eqref{eq_spikedmatrix} based on its non-spiked counterpart $\Sigma_0$ in \eqref{eq_basepopulationcovariancematrix}. We consider three models for the non-spiked part $\Sigma_0$. The eigenvectors $\{\bv_i\}$ are chosen as the columns of a pre-generated random orthogonal matrix, which is fixed throughout the simulations, while the eigenvalues satisfy one of the following three configurations:
\begin{enumerate}[(I)]
	\item $\sigma_1 \equiv \sigma_2 \equiv \cdots \equiv \sigma_p \equiv \mathfrak{a}$.
	\item The eigenvalues $\sigma_i$ are uniformly spaced over $[0.75,1.25]$.
	\item $\sigma_1 \equiv \sigma_2 \equiv \cdots \equiv \sigma_{\lfloor p/2 \rfloor} \equiv 1.25,$ and $\sigma_{\lfloor p/2 \rfloor +1} \equiv \cdots \equiv \sigma_p \equiv 0.75.$
\end{enumerate}

In the simulations, we set $\mathfrak{a}=0.9$ and consider three dimension pairs, $$(p,n)\in\{(200,500),(500,750),(750,500)\}. $$ Under these settings, the corresponding population edge $\mathsf E$ can be computed explicitly from \eqref{eq:edge} for all three cases.

In the numerical experiments below, for each dimension pair $(p,n)$, we implement the calibration procedure in Section~\ref{sec_choiceofSN} to select $\sN$ and then construct the multipliers as in Example~\ref{example_one}. For the three dimension pairs above, this calibration yields $\sN=4,9,15$, corresponding respectively to $(200,500)$, $(500,750)$, and $(750,500)$.

%

\subsection{Accuracy and power of the proposed procedure}\label{sec_simuanalysisofouralgorithm}
In this section, we investigate the finite-sample accuracy and power of the proposed method in Section \ref{sec_algorithm}.  Throughout, we adopt the simulation setup described in Section~\ref{sec:simu_setup} and choose $\mathrm{B}=2000$ in (\ref{eq_multipliernumberofB}).

To assess accuracy, we report the empirical coverage probability that $\mathsf E$ lies in the confidence interval \eqref{eq_CI} under the null hypothesis in \eqref{eq_nullpypothesis}, for various values of $r_0$ in the case $r_0=r+1$. This allows us to evaluate the finite-sample accuracy of the proposed procedure and to illustrate the conclusion of Theorems~\ref{thm_boostrapmain} and \ref{thm_power}. Similar conclusions continue to hold under the more general null hypothesis $r<r_0$. More specifically, under the null setting $r_0=r+1$, we consider the following three spike configurations in the model \eqref{eq_spikedmatrix}: (a). $r=0$; (b). $r=1$ with $\widetilde{\sigma}_1=6$; (c). $r=2$ with $\widetilde{\sigma}_1=7$ and $\widetilde{\sigma}_2=6$.

Table~\ref{tab:coverage_gaussian} of the supplement reports the finite-sample accuracy of the proposed confidence interval at the nominal levels $\alpha=0.05$ and $\alpha=0.1$ for Gaussian data. The empirical coverage probabilities are generally close to the corresponding nominal levels across different covariance structures, different values of $r$, and different dimension pairs $(p,n)$. This demonstrates that the proposed procedure achieves accurate finite-sample performance under a variety of settings. In Table~\ref{tab:coverage_student_appendix} of the supplement, we report the corresponding results for Student's-$t$ data and obtain similar conclusions.

To assess power, we report the corresponding coverage probabilities under the alternative hypothesis in \eqref{eq_nullpypothesis} with $r_0=2$. To examine how the power changes as the model moves from the null regime $r=r_0-1$ to the alternative regime $r \geq r_0$, we consider a two-spike model in which the first spike is fixed at $7$, while the second spike is placed near the theoretical separation threshold in \eqref{eq_spiketruedefinition}: $ \widetilde{\sigma}_1=7, \ \widetilde{\sigma}_2=-\mathsf{b}^{-1}+\varsigma,$ where $\varsigma>0$ is a constant. As $\varsigma$ increases, the second spike becomes more strongly separated from the edge, so the confidence interval in \eqref{eq_CI} with $r_0=2$ is increasingly unlikely to cover the true spectral edge $\mathsf E$. This setting allows us to evaluate the power of the proposed procedure and to illustrate the conclusion of Theorem~\ref{thm_power}. 


Figure~\ref{fig:power} of the supplement illustrates the empirical power of the proposed procedure under the local alternative. As $\varsigma$ increases, the second spike moves farther above the transition threshold, making it easier to distinguish the associated sample eigenvalue from the bulk edge. Consequently, the confidence interval for the edge is increasingly unlikely to contain the true edge, and the empirical power rises steadily toward one. This behavior is consistent across all three covariance models and all three dimension pairs. Although the finite-sample power curves differ slightly across settings, the overall trend is very stable: once the local separation exceeds a moderate level, the proposed procedure quickly achieves high power. This numerical behavior is in agreement with the theoretical prediction in Theorem~\ref{thm_power}.  In Figure \ref{fig:t10_power_all_dimensions} of the supplement, we report the corresponding results for Student's-$t$ data and obtain similar conclusions.

\subsection{Detection for spikes and comparison with other algorithms}\label{sec_simudetectionofspike}
In this section, we investigate a downstream application of the proposed method, namely estimating the number of spikes, and compare its performance with several existing procedures for spike-number estimation in covariance matrices. Since the literature on this problem is extensive, we focus on a representative, rather than selective, collection of existing methods. Specifically, following the review in Section~\ref{sec_literature}, we compare our method with the procedures of \citep{bai2002determining} (denoted by BN2002), \citep{onatski2009formal} (denoted by Onat2009), \citep{passemier2014estimation} (denoted by PY2014), \citep{braeken2017empirical} (denoted by BA2017), \citep{dobriban2019deterministic} (denoted by DO2019), \citep{9779233} (denoted by DY2022), \citep{fan2022estimating} (denoted by FGZ2022), and \citep{ke2023estimation} (denoted by KML2023). For the proposed procedure, in order to apply the estimator in \eqref{eq_estimator}, we choose a reasonably large value of $r_0$ as discussed in Remark \ref{rmk_rmkaftermaininferenceprocedure}.

For the simulation setup, we consider a possibly five-spiked covariance model. More specifically, in the notation of \eqref{eq_spikedmatrix} and \eqref{eq_spiketruedefinition}, we study a relatively challenging setting in which all spikes are identical, namely, $\widetilde{\sigma}_1=\widetilde{\sigma}_2=\widetilde{\sigma}_3=\widetilde{\sigma}_4=\widetilde{\sigma}_5=-\mathsf b^{-1}+\varsigma, $ where $\varsigma>0$ is a constant. As $\varsigma$ increases, the spikes move farther away from the edge and become more clearly separated, so that the population covariance matrix exhibits five detectable spikes. For the eigenvalues of $\Sigma_0$ in \eqref{eq_basepopulationcovariancematrix}, we consider Cases~(II) and~(III) described in Section~\ref{sec:simu_setup}.

We report the main results for the Gaussian case with $(p,n)=(200,500)$ in Figures~\ref{fig:mean_p200_n500} and~\ref{fig:accuracy_p200_n500} of the supplement. The corresponding results for other dimension pairs and for Student's-$t$ entries are presented in Section~\ref{appendix_additionalnumericalresult} of the supplement, and the qualitative conclusions are similar. In Figure~\ref{fig:mean_p200_n500}, we report the average estimated number of spikes, as $\varsigma$ varies from small to large values. The proposed method increases rapidly with $\varsigma$ in both covariance settings and approaches the true value $r=5$ substantially earlier than most competing procedures. This indicates that the method is able to identify weak spikes soon after they move beyond the transition. Among the competitors, KML2023 and BA2017 also perform reasonably well, but they generally require a larger separation level before stabilizing near the true spike number. By contrast, BN2002, Onat2009, and DY2022 tend to underestimate the number of spikes over a wide range of $\varsigma$, while FGZ2022 and PY2014 improve more gradually. Overall, the proposed method provides more accurate and stable spike-number estimation, especially in the weak-signal regime and in the more challenging two-mass bulk setting of Case~(III).

Moreover, in Figure~\ref{fig:accuracy_p200_n500}, we report the exact detection accuracy, that is, the proportion of 1000 replications in which all five spikes are correctly identified. The proposed method performs competitively over the full range of $\varsigma$, and approaches the theoretical asymptotic accuracy level $1-\alpha/2$. We note that the results are reported for $\alpha=0.01$, which explains why the empirical accuracy does not approach one even when $\varsigma$ is large. As $\alpha \downarrow 0$, this upper limit moves closer to one. Among the competing methods, KML2023 also performs well and eventually achieves nearly perfect accuracy when $\varsigma$ is sufficiently large. However, the proposed method has a clear advantage for small and moderate values of $\varsigma$, particularly in the two-mass-bulk setting, where it identifies the correct number of spikes substantially earlier. DO2019 is also competitive for moderate values of $\varsigma$, whereas BA2017, FGZ2022, and PY2014 improve more gradually as $\varsigma$ increases. BN2002 performs poorly throughout these settings, while Onat2009 and DY2022 require substantially larger values of $\varsigma$ before their accuracy begins to improve. Overall, these results show that the proposed method provides accurate and robust spike-number estimation across both bulk configurations, and is especially advantageous when the spikes are weak. In Section~\ref{appendix_additionalspikecomparision} of the supplement, we also report the results for other combinations of $(p,n)$ and for Student's-$t$ data, and obtain similar conclusions.

				\section{Real data analyses}
				\label{sec:real_data}
				
				In this section, we evaluate the empirical performance of the proposed method for estimating the number of spikes on two real datasets and compare it with the same competing approaches as in Section~\ref{sec_simudetectionofspike}.
				
				\subsection{Data description}
				\subsubsection{GEUVADIS gene expression data}
				The first dataset is the GEUVADIS gene expression dataset~\citep{lappalainen2013transcriptome}, which contains gene expression measurements from lymphoblastoid cell lines from five populations: CEU, FIN, GBR, TSI, and YRI.
				
				We use the gene-level RPKM matrix and apply standard preprocessing steps commonly used for RNA-seq expression data~\citep{watanabe2019genetic}. Specifically, we remove genes with more than $10\%$ missing values, transform the remaining expression measurements using $\log_2(\mathrm{RPKM}+1)$, and impute missing entries using the corresponding gene means. We then randomly sample $400$ individuals and retain the $p=200$ genes with the largest sample variances. The resulting data matrix therefore has dimension $(p,n)=(200,400)$.
				
				\subsubsection{EUR genotype data}
				The second dataset is obtained from the Phase 3 release of the 1000 Genomes Project~\citep{1000Genomes2015}. We restrict the analysis to the EUR super-population, which consists of five  populations: CEU, FIN, GBR, IBS, and TSI.
				
				Starting from the Phase 3 genotype data, we apply standard preprocessing steps for genotype-based multivariate analysis; see, e.g., \cite{prive2020efficient,chang2015second,ke2023estimation}. Specifically, we retain common biallelic variants, encode each variant additively by the minor-allele count, and perform linkage-disequilibrium pruning. After preprocessing, we randomly sample $400$ individuals from the five EUR populations and retain the $p=1500$ variants with the largest sample variances. The final data matrix therefore has dimension $(p,n)=(1500,400)$.
				
				\subsection{Results}
				Since both datasets contain samples from five labeled populations, we use $r=5-1=4$ as a label-based benchmark for the number of spikes. This follows the standard eigenanalysis interpretation that $r+1$ populations give $r$ leading between-population directions, which correspond to $r$ spiked eigenvalues in covariance matrices \citep{patterson2006population,ke2023estimation}.
				
				Figure~\ref{fig:real_data_scree} displays the leading sample eigenvalues and the corresponding multiplier bootstrap upper thresholds of the proposed procedure as in (\ref{eq_estimator}). For the GEUVADIS data, the first four sample eigenvalues exceed the threshold, while the remaining eigenvalues fall below it. The proposed procedure therefore estimates $\widehat r=4$. The same pattern appears for the EUR genotype data, where the proposed procedure also estimates $\widehat r=4$. In both datasets, the estimated number of spikes agrees with the label-based benchmark.
				
				
				\begin{figure}[ht]
					\centering
					\includegraphics[
					width=\textwidth
					]{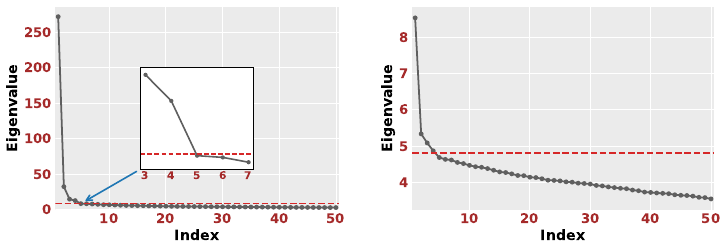}
					\caption{{\footnotesize Leading sample eigenvalues for the two real datasets. The left and the right correspond to 
							GEUVADIS RNA-seq data and EUR genotype data, 
							respectively. For both the left and the right, the red dashed line denotes the
							multiplier bootstrap threshold of the proposed procedure. The inset in
							the left zooms in on the third through seventh eigenvalues. The
							threshold in both the left and the right lies between the fourth and fifth eigenvalues
							and hence gives four estimated spikes.}}
					\label{fig:real_data_scree}
					
				\end{figure}

				Table~\ref{tab:real_data_comparison} reports the estimated numbers of spikes from all methods. For the GEUVADIS RNA-seq data, FGZ2022, KML2023, DO2019, PY2014, and the proposed method estimate four spikes. BA2017 overestimates the number of spikes, whereas BN2002, Onat2009, and DY2022 underestimate it. For the EUR genotype data, as illustrated in Figure~\ref{fig:real_data_scree}, the separation between the first four eigenvalues and the remaining bulk spectrum is weaker, making the estimation of the number of spikes more challenging. In this case, only our proposed method estimates four spikes. BA2017 overestimates the number of spikes, whereas the remaining methods estimate fewer spikes. Overall, the proposed method is the only one among methods compared that agrees with the label-based benchmarks in both datasets. This illustrates the stable performance of the proposed method across different types of real data.
				
				\begin{table}[H]
					\centering
					\resizebox{\textwidth}{!}{%
						\begin{tabular}{lccccccccc}
							\toprule
							& BA2017 & FGZ2022 & KML2023 & BN2002 & DO2019 & Onat2009 & DY2022 & PY2014 & Proposed \\
							\midrule
							GEUVADIS RNA-seq & 31 & $\mathbf{4}$ & $\mathbf{4}$ & 3 & $\mathbf{4}$ & 0 & 2 & $\mathbf{4}$ & $\mathbf{4}$ \\
							EUR genotype & 288 & 3 & 3 & 0 & 3 & 0 & 1 & 2 & $\mathbf{4}$ \\
							\bottomrule
							\vspace{0.2pt}
					\end{tabular}}
					\caption{Estimated numbers of spikes for the GEUVADIS RNA-seq data and the EUR genotype data. Estimates equal to the benchmark $r=4$ are shown in bold.
					}
					\label{tab:real_data_comparison}
				\end{table}

\clearpage
		\begin{center}
			{\large\bf SUPPLEMENTARY MATERIAL}
		\end{center}
		
In this supplement, we provide additional numerical results and detailed technical proofs for the theoretical results. 
			
			
		\setcounter{figure}{0} 
		\setcounter{table}{0}
		\setcounter{section}{0}  
		\counterwithin{table}{section}
		\counterwithin{figure}{section}

\appendix

\section{Additional simulation results}\label{appendix_additionalnumericalresult}

\subsection{Numerical results from the simulation studies}

In this section, we provide the numerical results from the simulation studies.

\begin{table}[ht]
	\centering
	\begingroup
	\footnotesize
	\setlength{\tabcolsep}{5pt}
	\renewcommand{\arraystretch}{1.0}
	\begin{tabular}{llcccccc}
		\toprule
		& & \multicolumn{3}{c}{$1-\alpha=0.95$} & \multicolumn{3}{c}{$1-\alpha=0.90$} \\
		\cmidrule(lr){3-5}\cmidrule(lr){6-8}
		$(p,n)$ & $\Sigma_0$ / $r$ & (a) & (b) & (c) & (a) & (b) & (c) \\
		\midrule
		\multirow{3}{*}{$(200,500)$}
		& (I)   & 0.960 & 0.950 & 0.943 & 0.903 & 0.898 & 0.895 \\
		& (II)  & 0.953 & 0.956 & 0.941 & 0.899 & 0.902 & 0.885 \\
		& (III) & 0.953 & 0.943 & 0.933 & 0.903 & 0.878 & 0.887 \\
		\midrule
		\multirow{3}{*}{$(500,750)$}
		&  (I)   & 0.956 & 0.946 & 0.949 & 0.910 & 0.909 & 0.895 \\
		&  (II)  & 0.959 & 0.945 & 0.946 & 0.912 & 0.889 & 0.895 \\
		&  (III) & 0.955 & 0.944 & 0.939 & 0.901 & 0.889 & 0.898 \\
		\midrule
		\multirow{3}{*}{$(750,500)$}
		&  (I)   & 0.953 & 0.941 & 0.942 & 0.914 & 0.891 & 0.880 \\
		&  (II)  & 0.953 & 0.946 & 0.946 & 0.901 & 0.893 & 0.891 \\
		&  (III) & 0.939 & 0.938 & 0.939 & 0.883 & 0.885 & 0.883 \\
		\bottomrule
		\vspace{0.2pt}
	\end{tabular}
	\endgroup
	\caption{Empirical coverage probabilities of the proposed confidence interval for the spectral edge $\mathsf E$ under Gaussian data. The rows labeled Case (I)--(III) correspond to the three covariance models described in Section~\ref{sec:simu_setup}, while the columns (a)--(c) correspond to the three null configurations for $r$: (a) $r=0$, (b) $r=1$ with $\widetilde{\sigma}_1=6$, and (c) $r=2$ with $\widetilde{\sigma}_1=7$ and $\widetilde{\sigma}_2=6$. The nominal coverage levels are $1-\alpha=0.95$ and $1-\alpha=0.90$. Across the different covariance models, spike configurations, and dimension pairs, the empirical coverage probabilities are generally close to their nominal levels, indicating good finite-sample accuracy of the proposed procedure. The reported results are based on 1000 simulation replications.}
	\label{tab:coverage_gaussian}
\end{table}

\begin{figure}[!ht]
	\centering
	\includegraphics[
	width=1.05\textwidth
	]{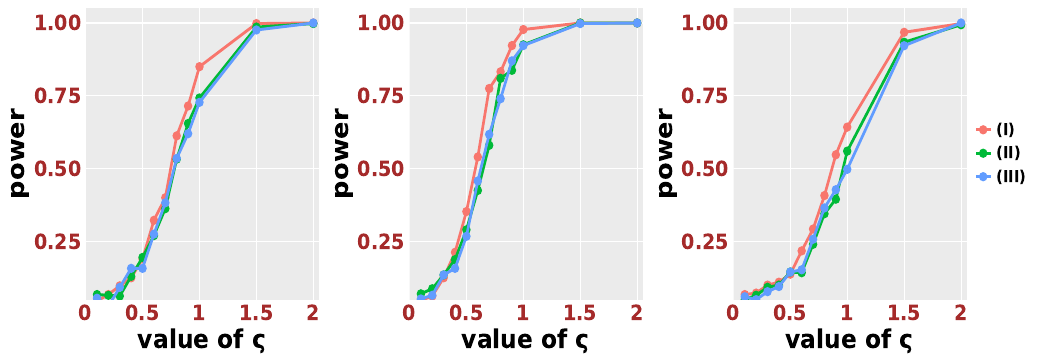}
	\caption{{\footnotesize Empirical power of the proposed procedure under the alternative hypothesis with $r_0=2$. From left to right, the three
			plots correspond to the three dimension pairs $(p,n)=(200,500), (500,750)$ and $(750,500)$, respectively. In all plots, the spike configuration is given by $\widetilde{\sigma}_1=7$ and $\widetilde{\sigma}_2=-\mathsf{b}^{-1}+\varsigma$, where $\varsigma>0$ controls the local separation of the second spike from the threshold in (\ref{eq_spiketruedefinition}). The horizontal axis shows the value of $\varsigma$, and the vertical axis reports the empirical power based on 1000 simulation replications. Within each plot, the three curves labeled (I)--(III) correspond to the covariance models in Section~\ref{sec:simu_setup}.}}
	\label{fig:power}
\end{figure}

\begin{figure}[!ht]
	\centering
	\includegraphics[width=1.05\textwidth]{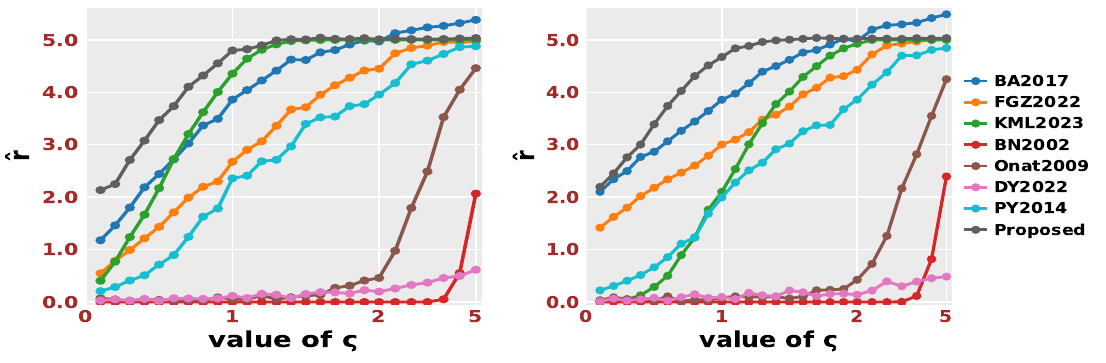}
	\caption{ {\footnotesize
			Comparison of spike-number estimation for $(p,n)=(200,500)$ under Gaussian data. The left and the right correspond to covariance models Case~(II) and Case~(III), respectively, as described in Section~\ref{sec:simu_setup}. In both settings, the population covariance matrix has $r=5$ identical spikes satisfying $\widetilde{\sigma}_1=\cdots=\widetilde{\sigma}_5=-\mathsf b^{-1}+\varsigma.$ The quantity $\widehat r$ denotes the average estimated number of spikes over 1000 simulation replications.}
	}
	\label{fig:mean_p200_n500}
\end{figure}

\begin{figure}[!ht]
	\centering
	\includegraphics[
	width=1.05\textwidth
	]{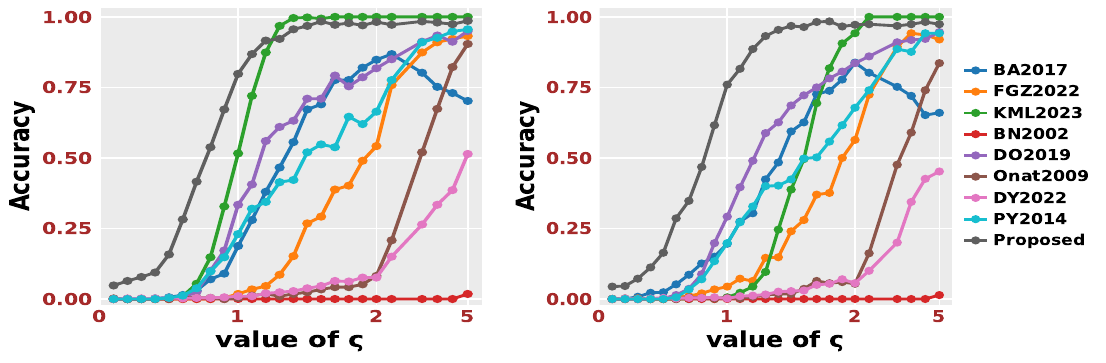}
	\caption{ {\footnotesize
			Comparison of the empirical accuracy of spike-number estimation
			for $(p,n)=(200,500)$ under Gaussian data. The left and the right
			correspond to the covariance matrix models Case~(II) and
			Case~(III), respectively, as described in
			Section~\ref{sec:simu_setup}. In both configurations, the
			population covariance matrix has $r=5$ equal spikes satisfying
			$\widetilde{\sigma}_1=\cdots=\widetilde{\sigma}_5
			=-\mathsf b^{-1}+\varsigma$. Empirical accuracy is defined as the
			proportion of replications in which the estimated number of
			spikes equals $5$. The results are based on $1{,}000$ simulation
			replications. }
	}
	\label{fig:accuracy_p200_n500}
	
\end{figure}

\subsection{Additional results of accuracy and power of the proposed procedure}
\label{appendix_additionalaccuracypower}

This appendix reports the finite-sample accuracy and power of the proposed procedure for standardized Student's-$t_{10}$ data. We use the same simulation setup as in Sections~\ref{sec:simu_setup} and \ref{sec_simuanalysisofouralgorithm}, except that the entries of $\mathbf z_i$ in \eqref{eq_populationfirstdefinition} are drawn i.i.d.\ from $t_{10}\sqrt{8/10}$, so that they have unit variance. This allows us to assess the robustness of the proposed procedure under moderately heavy-tailed data.

Table~\ref{tab:coverage_student_appendix} reports the finite-sample accuracy of the proposed confidence interval at the nominal levels $\alpha=0.05$ and $\alpha=0.1$ for standardized Student's-$t_{10}$ data. The empirical coverage probabilities are generally close to the corresponding nominal levels across different covariance structures, different values of $r$, and different dimension pairs $(p,n)$. These results demonstrate that the proposed procedure maintains accurate finite-sample performance under moderately heavy-tailed distributions.

\begin{table}[!ht]
	\centering
	\begingroup
	\footnotesize
	\setlength{\tabcolsep}{5pt}
	\renewcommand{\arraystretch}{1.0}
	\begin{tabular}{llcccccc}
		\toprule
		& & \multicolumn{3}{c}{$1-\alpha=0.95$} & \multicolumn{3}{c}{$1-\alpha=0.90$} \\
		\cmidrule(lr){3-5}\cmidrule(lr){6-8}
		$(p,n)$ & $\Sigma_0$ / $r$ & (a) & (b) & (c) & (a) & (b) & (c) \\
		\midrule
		\multirow{3}{*}{$(200,500)$}
		& (I)   & 0.956 & 0.953 & 0.955 & 0.906 & 0.905 & 0.914 \\
		& (II)  & 0.950 & 0.953 & 0.947 & 0.903 & 0.911 & 0.907 \\
		& (III) & 0.941 & 0.957 & 0.940 & 0.886 & 0.906 & 0.902 \\
		\midrule
		\multirow{3}{*}{$(500,750)$}
		& (I)   & 0.957 & 0.958 & 0.951 & 0.913 & 0.908 & 0.896 \\
		& (II)  & 0.957 & 0.949 & 0.953 & 0.911 & 0.890 & 0.897 \\
		& (III) & 0.952 & 0.941 & 0.944 & 0.910 & 0.895 & 0.885 \\
		\midrule
		\multirow{3}{*}{$(750,500)$}
		& (I)   & 0.953 & 0.956 & 0.954 & 0.906 & 0.912 & 0.909 \\
		& (II)  & 0.937 & 0.950 & 0.938 & 0.883 & 0.903 & 0.886 \\
		& (III) & 0.944 & 0.952 & 0.950 & 0.897 & 0.895 & 0.911 \\
		\bottomrule
		\vspace{0.2pt}
	\end{tabular}
	\endgroup
	\caption{Empirical coverage probabilities of the proposed confidence interval for the spectral edge $\mathsf E$ under standardized $t_{10}$ data. The rows labeled Case (I)--(III) correspond to the three covariance models described in Section~\ref{sec:simu_setup}, while the columns (a)--(c) correspond to the three null configurations for $r$: (a) $r=0$, (b) $r=1$ with $\widetilde{\sigma}_1=6$, and (c) $r=2$ with $\widetilde{\sigma}_1=7$ and $\widetilde{\sigma}_2=6$. The nominal coverage levels are $1-\alpha=0.95$ and $1-\alpha=0.90$. Across the different covariance models, spike configurations, and dimension pairs, the empirical coverage probabilities are generally close to their nominal levels, indicating good finite-sample accuracy of the proposed procedure. The reported results are based on 1000 simulation replications.}
	\label{tab:coverage_student_appendix}
\end{table}

Figure~\ref{fig:t10_power_all_dimensions} illustrates the empirical power of the proposed procedure under the local alternative. Across all three dimension pairs and covariance models, the empirical power increases steadily with $\varsigma$, and once the local separation exceeds a moderate level, the proposed procedure quickly achieves high power. Similar to the trend in Figure~\ref{fig:power}, this numerical behavior is consistent with the theoretical prediction in Theorem~\ref{thm_power}. These results further demonstrate the robust finite-sample performance of the proposed procedure under moderately heavy-tailed distributions.


\begin{figure}[!ht]
	\centering
	\includegraphics[
	width=1.05\textwidth
	]{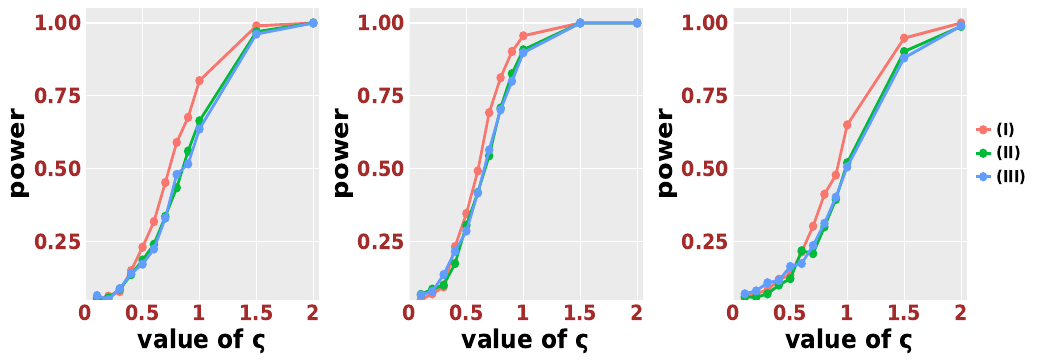}
	\caption{Empirical power of the proposed procedure under the alternative hypothesis with $r_0=2$ or standardized Student-$t_{10}$ data. From left to right, the three
		plots correspond to the three dimension pairs $(p,n)=(200,500), (500,750)$ and $(750,500)$, respectively. In all plots, the spike configuration is given by $\widetilde{\sigma}_1=7$ and $\widetilde{\sigma}_2=-\mathsf{b}^{-1}+\varsigma$, where $\varsigma>0$ controls the local separation of the second spike from the threshold in (\ref{eq_spiketruedefinition}). The horizontal axis shows the value of $\varsigma$, and the vertical axis reports the empirical power based on 1000 simulation replications. Within each plot, the three curves labeled (I)--(III) correspond to the covariance models in Section~\ref{sec:simu_setup}.}
	\label{fig:t10_power_all_dimensions}
	\par
\end{figure}

\subsection{Additional results of detection of spikes and comparison with existing methods}\label{appendix_additionalspikecomparision}

This appendix reports additional simulation results for spike-number estimation by the proposed procedure and the competing methods considered in Section~\ref{sec_simudetectionofspike}. The spike configurations and covariance models are the same as those in Section~\ref{sec_simudetectionofspike}. For Gaussian data, we include the additional dimension pair $(p,n)=(750,500)$. For standardized Student's-$t_{10}$ data, we report results for dimension pairs $(p,n)=(200,500)$ and $(750,500)$.

Figure \ref{fig:mean_p750_n500} shows the average estimated number of spikes under Gaussian data. Across both dimension pairs, the proposed procedure reaches the true value $r=5$ earlier than all competing methods as the separation parameter $\varsigma$ increases. KML2023 and BA2017 also perform reasonably well, but typically require a larger separation before stabilizing near the true spike number. In contrast, BN2002, Onat2009, and DY2022 tend to underestimate the number of spikes over a wider range of $\varsigma$, while FGZ2022 and PY2014 improve more gradually. Moreover,  Figure~\ref{fig:accuracy_p750_n500} reports the corresponding exact detection accuracy under Gaussian data. The proposed method is competitive throughout the range of $\varsigma$ and approaches the nominal asymptotic level $1-\alpha/2$. KML2023 also performs well for large $\varsigma$, eventually achieving nearly perfect accuracy, but the proposed method has a clear advantage for small and moderate values of $\varsigma$, especially in the two-mass Case (III) covariance model. DO2019 is also competitive for moderate $\varsigma$, whereas BA2017, FGZ2022, and PY2014 improve more slowly. BN2002 performs poorly in these settings, while Onat2009 and DY2022 require substantially larger separations before their accuracy improves.

		
		\begin{figure}[!ht]
			\centering
			\includegraphics[
			width=1.05\textwidth
			]{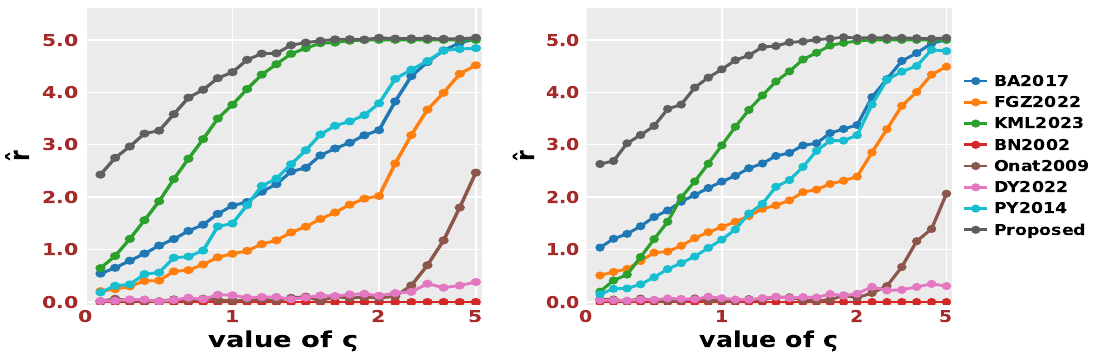}
			\caption{
				Comparison of spike-number estimation for $(p,n)=(750,500)$ under
				Gaussian data. The left and the right correspond to covariance models
				Case~(II) and Case~(III), respectively, as described in
				Section~\ref{sec:simu_setup}. In both settings, the population
				covariance matrix has $r=5$ identical spikes satisfying
				$\widetilde{\sigma}_1=\cdots=\widetilde{\sigma}_5
				=-\mathsf b^{-1}+\varsigma.$ The quantity $\widehat r$ denotes the
				average estimated number of spikes over 1000 simulation replications.
			}
			\label{fig:mean_p750_n500}
		\end{figure}
		
				
				\begin{figure}[!ht]
					\centering
					\includegraphics[
					width=1.05\textwidth
					]{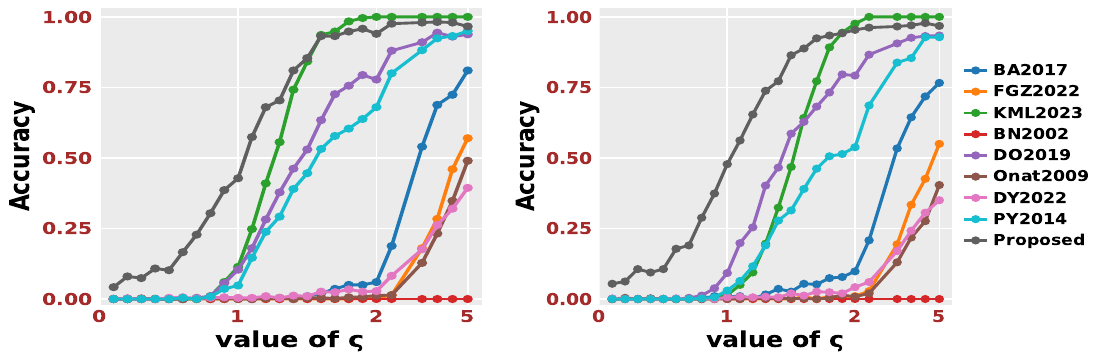}
					\caption{
						Comparison of the empirical accuracy of spike-number estimation
						for $(p,n)=(750,500)$ under Gaussian data. The left and the right
						correspond to the covariance matrix models Case~(II) and
						Case~(III), respectively, as described in
						Section~\ref{sec:simu_setup}. In both configurations, the
						population covariance matrix has $r=5$ equal spikes satisfying
						$\widetilde{\sigma}_1=\cdots=\widetilde{\sigma}_5
						=-\mathsf b^{-1}+\varsigma$. Empirical accuracy is defined as the
						proportion of replications in which the estimated number of spikes
						equals $5$. The results are based on 1000 simulation replications.
					}
					\label{fig:accuracy_p750_n500}
				\end{figure}

				Figures~\ref{fig:mean_t10_p200_n500} and \ref{fig:mean_t10_p750_n500} show the average estimated number of spikes under standardized Student's-$t_{10}$ data, while Figures~\ref{fig:accuracy_t10_p200_n500} and \ref{fig:accuracy_t10_p750_n500} report the exact detection accuracy. The qualitative conclusions are consistent with the Gaussian results: the proposed method remains accurate and robust across different dimension ratios, covariance structures, and entry distributions.
				%

						
						\begin{figure}[!ht]
							\centering
							\includegraphics[
							width=1.05\textwidth
							]{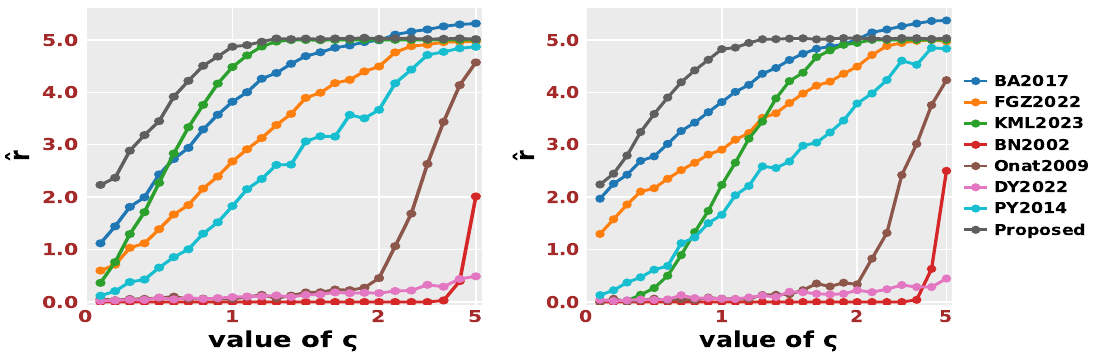}
							\caption{
								Comparison of spike-number estimation for $(p,n)=(200,500)$ under
								standardized $t_{10}$ data. The left and the right correspond to covariance
								models Case~(II) and Case~(III), respectively, as described in
								Section~\ref{sec:simu_setup}. In both settings, the population
								covariance matrix has $r=5$ identical spikes satisfying
								$\widetilde{\sigma}_1=\cdots=\widetilde{\sigma}_5
								=-\mathsf b^{-1}+\varsigma.$ The quantity $\widehat r$ denotes the
								average estimated number of spikes over 1000 simulation replications.
							}
							\label{fig:mean_t10_p200_n500}
						\end{figure}
						%
								
								\begin{figure}[!ht]
									\centering
									\includegraphics[
									width=1.05\textwidth
									]{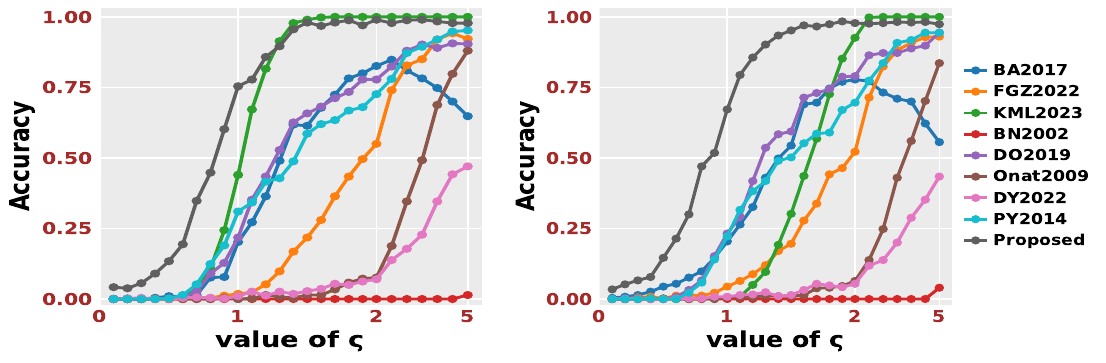}
									\caption{
										Comparison of the empirical accuracy of spike-number estimation
										for $(p,n)=(200,500)$ under standardized $t_{10}$ data. The left and the right correspond to covariance matrix models Case~(II) and
										Case~(III), respectively, as described in
										Section~\ref{sec:simu_setup}. In both configurations, the
										population covariance matrix has $r=5$ equal spikes satisfying
										$\widetilde{\sigma}_1=\cdots=\widetilde{\sigma}_5
										=-\mathsf b^{-1}+\varsigma$. Empirical accuracy is defined as the
										proportion of replications in which the estimated number of spikes
										equals $5$. The results are based on 1000 simulation replications.
									}
									\label{fig:accuracy_t10_p200_n500}
								\end{figure}
								%
								%
								%
								
										
										\begin{figure}[!ht]
											\centering
											\includegraphics[
											width=1.05\textwidth
											]{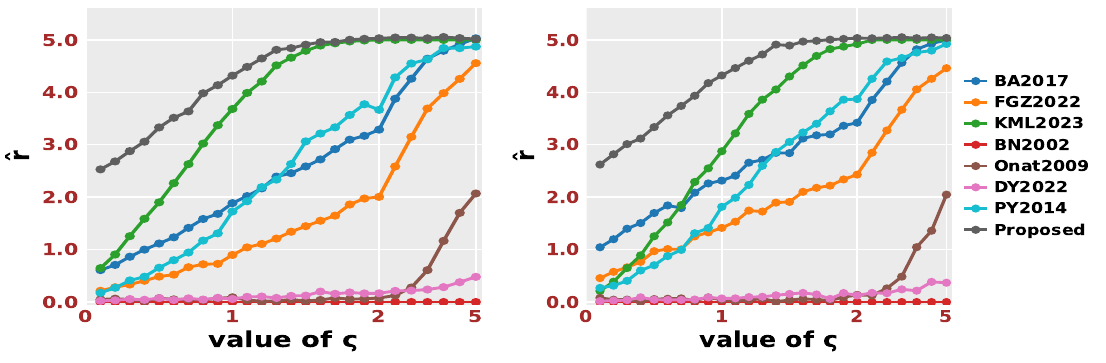}
											\caption{
												Comparison of spike-number estimation for $(p,n)=(750,500)$ under
												standardized $t_{10}$ data. The left and the right correspond to covariance
												models Case~(II) and Case~(III), respectively, as described in
												Section~\ref{sec:simu_setup}. In both settings, the population
												covariance matrix has $r=5$ identical spikes satisfying
												$\widetilde{\sigma}_1=\cdots=\widetilde{\sigma}_5
												=-\mathsf b^{-1}+\varsigma.$ The quantity $\widehat r$ denotes the
												average estimated number of spikes over 1000 simulation replications.
											}
											\label{fig:mean_t10_p750_n500}
										\end{figure}
										
												
												\begin{figure}[!ht]
													\centering
													\includegraphics[
													width=1.05\textwidth
													]{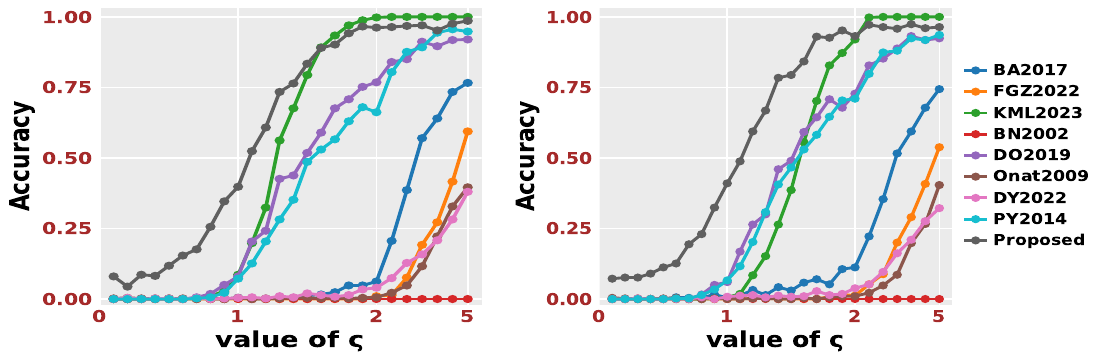}
													\caption{
														Comparison of the empirical accuracy of spike-number estimation
														for $(p,n)=(750,500)$ under standardized $t_{10}$ data. The left and the right correspond to covariance matrix models Case~(II) and
														Case~(III), respectively, as described in
														Section~\ref{sec:simu_setup}. In both configurations, the
														population covariance matrix has $r=5$ equal spikes satisfying
														$\widetilde{\sigma}_1=\cdots=\widetilde{\sigma}_5
														=-\mathsf b^{-1}+\varsigma$. Empirical accuracy is defined as the
														proportion of replications in which the estimated number of spikes
														equals $5$. The results are based on 1000 simulation replications.
													}
													\label{fig:accuracy_t10_p750_n500}
												\end{figure}

\section{Some preliminary results in random matrix theory}

\subsection{Notations, definitions and global laws}
In this subsection, we introduce notation, definitions, and several preliminary results. In terms of \eqref{eq_rmtmodel} and Assumption \ref{assum_X}, we rewrite the observed sample covariance matrix and its non-spiked counterpart by
\begin{align}\label{eq_def_originalcovariance}
	Q_{\mathrm{Sam}}=\Sigma^{1/2}XX^{\top}\Sigma^{1/2},\quad Q_{\mathrm{Sam}}^0:=\Sigma_0^{1/2}XX^{\top}\Sigma_0^{1/2}.
\end{align}
As discussed in Section \ref{sec_background} and in Section 2.2 of \citep{knowles2017anisotropic}, the LSD of $Q_{\mathrm{Sam}}^0$ is characterized by the generalized Marchenko-Pastur law with density $\varrho$, whose Stieltjes transform $m(z)$ satisfies the equation in \eqref{eq_mp} that $f(m(z))=z, \ z\in\mathbb{C}_{+}.$ Equivalently,
\begin{align}\label{eq_systemequation_mp}
	m(z)=\frac{1}{-z+\frac{1}{n}\sum_{i=1}^p\frac{\sigma_i}{1+\sigma_im(z)}}.
\end{align}
Moreover, the rightmost edge $\mathsf{E}$ of $\varrho$ satisfies that \citep{el2007tracy}
\begin{align}\label{eq_systemequation_edge_mp}
	m(\mathsf{E})=\frac{1}{-\mathsf{E}+\frac{1}{n}\sum_{i=1}^p\frac{\sigma_i}{1+\sigma_im(\mathsf{E})}},\quad 1=\frac{\frac{1}{n}\sum_{i=1}^p\frac{\sigma_i^2}{|1+\sigma_im(\mathsf{E})|^2}}{|\mathsf{E}-\frac{1}{n}\sum_{i=1}^p\frac{\sigma_i}{1+\sigma_im(\mathsf{E})}|^2}.
\end{align}

We now define the non-spiked multiplier bootstrap matrix by $Q_{\mathrm{MB}}^0=\Sigma_0^{1/2}X\Xi^2X^{\top}\Sigma_0^{1/2}$, where $\Xi$ is the diagonal multiplier matrix with diagonal entries $\{\xi_i\}_{1\le i\le n}$. We denote the companion matrix of $Q_{\mathrm{MB}}^0$ by $\mathcal{Q}_{\mathrm{MB}}^0=\Xi X^{\top}\Sigma_0 X\Xi$. The empirical spectral distributions (ESDs) of $Q_{\mathrm{MB}}^0$ and $\mathcal{Q}_{\mathrm{MB}}^0$ are denoted by
\begin{equation}\label{eq_esd}
	\mu_{Q_{\mathrm{MB}}^0}:=\frac{1}{p}\sum_{i=1}^p \delta_{\lambda^0_i},\quad \mu_{\mathcal{Q}_{\mathrm{MB}}^0}:=\frac{1}{n}\sum_{j=1}^n \delta_{\lambda^0_j},
\end{equation}
where $\lambda_i^0:=\lambda_i(Q_{\mathrm{MB}}^0), i=1,\dots, p\wedge n$, while $\lambda_i^0=0,i=p\wedge n,\dots, p\vee n$, since the non-zero eigenvalues of $Q_{\mathrm{MB}}^0$ and $\mathcal{Q}_{\mathrm{MB}}^0$ are the same.

It is well known that these ESDs can be characterized through their Stieltjes transforms:
\begin{equation}\label{eq_mq}
	m_{Q_{\mathrm{MB}}^0}(z):=\int\frac{1}{x-z}\mu_{Q_{\mathrm{MB}}^0},\quad m_{\mathcal{Q}_{\mathrm{MB}}^0}(z):=\int\frac{1}{x-z}\mu_{\mathcal{Q}_{\mathrm{MB}}^0}, \ z \in \mathbb{C}_+.
\end{equation}
In the finite-sample regime, $\mu_{Q_{\mathrm{MB}}^0}$ reflects both the deformation of the spectral distribution of $\Sigma_0$ and the small perturbation induced by the multipliers. For sufficiently large $n$, the Stieltjes transform of $\mu_{Q_{\mathrm{MB}}^0}$ admits a nonrandom deterministic equivalent $m_{n,c}(z)$, which is uniquely characterized by the following system of equations:

\begin{align}\label{eq_systemequation1}
	m_{1n,c}(z)=\frac{1}{n}\sum_{i=1}^p &\frac{\sigma_i}{-z(1+\sigma_im_{2n,c}(z))},\quad m_{2n,c}(z)=\int\frac{t}{-z(1+ tm_{1n,c}(z))}\mathrm{d}F_{\xi^2}(t) , \nonumber\\
	&m_{n,c}(z)=\frac{1}{p}\sum_{i=1}^p\frac{1}{-z(1+\sigma_im_{2n,c}(z))},
\end{align} 
where $F_{\xi^2}(t)$ is the cumulative distribution function of the multiplier $\xi^2$. We impose the following technical assumption, which excludes certain singular behaviors of the limiting spectral density of $Q_{\mathrm{MB}}^0$ and is standard in the random matrix theory literature; see, for example, \cite{Bao2015,Ding&Yang2018,ding2021spiked,9779233,el2007tracy,Fan2019separable,knowles2017anisotropic,lee2016tracy}.
\begin{assumption}\label{assum_technique}
	When multipliers satisfy Definition \ref{defn_feasiblemultiplier}, for $\Sigma_0$ defined in \eqref{eq_basepopulationcovariancematrix}, we assume for some constant $\tau_0>0$ such that 
	\begin{align*}
		\min_{1\le i\le p}|1+\sigma_i m_{2n,c}(\mathsf{E}_{\mathrm{MB}})|\ge\tau_0,\quad \inf_{t\in\operatorname{supp}(F_{\xi^2})}|1+tm_{1n,c}(\mathsf{E}_{\mathrm{MB}})|\ge\tau_0,
	\end{align*}
	where recall that $\mathsf{E}_{\mathrm{MB}}$ is the rightmost edge of the limiting spectral density of $Q_{\mathrm{MB}}^0$.
\end{assumption}

The following result characterizes the limiting spectral density $\varrho_{\mathrm{MB}}$ of $Q_{\mathrm{MB}}^0$ through (\ref{eq_systemequation1}). Its proof follows by adapting the arguments of \cite[Theorem 2]{Karoui2009} and \cite[Theorem 1]{paul2009no}, so we omit the details.
\begin{theorem}\label{lem_solutionsystem1}
	Suppose Assumptions \ref{assu_modelassumption} and \ref{assum_technique}, and Definition \ref{defn_feasiblemultiplier} hold. Then 
	for any $z \in \mathbb{C}_+,$ when $n$ is sufficiently large, there exists a unique solution $(m_{1n,c}(z), m_{2n,c}(z), m_{n,c}(z)) \in \mathbb{C}_+^3$ to the systems of equations in  (\ref{eq_systemequation1}). Moreover, $m_{n,c}(z)$ is the Stieltjes transform associated with probability density function $\varrho_{\mathrm{MB}}\equiv\varrho_{\mathrm{MB},n}$ defined on $\mathbb{R}.$ 
\end{theorem}
For technical convenience, we also introduce the following random discretized analogue of \eqref{eq_systemequation1} for sufficiently large $n$.

\begin{definition}\label{defn_couplesystem} For $z \in \mathbb{C}_+,$  we define the triplets $(m_{1n}(z), m_{2n}(z), m_n(z))\in\mathbb{C}^3_+,$  via the following systems of equations.
	
	\begin{align}\label{eq_systemequation2}
		m_{1n}(z)=\frac{1}{n}\sum_{i=1}^p & \frac{\sigma_i}{-z(1+\sigma_im_{2n}(z))},\quad m_{2n}(z)=\frac{1}{n}\sum_{i=1}^n\frac{\xi_i^2}{-z(1+\xi^2_im_{1n}(z))},\\
		& m_{n}(z)=\frac{1}{p}\sum_{i=1}^p\frac{1}{-z(1+\sigma_im_{2n}(z))}. \nonumber
	\end{align}
\end{definition}
Analogously to Theorem \ref{lem_solutionsystem1}, we have the following result for $(m_{1n}(z),m_{2n}(z),m_n(z))$.
\begin{theorem}\label{lem_solutionsystem2}
	Suppose Assumptions \ref{assu_modelassumption} and \ref{assum_technique} and Definition \ref{defn_feasiblemultiplier} hold. Then conditional on some event $\Omega_{\Xi} \equiv \Omega_{n,\Xi}$ that $\mathbb{P}(\Omega_{\Xi})=1-\mathrm{o}(1),$ for any $z \in \mathbb{C}_+,$ when $n$ is sufficiently large, there exists a unique solution $(m_{1n}(z), m_{2n}(z), m_{n}(z)) \in \mathbb{C}_+^3$ to the systems of equations in  (\ref{eq_systemequation2}). Moreover, $m_n(z)$ is the Stieltjes transform of some probability density function $\widehat{\varrho}_{\mathrm{MB}} \equiv \widehat{\varrho}_{\mathrm{MB},n}$ defined on $\mathbb{R}.$ 
\end{theorem}

\begin{remark}
	Several remarks on Theorem \ref{lem_solutionsystem2} are in order. First, its proof can be obtained by adapting the same arguments used for Theorem \ref{lem_solutionsystem1}; see again \cite[Theorem 2]{Karoui2009} and \cite[Theorem 1]{paul2009no}. Second, the high-probability event $\Omega_{\Xi}$ can be constructed explicitly as in Definition \ref{def_OmegaM} and Lemma \ref{lem_probabilitycontrol_M}. Third, we write $\widehat{\mathsf{E}}_{\mathrm{MB}}$ and $\mathsf{E}_{\mathrm{MB}}$ for the rightmost edges of $\widehat{\varrho}_{\mathrm{MB}}$ and $\varrho_{\mathrm{MB}}$, respectively.
\end{remark}

According to \eqref{eq_systemequation1} and \eqref{eq_systemequation2}, we define the following two bivariate functions: 
\begin{align}\label{eq_def_systemequation}
	F_{n,c}(x,y)=\frac{1}{n}\sum_{i=1}^p\frac{\sigma_i}{-y+\sigma_i\int\frac{t}{1+tx}dF_{\xi^2}(t)}-x,\quad F_n(x,y)=\frac{1}{n}\sum_{i=1}^p\frac{\sigma_i}{-y+\frac{\sigma_i}{n}\sum_{j=1}^n\frac{\xi^2_j}{1+x\xi^2_j}}-x.
\end{align}
These bivariate functions are standard characterization of the Stieltjes transforms in Theorem \ref{lem_solutionsystem1} and Theorem \ref{lem_solutionsystem2} in literature, see, for example, \citep{DingYang2021,YuZhaoZhou2025,boostsanalysis}.
Then the pairs $(m_{1n,c}(\mathsf{E}_{\mathrm{MB}}),\mathsf{E}_{\mathrm{MB}})$ and $(m_{1n}(\widehat{\mathsf{E}}_{\mathrm{MB}}),\widehat{\mathsf{E}}_{\mathrm{MB}})$ satisfy the following systems, respectively:
\begin{gather}\label{eq_systemequations_edge}
	F_{n,c}(m_{1n,c}(\mathsf{E}_{\mathrm{MB}}),\mathsf{E}_{\mathrm{MB}})=0,\quad \frac{\partial F_{n,c}}{\partial x}(m_{1n,c}(\mathsf{E}_{\mathrm{MB}}),\mathsf{E}_{\mathrm{MB}})=0,\\
	F_{n}(m_{1n}(\widehat{\mathsf{E}}_{\mathrm{MB}}),\widehat{\mathsf{E}}_{\mathrm{MB}})=0,\quad\frac{\partial F_n}{\partial x}(m_{1n}(\widehat{\mathsf{E}}_{\mathrm{MB}}),\widehat{\mathsf{E}}_{\mathrm{MB}})=0.
\end{gather}
Moreover, by \eqref{eq_systemequation1} and \eqref{eq_systemequations_edge}, we may rewrite
\begin{align*}
	m_{2n,c}(\mathsf{E}_{\mathrm{MB}})=\int\frac{t}{-\mathsf{E}_{\mathrm{MB}}+\frac{1}{n}\sum_{i=1}^p\frac{\sigma_i t}{1+\sigma_im_{2n,c}(\mathsf{E}_{\mathrm{MB}})}}\mathrm{d}F_{\xi^2}(t).
\end{align*}
Comparing this with \eqref{eq_systemequation_edge_mp}, we see under Assumption \ref{assum_technique} that
\begin{align}\label{eq_control_edgebias}
	|\mathsf{E}-\mathsf{E}_{\mathrm{MB}}|=\mathrm{O}(\operatorname{Var}\xi^2),\quad |m(\mathsf{E})-m_{2n,c}(\mathsf{E}_{\mathrm{MB}})|=\mathrm{O}(\operatorname{Var}\xi^2).
\end{align}

To simplify the notation, in the sequel of this paper, we will use the following convention to denote the eigenvalues of $Q^0_{\mathrm{Sam}}$, $Q_{\mathrm{Sam}}$, $Q^0_{\mathrm{MB}}$ and $Q_{\mathrm{MB}}$, respectively.

\begin{table}[!ht]
	\centering
	\captionsetup{justification=centering,singlelinecheck=true}
	\begin{tabular}{llll}
		\toprule
		Model           &   &Sample version & Bootstrapped version\\ \midrule
		\multirow{2}{*}{Non-spiked} & Matrix & $Q_{\mathrm{Sam}}^0:=\Sigma_0^{1/2}XX^\top \Sigma_0^{1/2}$ & $Q_{\mathrm{MB}}^0:=\Sigma_0^{1/2}X\Xi^2X^\top\Sigma_0^{1/2}$\\ \cmidrule(l){3-4} 
		& Eigenvalues & $\{\mu^0_i\}$ &  $\{\lambda^0_i\}$ \\ \midrule
		\multirow{2}{*}{Spiked} &Matrix & $Q_{\mathrm{Sam}}:=\Sigma^{1/2}XX^\top\Sigma^{1/2}$ & $Q_{\mathrm{MB}}:=\Sigma^{1/2}X\Xi^2X^\top\Sigma^{1/2}$  \\ \cmidrule(l){3-4}
		& Eigenvalues & $\{\mu_i\}$ & $\{\lambda_i\}$\\\bottomrule              \vspace{0.2pt}
	\end{tabular}
	\caption{Summary of commonly used notation that appears frequently in the proofs.}\label{table_notations}
\end{table}

\subsection{Some probability events}
We next introduce several high-probability events. First, we define an event measurable with respect to the $\sigma$-algebra generated by the multiplier matrix $\Xi$. Recalling Definition \ref{defn_feasiblemultiplier} and $\Xi^2=\{\xi^2_1,\dots,\xi^2_n\}$, we define
\begin{definition}\label{def_OmegaM}
	Let $C_{\Xi,1}, C_{\Xi,2}, C_{\Xi,3}$ and $C_{\Xi,4}$  be some positive constants and $0<c_{\Xi,1},c_{\Xi,2}, c_{\Xi,3}<1$ are some sufficiently small constants. Denote $\Omega_\Xi\equiv\Omega_{n,\Xi}$ be the event on $\{\xi^2_i\}$ so that the following conditions hold:
	\begin{align*}
		&\frac{1}{n}\sum_{i=1}^n\xi^2_i\leq C_{\Xi,1},\\
		&\left|\frac{1}{n}\sum_{i=1}^n\frac{\xi^2_i}{1+\xi^2_im_{1n,c}(\mathsf{E}_{\mathrm{MB}})}-\int\frac{t}{1+tm_{1n,c}(\mathsf{E}_{\mathrm{MB}})}\mathrm{d}F_{\xi^2}(t)\right|\leq \frac{C_{\Xi,2}n^{c_{\Xi,1}}}{\sqrt{n}},\\
		&\left|\sum_{i=1}^n(\xi^2_i-1)\right|\leq C_{\Xi,3}\sqrt{(n\log^{c_{\Xi,2}}n)\operatorname{Var}\xi^2},\\
		&\frac{1}{n}\sum_{i=1}^n(\xi^2_i-1)^2\leq C_{\Xi,4}\log^{c_{\Xi,3}}n\operatorname{Var}\xi^2.
	\end{align*}
\end{definition}

The following lemma shows that $\Omega_\Xi$ occurs with high probability under Definition \ref{defn_feasiblemultiplier}.

\begin{lemma}\label{lem_probabilitycontrol_M}
	Let $\Omega_\Xi$ be the event defined in Definition \ref{def_OmegaM}, suppose Definition \ref{defn_feasiblemultiplier} holds, we then have that when $n$ is sufficiently large  
	\begin{equation*}
		\mathbb{P}(\Omega_\Xi)=1-\mathrm{O}(\log^{-D_1}n),
	\end{equation*}
	for some constant $D_1>0$.
\end{lemma}
The proof of Lemma \ref{lem_probabilitycontrol_M} will be postponed to Section \ref{sec_prf_Omega_M}. Next, when we study $Q_{\mathrm{MB}}^0$ conditional on $X$, we work on a high-probability event measurable with respect to the $\sigma$-algebra generated by $X$ under Assumption \ref{assum_X}. We denote this event by $\Omega_X$ and specify it below.
\begin{definition}\label{def_OmegaX}
	Denote $\Omega_X$ as the event on $X$ so that the following conditions hold for some small constants $\delta_*$ and $\tau_*$:
	\begin{enumerate}
		\setlength{\itemsep}{0.3em}
		\setlength{\parsep}{0pt}
		\setlength{\topsep}{0.3em}
		\setlength{\labelsep}{0.6em}
		\item [(1)] For any deterministic $p\times p$ matrix $A$, we have all $1\leq i,j\leq n$, the columns of $X$ satisfy
		\begin{align*}
			|\mathbf{x}_i^\top \mathbf{x}_j|=\mathrm{O}\left(n^{c_{X,1}}\sqrt{\frac{\|\mathbf{x}_i\|^2}{n}}\right), \, |\mathbf{x}_i^\top A\mathbf{x}_j|=\mathrm{O}\left(n^{c_{X,1}}\frac{\|A\|_F}{n}\right),\,|\mathbf{x}_i^\top A\mathbf{x}_i-\frac{1}{n}\operatorname{Tr}A|=\mathrm{O}\left(n^{c_{X,1}}\frac{\|A\|_F}{n}\right),
		\end{align*}
		for an arbitrarily small constant $c_{X,1}>0$.
		\item [(2)] $\|XX^\top\|\leq C_{X,1}$ for some constant $C_{X,1}>0$.
		\item [(3)] The eigenvalue rigidity and level repulsion estimates for $Q_{\mathrm{Sam}}^0$ take the form
		\begin{align*}
			&|\mu^0_i-\gamma_i|\leq C_{X,2}(i\wedge n+1-i)^{-1/3}n^{-2/3+c_{X,2}},\quad \text{where}\; n\int_{\gamma_i}^{\infty}\mathrm{d}\varrho=i-\frac{1}{2},\quad i=1,\dots,n;\\
			&|\mu^0_k-\mu^0_{k+1}|\geq C_{X,3}k^{-1/3}n^{-2/3-c_{X,3}}, \quad k=1,\dots, \lfloor \omega p \rfloor;\\
			&|\mu^0_1-\mathsf{E}|\geq C_{X,4} n^{-2/3-c_{X,3}},
		\end{align*}
		for some values $0<\omega<1$, $C_{X,2},C_{X,3},C_{X,4}>0$ and arbitrarily small constants $0<c_{X,2},c_{X,3}\le\delta_*$.
		\item [(4)] For any deterministic diagonal matrix $\Xi^2$, the rigidity of spiked eigenvalues of $Q_{\mathrm{MB}}$ holds as 
		\begin{align*}
			|\lambda_i-\vartheta^{\mathtt{MB}}_i|\le C_{X,5}n^{-1/2+c_{X,4}}\tau_{*}^{1/2},\quad i=1,\dots,r,
		\end{align*}
		where $1+\widetilde{\sigma}_im_{2n}(\vartheta^{\mathtt{MB}}_i)=0$ and $\tau_{*}=\widetilde{\sigma}_i+m_{2n}^{-1}(\widehat{\mathsf{E}}_{\mathrm{MB}})$,
		for some value $C_{X,5}>0$ and arbitrarily small constant $c_{X,4}>0$.
	\end{enumerate}
\end{definition}

The following lemma shows that, under appropriate regularity conditions, the event $\Omega_X$ also holds with high probability. 
\begin{lemma}\label{lem_probabilitycontrol_X}
	Let $\Omega_X$ be the event defined in Definition \ref{def_OmegaX}, under Assumptions \ref{assum_technique}, we have that when $n$ is sufficiently large,
	\begin{align*}
		\mathbb{P}(\Omega_X)=1-\mathrm{O}(n^{-D_2}),
	\end{align*}
	for some large constant $D_2>0$.
\end{lemma}
\begin{proof}
	The results in Definition \ref{def_OmegaX} collects several established results in random matrix theory literature, see \cite{Bao2015,Alex2014,ding2024eigenvector,Ding&Yang2018,knowles2013isotropic,knowles2017anisotropic}. We omit further details here.
\end{proof}

\begin{remark}
	The following remark records several probabilistic conventions. In our setting there are two sources of randomness: the data matrix $X$ and the multiplier matrix $\Xi$.
	
	Since the multiplier matrix $\Xi$ is assumed to be independent of the data matrix $X$, we can realize them on a common product probability space. More specifically, let
	\begin{align*}
		(\widetilde{\Omega}_X,\mathcal{F}_X,\mathbb{P}_X)
		\qquad \text{and} \qquad
		(\widetilde{\Omega}_{\Xi},\mathcal{F}_{\Xi},\mathbb{P}_{\Xi}),
	\end{align*}
	be probability spaces on which $X$ and $\Xi$ are defined, respectively. We then consider the product probability space
	\[
	(\widetilde{\Omega},\mathcal{F},\mathbb{P})
	:=
	(\widetilde{\Omega}_X \times \widetilde{\Omega}_{\Xi},\ \mathcal{F}_X \otimes \mathcal{F}_{\Xi},\ \mathbb{P}_X \otimes \mathbb{P}_{\Xi}),
	\]
	where $\widetilde{\Omega}_X \times \widetilde{\Omega}_{\Xi}$ denotes the Cartesian product of the two sample spaces, $\mathcal{F}_X \otimes \mathcal{F}_{\Xi}$ denotes the product $\sigma$-algebra, and $\mathbb{P}_X \otimes \mathbb{P}_{\Xi}$ denotes the product probability measure. We regard $X$ and $\Xi$ as coordinate random elements on $\widetilde{\Omega}$, namely,
	\begin{align*}
		X(\omega_X,\omega_{\Xi}) := X(\omega_X),
		\qquad
		\Xi(\omega_X,\omega_{\Xi}) := \Xi(\omega_{\Xi}).
	\end{align*}
	In particular, $X$ and $\Xi$ are independent by construction.
	
	Building on the above product probability space, all random quantities considered below— including $Q_{\mathrm{MB}}^0=\Sigma_0^{1/2}X\Xi^2X^\top \Sigma_0^{1/2}$, $\lambda_1^0$, and $\xi_{(1)}^2$—are understood as random variables defined on $(\widetilde{\Omega},\mathcal{F},\mathbb{P})$. By convention, $\mathbb{P}$ denotes the joint probability measure on this product space.
	If an event depends only on $X$ (respectively, only on $\Xi$), then its probability under $\mathbb{P}$ coincides with the corresponding marginal probability under $\mathbb{P}_X$ (respectively, $\mathbb{P}_{\Xi}$). Moreover, $\mathbb{P}(\,\cdot\,\mid X)$ denotes the conditional probability with respect to $\sigma(X)$, as introduced above. 
	
	Based on the above discussion, let $\Omega_{\Xi} \in \sigma(\Xi)$ and $\Omega_X \in \sigma(X)$ denote the “good” events introduced in Definition~\ref{def_OmegaM} and Definition~\ref{def_OmegaX}, respectively. Since both are events in the common $\sigma$-algebra $\mathcal{F}$, their intersection $\Omega_X \cap \Omega_{\Xi}$ is well defined. Moreover, by independence,
	\begin{align*}
		\mathbb{P}(\Omega_X \cap \Omega_{\Xi})
		=
		\mathbb{P}(\Omega_X)\,\mathbb{P}(\Omega_{\Xi}).
	\end{align*}
\end{remark}

\section{Technical proofs}
\subsection{Proof of Theorem \ref{thm_boostrapmain}}\label{appendix_proofoftheorem32}

In this subsection, we prove Theorem \ref{thm_boostrapmain}. As discussed in Section \ref{sec_strate}, our basic idea is the decomposition \eqref{eq_decompositionproofstrategy}.  We start by considering the limiting distributions of the top eigenvalues of the non-spiked covariance matrix $Q_{\mathrm{MB}}^0$, and then transfer these distributions to the largest non-spiked eigenvalues of $Q_{\mathrm{MB}}$ using Lemma \ref{lem_finite_rank_edge_comparison}. To this end, we decompose the largest eigenvalue of $Q_{\mathrm{MB}}^0$ parallel to that of \eqref{eq_decompositionproofstrategy},
\begin{align}\label{eq_decomp_largesteigenvalue}
	\lambda_1^0-\mathsf{E}=\lambda_1^0-\widehat{\mathsf{E}}_{\mathrm{MB}}+\widehat{\mathsf{E}}_{\mathrm{MB}}-\mathsf{E}_{\mathrm{MB}}+\mathsf{E}_{\mathrm{MB}}-\mathsf{E}.
\end{align}
The key point is that, when the multipliers satisfy Definition \ref{defn_feasiblemultiplier}, we have $\widehat{\mathsf{E}}_{\mathrm{MB}}-\mathsf{E}_{\mathrm{MB}}$ contributes the leading Gaussian fluctuation, whereas $\lambda_1^0-\widehat{\mathsf{E}}_{\mathrm{MB}}$ is the intrinsic Tracy-Widom fluctuation and $\Delta_{\mathrm{edge}}=\mathsf{E}-\mathsf{E}_{\mathrm{MB}}$ acts as a deterministic shift, conditional on $X$. 

Now, we begin by showing that $\widehat{\mathsf{E}}_{\mathrm{MB}}-\mathsf{E}_{\mathrm{MB}}$ admits a Gaussian limit. Recall that $\widehat{\mathsf{E}}_{\mathrm{MB}}$ and $\mathsf{E}_{\mathrm{MB}}$ are the rightmost edges of $\widehat{\varrho}_{\mathrm{MB}}$ and $\varrho_{\mathrm{MB}}$, respectively. The pair $(\widehat{\mathsf{E}}_{\mathrm{MB}}, m_{1n}(\widehat{\mathsf{E}}_{\mathrm{MB}}))$ satisfies the system in \eqref{eq_systemequations_edge}:
\begin{gather}\label{eq_edgeequations1}
	m_{1n}(\widehat{\mathsf{E}}_{\mathrm{MB}})=\frac{1}{n}\sum_{i=1}^p \frac{\sigma_i}{-\widehat{\mathsf{E}}_{\mathrm{MB}}+\frac{\sigma_i}{n}\sum_j\frac{\xi^2_j}{1+\xi^2_jm_{1n}(\widehat{\mathsf{E}}_{\mathrm{MB}})}},\\ 1=\frac{1}{n}\sum_{i=1}^p \frac{\frac{\sigma_i^2}{n}\sum_j\frac{\xi^4_j}{|1+\xi^2_jm_{1n}(\widehat{\mathsf{E}}_{\mathrm{MB}})|^2}}{\left|\widehat{\mathsf{E}}_{\mathrm{MB}}-\frac{\sigma_i}{n}\sum_j\frac{\xi^2_j}{1+\xi^2_jm_{1n}(\widehat{\mathsf{E}}_{\mathrm{MB}})}\right|^2}. \nonumber
\end{gather}
Similarly, the pair $(\mathsf{E}_{\mathrm{MB}}, m_{1n,c}(\mathsf{E}_{\mathrm{MB}}))$ satisfies
\begin{gather}\label{eq_edgeequations2}
	m_{1n,c}(\mathsf{E}_{\mathrm{MB}})=\frac{1}{n}\sum_{i=1}^p\frac{\sigma_i}{-\mathsf{E}_{\mathrm{MB}}+\sigma_i\int\frac{t}{1+tm_{1n,c}(\mathsf{E}_{\mathrm{MB}})}\mathrm{d}F_{\xi^2}(t)},\\ 1=\frac{1}{n}\sum_{i=1}^p\frac{\sigma^2_i\int\frac{t^2}{|1+tm_{1n,c}(\mathsf{E}_{\mathrm{MB}})|^2}\mathrm{d}F_{\xi^2}(t)}{|\mathsf{E}_{\mathrm{MB}}-\sigma_i\int\frac{t}{1+tm_{1n,c}(\mathsf{E}_{\mathrm{MB}})}\mathrm{d}F_{\xi^2}(t)|^2}. \nonumber
\end{gather}
Define
\begin{align*}
	&\mathsf{C}_1:=\frac{1}{n} \sum_{i=1}^p \frac{\sigma_i}{ \left( \mathsf{E}_{\mathrm{MB}}- \sigma_i\int\frac{t}{1+tm_{1n,c}(\mathsf{E}_{\mathrm{MB}})}\mathrm{d}F_{\xi^2}(t)\right)^2},\quad \mathsf{C}_2:=\frac{1}{n} \sum_{i=1}^p \frac{\sigma_i^2}{ \left( \mathsf{E}_{\mathrm{MB}}- \sigma_i\int\frac{t}{1+tm_{1n,c}(\mathsf{E}_{\mathrm{MB}})}\mathrm{d}F_{\xi^2}(t)\right)^2},
\end{align*}
and
\begin{equation*}
	\mathcal{X}:=\frac{1}{n} \sum_{j=1}^n \Big( \frac{\xi_j^2}{1+\xi_j^2 m_{1n,c}(\mathsf{E}_{\mathrm{MB}})}-\int\frac{t}{1+tm_{1n,c}(\mathsf{E}_{\mathrm{MB}})}\mathrm{d}F_{\xi^2}(t)  \Big).
\end{equation*}
We have the following lemma that specifies the relationship between $\widehat{\mathsf{E}}_{\mathrm{MB}}$ and $\mathsf{E}_{\mathrm{MB}}$.
\begin{lemma}\label{lem_edgeconvergence1}
	Suppose that the assumptions of Theorem \ref{thm_boostrapmain} hold. When restricted to the event $\Omega_{\Xi}$, we have
	\begin{align}\label{eq_edgeconvergence_firststatement}
		\big|m_{1n,c}(\mathsf{E}_{\mathrm{MB}})-m_{1n}(\widehat{\mathsf{E}}_{\mathrm{MB}})\big|=\mathrm{O}(n^{-1/2+c_{\Xi,1}}),\quad \big|\mathsf{E}_{\mathrm{MB}}-\widehat{\mathsf{E}}_{\mathrm{MB}}\big|=\mathrm{O}(n^{-1/2+c_{\Xi,1}}).
	\end{align}
	As a consequence,
	\begin{align}\label{eq_edgeconvergence_secondstatement}
		\mathsf{C}_1(\widehat{\mathsf{E}}_{\mathrm{MB}}-\mathsf{E}_{\mathrm{MB}})=\mathsf{C}_2\mathcal{X}+\mathrm{O}(n^{-1+2c_{\Xi,1}}),
	\end{align}
	with $\mathsf{C}_1,\mathsf{C}_2\asymp1$.
\end{lemma}
The proof of Lemma \ref{lem_edgeconvergence1} will be postponed to Section \ref{sec_prf_lem_edgeconvergence1}.

Since the multipliers in Definition \ref{defn_feasiblemultiplier} are i.i.d. with vanishing variance, $\mathcal{X}$ admits a CLT mechanism with variance of order $n^{-1}\operatorname{Var}\xi^2$.
Define $\tilde{\mathsf{v}}$ analogously to $\mathsf{v}$ in \eqref{eq_def_v} by
\begin{align*}
	\tilde{\mathsf{v}}=\Big(\frac{\mathsf{C}_1}{\mathsf{C}_2}\Big)^2\operatorname{Var} \left(\frac{\xi^2}{1+\xi^2m_{1n,c}(\mathsf{E}_{\mathrm{MB}})}\right),\quad \mathsf{C}_k=\frac{1}{n}\sum_{i=1}^p\frac{\sigma_i^k}{(\mathsf{E}_{\mathrm{MB}}-\sigma_i\int\frac{t}{1+tm_{1n,c}(\mathsf{E}_{\mathrm{MB}})}\mathsf{d}F_{\xi^2}(t))^2}.
\end{align*}
The above discussion implies the following theorem.
\begin{theorem}\label{thm_conditionalclt_nonspike}
	Suppose Assumptions \ref{assu_modelassumption} and \ref{assum_technique} hold, and assume Definition \ref{def_regularedge}. Moreover, we assume that $\{\xi_i^2\}$ are feasible multipliers constructed as in 
	Example \ref{example_one} and satisfying Definition \ref{defn_feasiblemultiplier}.
	Then, for $n$ sufficiently large, there exists some constant $\mathsf{c}>0$ such that with probability at least $1-\mathrm{O}(n^{-\mathsf{c}})$,
	\begin{align}\label{eq_conditionalclt_nonspike}
		\sup_{x\in\mathbb R}\left|\mathbb{P} \left( \sqrt{n \tilde{\mathsf{v}}^{-1}} (\lambda^0_{1}-\mathsf{E}+\Delta_{\mathrm{edge}}) \leq x \big|X \right)-\Phi(x)\right|=\mathrm{o}(1).
	\end{align}
\end{theorem}
\begin{proof}
	We define
	\begin{align*}
		\mathtt{Y}_n:=\sqrt{n/\tilde{\mathsf{v}}}(\lambda^0_1-\mathsf{E}_{\mathrm{MB}}),\quad \mathtt{Z}_n:=\sqrt{n/\tilde{\mathsf{v}}} \frac{\mathsf{C}_2}{\mathsf{C}_1} \mathcal{X},
	\end{align*}
	where $\operatorname{Var}(\mathcal{X})=(\mathsf{C}_1^2/\mathsf{C}_2^2)\tilde{\mathsf{v}}n^{-1}$ and note $\mathsf{E}-\Delta_{\mathrm{edge}}=\mathsf{E}_{\mathrm{MB}}$. According to Theorem 3.7 of \cite{ding2021spiked}, on the event $\Omega_X$, for some small $0<\epsilon<c_{X,2},$ we have that 
	\begin{align*}
		\lambda_1^0-\widehat{\mathsf{E}}_{\mathrm{MB}}=\mathrm{O}(n^{-2/3+\epsilon}).
	\end{align*}
	In addition, according to Lemma \ref{lem_edgeconvergence1}, on the event $\Omega_{\Xi}$, 
	\begin{align*}
		\widehat{\mathsf{E}}_{\mathrm{MB}}-\mathsf{E}_{\mathrm{MB}}=\frac{\mathsf{C}_2}{\mathsf{C}_1}\mathcal{X}+\mathrm{O}(n^{-1+2c_{\Xi,1}}).
	\end{align*}
	Therefore, we obtain on $\Omega_X\cap\Omega_{\Xi}$,
	\begin{align*}
		|\mathtt{Y}_n-\mathtt{Z}_n|
		\le
		C\delta_n,
	\end{align*}
	where
	\begin{align*}
		\delta_n:=\frac{n^{-1/6+\epsilon}}{\sqrt{\tilde{\mathsf{v}}}}+\frac{1}{\sqrt{n\tilde{\mathsf{v}}}}=\mathrm{o}(1),
	\end{align*}
	because $\tilde{\mathsf v}\gg n^{-1/3+2\epsilon}$.
	Moreover, $\mathcal{X}$ depends only on the multipliers, and hence $\mathtt{Z}_n$ is $\sigma(\Xi)$-measurable. Since the multipliers are bounded in the theoretical construction and $\tilde{\mathsf v}\asymp \operatorname{Var}(\xi^2)$, the Berry--Esseen theorem yields
	\begin{align}\label{eq_asymGaussian_Z}
		\sup_{x\in\mathbb R}|\mathbb{P}(\mathtt{Z}_n\leq x)-\Phi(x)|=\mathrm{o}(1).
	\end{align}
	Let $\delta^{\prime}>0$ be arbitrary. A standard smoothing inequality for Kolmogorov distance gives
	\begin{align*}
		\sup_{x\in\mathbb R}\left|\mathbb P(\mathtt Y_n\le x\mid X)-\Phi(x)\right|
		\le
		\sup_{x\in\mathbb R}\left|\mathbb P(\mathtt Z_n\le x)-\Phi(x)\right|
		+
		\mathbb P(|\mathtt Y_n-\mathtt Z_n|>\delta^{\prime}\mid X)
		+C\delta^{\prime} .
	\end{align*}
	We now choose $\delta^{\prime}=2C\delta_n$. On the event $\Omega_X$, the bound on $\Omega_X\cap\Omega_{\Xi}$ implies
	\begin{align*}
		\mathbb P(|\mathtt Y_n-\mathtt Z_n|>\delta^{\prime}\mid X)
		\le\mathbb P(\Omega_{\Xi}^c\mid X)=
		\mathbb P(\Omega_{\Xi}^c)=\mathrm{o}(1),
	\end{align*}
	where we used the independence of $X$ and $\Xi$. Combining this estimate with \eqref{eq_asymGaussian_Z}, we obtain on $\Omega_X$ that
	\begin{align*}
		\sup_{x\in\mathbb R}\left|\mathbb P(\mathtt Y_n\le x\mid X)-\Phi(x)\right|=\mathrm{o}(1).
	\end{align*}
	Finally, Lemma \ref{lem_probabilitycontrol_X} gives $\mathbb P(\Omega_X^c)=\mathrm{o}(1)$, so the same bound holds with probability at least $1-\mathrm{O}(n^{-\mathsf c})$. This completes the proof.
\end{proof}

The proof of Theorem \ref{thm_conditionalclt_nonspike} extends directly to the top eigenvalues of $Q_{\mathrm{MB}}^0$. In particular, given a data set $\Sigma_0^{1/2}X$ and properly chosen multipliers, the bootstrapped sample covariance matrix $Q_{\mathrm{MB}}^0$ exhibits Gaussian limits for its leading edge eigenvalues.

We next transfer these limiting distributions from $Q_{\mathrm{MB}}^0$ to the largest non-spiked eigenvalues of $Q_{\mathrm{MB}}$. Recall
\begin{align*}
	Q_{\mathrm{MB}}=\Sigma^{1/2}X\Xi^2X^\top \Sigma^{1/2},
	\qquad
	Q_{\mathrm{MB}}^0=\Sigma_0^{1/2}X\Xi^2X^\top \Sigma_0^{1/2}.
\end{align*}

\begin{lemma}[Finite-rank comparison for non-outlier edge eigenvalues]
	\label{lem_finite_rank_edge_comparison}
	Suppose Assumption \ref{assu_modelassumption} holds. Then there exists a fixed integer $\mathtt{C}=\mathtt{C}(r)\le r$ such that
	for any fixed integer $\mathrm{K}\ge1$,
	\begin{align*}
		\lambda_{r+i+\mathtt{C}}(Q_{\mathrm{MB}}^0)\le\lambda_{r+i}(Q_{\mathrm{MB}})\le\lambda_{r+i-\mathtt{C}}(Q_{\mathrm{MB}}^0),\qquad 1\le i\le \mathrm{K},
	\end{align*}
	with the convention that invalid indices are omitted. Consequently, on the event $\Omega_X\cap\Omega_{\Xi}$, using the edge decomposition for $Q_{\mathrm{MB}}^0$, we have
	\begin{align*}
		\max_{1\le i\le \mathrm{K}}\left|\lambda_{r+i}-\mathsf{E}_{\mathrm{MB}}\right|=
		\mathrm{o}(\sqrt{\frac{\tilde{\mathsf{v}}}{n}}).
	\end{align*}
	Then, for $n$ sufficiently large, there exists some constant $\mathsf{c}_1>0$ such that with probability at least $1-\mathrm{O}(n^{-\mathsf{c}_1})$
	\begin{align}\label{eq_conditionalclt_nonspike_QMB}
		\sup_{x\in\mathbb R}\big|\mathbb{P}(\sqrt{n\tilde{\mathsf{v}}^{-1}}(\lambda_{r+i}-\mathsf{E}+\Delta_{\mathrm{edge}})\leq x|X)-\Phi(x)\big|=\mathrm{o}(1),
	\end{align}
	for any $1\leq i\leq \mathrm{K}$.
\end{lemma}
\begin{proof}
	For notational simplicity, write $\lambda_i^0=\lambda_i(Q_{\mathrm{MB}}^0)$ and $\lambda_i=\lambda_i(Q_{\mathrm{MB}})$ according to the Table \ref{table_notations} for the remainder of the proof. For any fixed $i$, we have the decomposition
	\begin{align*}
		\lambda_i^0-\mathsf{E}+\Delta_{\mathrm{edge}}=\lambda_i^0-\widehat{\mathsf{E}}_{\mathrm{MB}}+\widehat{\mathsf{E}}_{\mathrm{MB}}-\mathsf{E}_{\mathrm{MB}}.
	\end{align*}
	On the other hand, by Lemma \ref{lem_edgeconvergence1}, we have, on the event $\Omega_X\cap\Omega_{\Xi}$,
	\begin{align*}
		\widehat{\mathsf{E}}_{\mathrm{MB}}-\mathsf{E}_{\mathrm{MB}}=\frac{\mathsf{C}_1}{\mathsf{C}_2}\mathcal{X}+\mathrm{O}(n^{-1+2c_{\Xi,1}}),
	\end{align*}
	where $\mathsf{C}_1$ and $\mathsf{C}_2$ are deterministic quantities of order one. It follows that for any $1\leq i\leq \mathrm{K}$,
	\begin{align}
		\lambda_i^0&=\mathsf{E}_{\mathrm{MB}}+\lambda_i^0-\widehat{\mathsf{E}}_{\mathrm{MB}}+\frac{\mathsf{C}_1}{\mathsf{C}_2}\mathcal{X}+\mathrm{O}(n^{-1+2c_{\Xi,1}})\\
		&=\mathsf{E}_{\mathrm{MB}}+\frac{\mathsf{C}_1}{\mathsf{C}_2}\mathcal{X}+\mathrm{O}(n^{-2/3+c_{X,2}}),
	\end{align}
	where we used the estimate $|\lambda_i^0-\widehat{\mathsf{E}}_{\mathrm{MB}}|\lesssim n^{-2/3+c_{X,2}}$ from Theorem 3.7 of \citep{ding2021spiked} on the event $\Omega_X\cap\Omega_{\Xi}$. Hence, on $\Omega_X\cap\Omega_{\Xi}$, the eigenvalues $\lambda_i^0$, $1\leq i\le \mathrm{K}$, satisfy the rigidity estimate
	\begin{align}\label{eq_rigidity_lambda0}
		|\lambda_i^0-\lambda_{i+1}^0|\leq Cn^{-2/3+c_{X,2}},
	\end{align}
	for some positive constant $C>0$.
	
	It remains to compare $\lambda_{r+i}$ with $\lambda_i^0$ for $1\leq i\leq \mathrm{K}$. Since $\Sigma$ and $\Sigma_0$ differ only in the fixed $r$ spike directions under condition \eqref{eq_spikedmatrix}, we have
	\begin{align*}
		\operatorname{rank}(Q_{\mathrm{MB}}-Q_{\mathrm{MB}}^0)\le 2r.
	\end{align*}
	Hence the rank interlacing inequality implies that there exists a constant $\mathtt C\le 2r$ such that
	\begin{align*}
		\lambda_{i+\mathtt C}^0\leq\lambda_{r+i}\leq\lambda_{(i-\mathtt C)\vee 1}^0,\quad 1\leq i\leq \mathrm{K}.
	\end{align*}
	By \eqref{eq_rigidity_lambda0}, we have on the event $\Omega_X\cap\Omega_{\Xi}$ that for $1\leq i\leq \mathrm{K}$,
	\begin{align*}
		\lambda_{r+i}&=\lambda_i^0+\mathrm{O}(n^{-2/3+c_{X,2}})\\
		&=\mathsf{E}_{\mathrm{MB}}+\frac{\mathsf{C}_1}{\mathsf{C}_2}\mathcal{X}+\mathrm{O}(n^{-2/3+c_{X,2}}).
	\end{align*}
	According to Definition \ref{defn_feasiblemultiplier}, the feasible multipliers satisfy $n^{-1/3+2\delta_{*}}\ll \tilde{\mathsf{v}}\ll 1$ for some $\delta_{*}\ge c_{X,2}$. Therefore, $(\tilde{\mathsf{v}}n^{-1})^{1/2}\gg n^{-2/3+c_{X,2}}$. Then, on the event $\Omega_X\cap\Omega_{\Xi}$,
	\begin{align*}
		\sqrt{n/\tilde{\mathsf{v}}}(\lambda_{r+i}-\mathsf{E}_{\mathrm{MB}})=\sqrt{n/\tilde{\mathsf{v}}}\frac{\mathsf{C}_1}{\mathsf{C}_2}\mathcal{X}+\mathrm{o}(1),
	\end{align*}
	for $1\leq i\leq \mathrm{K}$. The proof is then completed by repeating the argument used for Theorem \ref{thm_conditionalclt_nonspike}.
\end{proof}
We have therefore established the results \eqref{eq_conditionalclt_nonspike_QMB}.

Now, recall that in \eqref{eq_control_edgebias}, we have concluded that $\Delta_{\mathrm{edge}}=\mathrm{o}(1)$ and $|m(\mathsf{E})-m_{2n,c}(\mathsf{E}_{\mathrm{MB}})|=\mathrm{o}(1)$ for $n^{-1/3}\ll\operatorname{Var}(\xi^2)\ll 1$. Then, by \eqref{eq_systemequation1} and \eqref{eq_systemequation2}, we have
\begin{align*}
	-\mathsf{E}_{\mathrm{MB}}\times m_{2n,c}(\mathsf{E}_{\mathrm{MB}})=\int\frac{t}{1+tm_{1n,c}(\mathsf{E}_{\mathrm{MB}})}\mathrm{d}F_{\xi^2}(t)=-\mathsf{E}\times m(\mathsf{E})+\mathrm{o}(1).
\end{align*}
Thus, by Assumption \ref{assum_technique}, we have
\begin{align*}
	|\mathfrak{C}_1-\mathsf{C}_1|=\mathrm{o}(1),\quad |\mathfrak{C}_2-\mathsf{C}_2|=\mathrm{o}(1).
\end{align*}
On the other hand, by Assumption \ref{assum_technique}, using Taylor expansion, \eqref{eq_systemequation1},  \eqref{eq_systemequation_edge_mp} and the approximation \eqref{eq_control_edgebias}, we have
\begin{align*}
	&\operatorname{Var}\Big(\frac{\xi^2}{1+\xi^2m_{1n,c}(\mathsf{E}_{\mathrm{MB}})}\Big)=\Big(\frac{1}{1+m_{1n,c}(\mathsf{E}_{\mathrm{MB}})}\Big)^2\times\operatorname{Var}\xi^2\times(1+\mathrm{o}(1))\\
	&=\Big(\frac{\mathsf{E}_{\mathrm{MB}}}{\mathsf{E}_{\mathrm{MB}}-\frac{1}{n}\sum_{i=1}^p\frac{\sigma_i}{1+\sigma_im_{2n,c}(\mathsf{E}_{\mathrm{MB}})}}\Big)^2\times\operatorname{Var}\xi^2\times(1+\mathrm{o}(1))\\
	&=\Big(\frac{\mathsf{E}}{\mathsf{E}-\frac{1}{n}\sum_{i=1}^p\frac{\sigma_i}{1+\sigma_im(\mathsf{E})}}\Big)^2\times\operatorname{Var}\xi^2\times(1+\mathrm{o}(1))=\mathsf{E}^2m^2(\mathsf{E})\operatorname{Var}(\xi^2)\times(1+\mathrm{o}(1)).
\end{align*}
It follows that 
\begin{align*}
	\tilde{\mathsf{v}}=\mathsf{v}(1+\mathrm{o}(1)).
\end{align*}
Therefore, Theorem \ref{thm_boostrapmain} can be directly concluded from \eqref{eq_conditionalclt_nonspike_QMB}. We complete the proof.

\subsection{Proof of Theorem \ref{thm_power}}\label{appendxi_secthreoem32}

In this subsection, we prove Theorem \ref{thm_power}. The proof consists of two parts. In the first part, we establish \eqref{eq_thmpower_one} by showing that the plug-in estimators $\Delta_{r_0}$ and $\mathfrak{s}$ introduce only negligible errors according to Theorem \ref{thm_boostrapmain}. In the second part, we establish \eqref{eq_thmpower_two} under the alternative. It is necessary to show that the confidence interval has a clear separation from $\mathsf{E}$ based on the decomposition \eqref{eq_spikepicture}.

Now, we consider the first statement \eqref{eq_thmpower_one} of Theorem \ref{thm_power}.
By the Gaussian approximation established in Theorem \ref{thm_boostrapmain}, we have
\begin{align}\label{eq_power_null_clt}
	\sup_{x\in\mathbb R}|\mathbb{P}\big(\sqrt{n\mathsf{v}^{-1}}(\lambda_{r_0}-\mathsf{E}+\Delta_{\mathrm{edge}})\le x|X)-\Phi(x)\big|=\mathrm{o}(1).
\end{align}

Recall that, conditional on $X$, the bootstrap eigenvalues
\[
\lambda_{r_0,1},\dots,\lambda_{r_0,\mathrm B},\lambda_{r_0,\mathrm{B}+1}
\]
are i.i.d., and the fresh draw $\lambda_{r_0,\mathrm{B}+1}$ is independent of the first $\mathrm B$ draws used to construct $\Delta_{r_0}$ and $\mathfrak s$.

Moreover, by the edge decomposition \eqref{eq_decomp_largesteigenvalue}, each $\lambda_{r_0,k}$ satisfies
\begin{align*}
	\lambda_{r_0,k}-\mathsf{E}_{\mathrm{MB}}=\lambda_{r_0,k}-\widehat{\mathsf{E}}_{\mathrm{MB}}+\widehat{\mathsf{E}}_{\mathrm{MB}}-\mathsf{E}_{\mathrm{MB}}\quad 1\le k\le \mathrm{B}+1.
\end{align*}
Therefore,
\begin{align*}
	\bar{\lambda}_{r_0}
	=
	\mathsf{E}_{\mathrm{MB}}+\mathrm{O}_{\mathbb P(\cdot|X)}\!\left(\sqrt{\frac{\mathsf{v}}{n\mathrm{B}}}+n^{-2/3+c_{X,2}}\right).
\end{align*}
Hence,
\begin{align}\label{eq_biasestimation}
	|\Delta_{r_0}-\Delta_{\mathrm{edge}}|
	=
	\mathrm{O}_{\mathbb P(\cdot|X)}\!\left( \sqrt{\frac{\mathsf{v}}{n\mathrm{B}}}+n^{-2/3+c_{X,2}}\right)
	=
	\mathrm{o}_{\mathbb P(\cdot|X)}\!\left(\sqrt{\frac{\mathsf v}{n}}\right),
\end{align}
where we used $\mathrm B\gtrsim n^{\mathsf c}$ and $\mathsf{v}\gg n^{-1/3+2c_{X,2}}$. 

Next we control the bootstrap sample variance. By \eqref{eq_power_null_clt}, the normalized variable
\begin{align*}
	\sqrt{n\mathsf v^{-1}}
	(\lambda_{r_0,1}-\mathsf E+\Delta_{\mathrm{edge}}) 
\end{align*}
converges conditionally to a standard Gaussian law on $\Omega_X$. We have
\begin{align*}
	\operatorname{Var}(\lambda_{r_0,1}\mid X)=\frac{\mathsf v}{n}(1+\mathrm o(1)).
\end{align*}
Since the bootstrap draws are conditionally i.i.d. and have uniformly bounded fourth moments after truncation of the multipliers, the conditional law of large numbers for the sample variance yields
\begin{align}\label{eq_power_null_variance}
	\mathfrak s^2=
	\frac{1}{\mathrm B-1}\sum_{k=1}^{\mathrm B}(\lambda_{r_0,k}-\bar\lambda_{r_0})^2=
	\frac{\mathsf v}{n}
	\big(1+\mathrm o_{\mathbb P(\cdot\mid X)}(1)\big).
\end{align}
Define
\begin{align*}
	Y_n:=\sqrt{n\mathsf v^{-1}}(\lambda_{r_0,\mathrm{B}+1}-\mathsf E+\Delta_{\mathrm{edge}}),\qquad
	Z_n:=\mathfrak s^{-1}(\lambda_{r_0,\mathrm{B}+1}-\mathsf E+\Delta_{r_0}). 
\end{align*}
Then
\begin{align*}
	Z_n=Y_n\cdot \frac{\sqrt{\mathsf v/n}}{\mathfrak s}+\frac{\Delta_{r_0}-\Delta_{\mathrm{edge}}}{\mathfrak s}.
\end{align*}
By \eqref{eq_biasestimation} and \eqref{eq_power_null_variance}, both multiplicative and additive errors are $\mathrm o_{\mathbb P(\cdot\mid X)}(1)$ on $\Omega_X$. Combining this with \eqref{eq_power_null_clt}, we obtain
\begin{align*}
	\sup_{x\in\mathbb R}
	\left|\mathbb P(Z_n\le x\mid X)-\Phi(x)
	\right|=\mathrm o(1)
\end{align*}
Equivalently, we have
\begin{align*}
	\sup_{x\in\mathbb R}\left|\mathbb P\big(\mathfrak{s}^{-1}(\lambda_{r_0,\mathrm{B}+1}-\mathsf{E}+\Delta_{r_0})\le x\mid X\big)-\Phi(x)\right|=\mathrm{o}(1),
\end{align*}
with probability at least $1-\mathrm{o}(1)$. It follows that
\begin{align*}
	\mathbb{P}\big(\mathsf{E}\in[\lambda_{r_0,\mathrm{B}+1}+\Delta_{r_0}-z_{\alpha/2}\mathfrak{s},\lambda_{r_0,\mathrm{B}+1}+\Delta_{r_0}+z_{\alpha/2}\mathfrak{s}]\mid X\big)=1-\alpha+\mathrm{o}(1).
\end{align*}
We complete the proof of \eqref{eq_thmpower_one} in Theorem \ref{thm_power}.

Now, we consider the second statement \eqref{eq_thmpower_two} of Theorem \ref{thm_power}.
Under $\mathbf{H}_a$, the eigenvalue $\lambda_{r_0}$ of $Q_{\mathrm{MB}}$ is associated with the supercritical spike $\widetilde{\sigma}_{r_0}$. The proof is based on showing that the center of the interval in \eqref{eq_CI} stays at distance to $\mathsf E$ of order much larger than the half-width of the constructed confidence interval. We first collect several standard estimates for weak outliers of $Q_{\mathrm{Sam}}$ in \eqref{eq_def_originalcovariance}. Define
\begin{align}\label{eq_def_vartheta_S}
	\vartheta_i^{\mathtt{S}}:=f(-\widetilde{\sigma}_i^{-1}),\quad i=1,\dots, r,
\end{align}
where $f$ is defined in \eqref{eq_mp}. Recall Assumption \ref{assu_modelassumption}, we define $\varsigma_i:=\widetilde{\sigma}_i+\mathsf{b}^{-1}$. The following results give the distance of $\vartheta_i^{\mathtt{S}}$ to $\mathsf{E}$ for $i=1,\dots, r$ under $\mathbf{H}_a$.
\begin{lemma}[Near-critical expansion of the outlier spikes]\label{lem_power_outlier_expansion}
	Suppose $f^{\prime\prime}(\mathsf{b})\ge c_0$ for some constant $c_0>0$, then
	\begin{align}\label{eq_power_outlier_expansion}
		\vartheta_i^{\mathtt{S}}-\mathsf{E}
		=\frac{1}{2} f^{\prime\prime}(\mathsf{b})\mathsf b^4\varsigma_i^2+\mathrm{O}(\varsigma_i^3).
	\end{align}
	Consequently, there exists a constant $c_1>0$ such that
	\begin{align}\label{eq_power_outlier_lower}
		\vartheta_i^{\mathtt{S}}-\mathsf{E}\ge c_1\varsigma_i^2
	\end{align}
	for all sufficiently large $n$ and $i=1,\dots, r$.
\end{lemma}
In the sequel, we mainly discuss the case $\varsigma_i=\mathrm{o}(1)$ for $1\le i\le r$ while the other cases can be handled similarly. The first $r$ spiked eigenvalues of $Q_{\mathrm{Sam}}$ can be estimated by the following results.
\begin{lemma}[Sample spiked eigenvalues]\label{lem_power_sample_localization}
	On the event $\Omega_X$, the sample spiked eigenvalue $\mu_i, 1\le i\le r$ satisfies
	\begin{align}\label{eq_power_sample_secular}
		1+\widetilde{\sigma}_i m(\mu_i)=\mathrm{O}\left(n^{-1/2+c_{X,4}}\varsigma_i^{-1/2}\right).
	\end{align}
	Consequently,
	\begin{align}\label{eq_power_sample_localization}
		\mu_i-\vartheta_i^{\mathtt{S}}
		=\mathrm{O}\left(n^{-1/2+c_{X,4}}\varsigma_i^{1/2}\right).
	\end{align}
\end{lemma}
The proofs of Lemma \ref{lem_power_outlier_expansion} and Lemma \ref{lem_power_sample_localization} are standard  in literature, see, for example, \citep{ding2021spiked,ding2024eigenvector}. We omit further details here. Lemma \ref{lem_power_outlier_expansion} is the starting point to estimate the distance of $\lambda_i$ to $\mathsf{E}$ while Lemma \ref{lem_power_sample_localization} gives a way to estimate $\Delta_{r_0}$ in \eqref{eq_CI}.

Now we move to the multiplier bootstrapped sample covariance matrix $Q_{\mathrm{MB}}$. We introduce two auxiliary quantities $\widehat{\vartheta}^{\mathtt{MB}}_i$ and $\vartheta^{\mathtt{MB}}_i, i=1,\dots, r$ of $Q_{\mathrm{MB}}$ which are defined as
\begin{align}\label{eq_def_vartheta_MB}
	1+\widetilde{\sigma}_im_{2n}(\widehat{\vartheta}^{\mathtt{MB}}_i)=0,\quad 1+\widetilde{\sigma}_im_{2n,c}(\vartheta^{\mathtt{MB}}_i)=0.
\end{align}
On the event $\Omega_{\Xi}$, Theorem 3.6 of \citet{DingYang2021} shows that $\widehat{\vartheta}^{\mathtt{MB}}_i$ characterize the locations of the outlier spiked eigenvalues of $Q_{\mathrm{MB}}$ with a fixed realization of the multipliers, while $\vartheta^{\mathtt{MB}}_i$ is the limit of $\widehat{\vartheta}^{\mathtt{MB}}_i$ for all feasible multipliers in Definition \ref{defn_feasiblemultiplier}. The results in the sequel quantify the discrepancy between $\vartheta^{\mathtt{MB}}_i$  and $\vartheta_i^{\mathtt{S}}$.
\begin{lemma}[Gap between $\vartheta_i^{\mathtt{S}}$ and $\vartheta^{\mathtt{MB}}_i$]\label{lem_power_bootstrap_gap}
	Let $\vartheta_i^{\mathtt{S}}$ and $\vartheta^{\mathtt{MB}}_i$ be defined as in \eqref{eq_def_vartheta_S} and \eqref{eq_def_vartheta_MB}. For multipliers satisfying Definition \ref{defn_feasiblemultiplier}, it holds that for $i=1,\dots, r$
	\begin{align}\label{eq_power_gap}
		\vartheta^{\mathtt{MB}}_i-\vartheta_i^{\mathtt{S}}
		=
		\mathrm{O}\big(\operatorname{Var}(\xi^2)\big),
	\end{align}
	Consequently,
	\begin{align}\label{eq_power_theta_c_gap}
		\vartheta^{\mathtt{MB}}_i-\mathsf E
		=\frac{1}{2}f^{\prime\prime}(\mathsf{b})\mathsf{b}^4\varsigma_i^2+
		\mathrm{O}\big(\varsigma_i^3+\operatorname{Var}(\xi^2)\big).
	\end{align}
\end{lemma}
The proof of this lemma will be postponed to Section \ref{sec_proof_lemma_bootstrapgap}.

Secondly, we use the following decomposition of the outlier spiked eigenvalues $\lambda_i,1\le i\le r$ of $Q_{\mathrm{MB}}$ under $\mathbf{H}_a$,
\begin{align}\label{eq_power_decomp}
	\lambda_{i}=\lambda_i-\widehat{\vartheta}^{\mathtt{MB}}_i+\widehat{\vartheta}^{\mathtt{MB}}_i-\vartheta^{\mathtt{MB}}_i+\vartheta^{\mathtt{MB}}_i.
\end{align}
The term $\lambda_i-\widehat{\vartheta}_i^{\mathtt{MB}}$ is controlled by outlier localization, whereas $\widehat{\vartheta}^{\mathtt{MB}}_i-\vartheta^{\mathtt{MB}}_i$ is generated by the multiplier fluctuation. We first analyze $\lambda_i-\widehat{\vartheta}^{\mathtt{MB}}_i$ for $i=1,\dots, r$. On the event $\Omega_{\Xi}$, by the definition $\mathsf{E}=f(\mathsf{b})$ and the fact $|\mathsf{E}-\mathsf{E}_{\mathrm{MB}}|=\mathrm{O}(\operatorname{Var}(\xi^2))$, we have  
\begin{align*}
	1-\mathsf{b}^{-1}m(\mathsf{E}_{\mathrm{MB}})+\mathrm{O}(\operatorname{Var}(\xi^2))=0.
\end{align*}
On the other hand, we find that from Lemma \ref{lem_edgeconvergence1} that on the event $\Omega_{\Xi}$, for a fixed realization of multipliers,
\begin{align*}
	1-\mathsf{b}^{-1}m(\widehat{\mathsf{E}}_{\mathrm{MB}})+\mathrm{O}(\operatorname{Var}(\xi^2))=0,
\end{align*}
since $n^{-1/2+c_{\Xi,1}}\ll \operatorname{Var}(\xi^2)$ by Definition \ref{defn_feasiblemultiplier}. In addition, one can check from Definition \ref{def_OmegaM} that $|m_{2n}(\widehat{\mathsf{E}}_{\mathrm{MB}})-m(\widehat{\mathsf{E}}_{\mathrm{MB}})|=\mathrm{O}(\log^{c_{\Xi,3}}n\operatorname{Var}(\xi^2))$ on the event $\Omega_{\Xi}$. It follows that 
\begin{align*}
	1-\mathsf{b}^{-1}m_{2n}(\widehat{\mathsf{E}}_{\mathrm{MB}})+\mathrm{O}(\log^{c_{\Xi,3}}n\operatorname{Var}(\xi^2))=0.
\end{align*}
By definition of $\varsigma_i$, we have for $1\le i \le r$,
\begin{align*}
	1+\widetilde{\sigma}_im_{2n}(\widehat{\mathsf{E}}_{\mathrm{MB}})+\mathrm{O}(\log^{c_{\Xi,3}}n\operatorname{Var}(\xi^2))= \varsigma_i m_{2n}(\widehat{\mathsf{E}}_{\mathrm{MB}}).
\end{align*}
It is equivalent to
\begin{align}\label{eq_spike_condition}
	\widetilde{\sigma}_i=-m^{-1}_{2n}(\widehat{\mathsf{E}}_{\mathrm{MB}})+\varsigma_i+\mathrm{O}(\log^{c_{\Xi,3}}n\operatorname{Var}(\xi^2)).
\end{align}
Then, on the event $\Omega_{\Xi}$, as long as $\varsigma_i\gg \log^{c_{\Xi,3}}n\operatorname{Var}(\xi^2)$, the Assumption 3.2 of \cite{DingYang2021} holds as
\begin{align*}
	\widetilde{\sigma}_i+m^{-1}_{2n}(\widehat{\mathsf{E}}_{\mathrm{MB}})=\varsigma_i(1+\mathrm{o}(1)).
\end{align*}
Then, on the event $\Omega_X\cap\Omega_{\Xi}$, by Definition \ref{def_OmegaX}, we have
\begin{align}\label{eq_power_outlier_est1}
	\lambda_i-\widehat{\vartheta}_i^{\mathtt{MB}}=\mathrm{O}(n^{-1/2+c_{X,5}}\varsigma_i^{1/2}).
\end{align}

Next, we consider $\widehat{\vartheta}_i^{\mathtt{MB}}-\vartheta_i^{\mathtt{MB}}$. Due to the perturbation of multipliers, $\widehat{\vartheta}_i^{\mathtt{MB}}$ fluctuates around $\vartheta_i^{\mathtt{MB}}$ by the CLT mechanism. We have the following results.
\begin{lemma}[Outlier location fluctuation]\label{lem_power_bootstrap_localization}
	On the event $\Omega_{\Xi}$, for each fixed $1\le i\le r$,
	\begin{align}\label{eq_power_bootstrap_localization}
		\widehat{\vartheta}_i^{\mathtt{MB}}-\vartheta_i^{\mathtt{MB}}=\mathrm{O}\left(n^{-1/2+c_{\Xi,1}}\sqrt{\operatorname{Var}(\xi^2)}\right).
	\end{align}
\end{lemma}
The proof of Lemma \ref{lem_power_bootstrap_localization} will be postponed to Section \ref{sec_proof_lemma_bootstraplocation}. 

Now, with \eqref{eq_power_decomp}, \eqref{eq_power_outlier_est1}, and Lemma \ref{lem_power_bootstrap_localization} in hand, fix $1\le i\le r$ and set
\[
\eta_{n,i}
:=
n^{-1/2+c_{X,5}}\varsigma_i^{1/2}
+
n^{-1/2+c_{\Xi,1}}\sqrt{\operatorname{Var}(\xi^2)}.
\]
On the event $\Omega_X$, define
\[
\mathcal B_{n,i}
:=
\left\{
\left|\lambda_i-\widehat{\vartheta}_i^{\mathtt{MB}}\right|
\le C_1 n^{-1/2+c_{X,5}}\varsigma_i^{1/2},
\quad
\left|\widehat{\vartheta}_i^{\mathtt{MB}}-\vartheta_i^{\mathtt{MB}}\right|
\le C_2 n^{-1/2+c_{\Xi,1}}\sqrt{\operatorname{Var}(\xi^2)}
\right\},
\]
for sufficiently large constants $C_1,C_2>0$. By \eqref{eq_power_outlier_est1} and Lemma \ref{lem_power_bootstrap_localization}, we know that the event $\mathcal B_{n,i}$ is implied by $\Omega_{\Xi}$ once $\Omega_X$ is fixed. Therefore, by independence of $X$ and $\Xi$,
\[
\mathbb P(\mathcal B_{n,i}^c\mid X)
\le
\mathbb P(\Omega_{\Xi}^c\mid X)
=
\mathbb P(\Omega_{\Xi}^c)
=
\mathrm o(1)
\]
on the event $\Omega_X$. On $\mathcal B_{n,i}$, the decomposition \eqref{eq_power_decomp} implies
\[
|\lambda_i-\vartheta_i^{\mathtt{MB}}|
\le C\eta_{n,i},
\]
and hence, on $\Omega_X$,
\begin{align}\label{eq_power_outlier_est21}
	\lambda_i-\vartheta_i^{\mathtt{MB}}
	=
	\mathrm O_{\mathbb P(\cdot\mid X)}(\eta_{n,i}).
\end{align}

Applying \eqref{eq_power_outlier_est21} with $i=r_0$ and using conditional i.i.d. of the bootstrap draws, we obtain
\begin{align}\label{eq_power_outlier_est2}
	\lambda_{r_0}-\vartheta_{r_0}^{\mathtt{MB}}=\mathrm{O}_{\mathbb{P}(\cdot|X)}(\eta_{n,r_0}),
\end{align}
and
\begin{align}\label{eq_power_width}
	\mathfrak{s}^2=\frac{1}{\mathrm{B}-1}\sum_{k=1}^{\mathrm{B}}(\lambda_{r_0,k}-\bar{\lambda}_{r_0})^2=\mathrm{O}_{\mathbb{P}(\cdot|X)}(\eta^2_{n,r_0}).
\end{align}
In particular,
\begin{align*}
	\mathfrak s
	=
	\mathrm O_{\mathbb P(\cdot\mid X)}(\eta_{n,r_0}).
\end{align*}

Next, by Lemmas \ref{lem_power_bootstrap_gap} and \ref{lem_power_sample_localization}, and \eqref{eq_power_outlier_est2}, we have on the event $\Omega_X$
\begin{align*}
	\Delta_{r_0}&=\mu_{r_0}-\bar{\lambda}_{r_0}=\vartheta_{r_0}^{\mathtt{S}}-\vartheta_{r_0}^{\mathtt{MB}}+(\mu_{r_0}-\vartheta_{r_0}^{\mathtt{S}})+(\vartheta_{r_0}^{\mathtt{MB}}-\bar{\lambda}_{r_0})\\
	&=\mathrm{O}(\varsigma_{r_0}^3+\operatorname{Var}(\xi^2))+\mathrm{O}(n^{-1/2+c_X,4}\varsigma_{r_0}^{1/2})+\mathrm{O}_{\mathbb{P}(\cdot|X)}(\eta_{n,r_0}).
\end{align*}
Under \eqref{eq_spiketruedefinition} and Definition \ref{defn_feasiblemultiplier}, all three terms above are $\mathrm o(\varsigma_{r_0}^2)$, so that
\begin{align}\label{eq_power_delta_small}
	\Delta_{r_0}=
	\mathrm o_{\mathbb P(\cdot\mid X)}(\varsigma_{r_0}^2)
\end{align}
on $\Omega_X$.

By Lemma \ref{lem_power_bootstrap_gap}, again on $\Omega_X$,
\begin{align*}
	\vartheta_{r_0}^{\mathtt{MB}}-\mathsf E=
	\frac{1}{2} f^{\prime\prime}(\mathsf b)\mathsf b^4\varsigma_{r_0}^2+
	\mathrm O(\varsigma_{r_0}^3+\operatorname{Var}(\xi^2)).
\end{align*}
Since $\varsigma_{r_0}=\mathrm{o}(1)$ and $\varsigma_{r_0}^2\gg \operatorname{Var}(\xi^2)$ by Assumption \ref{assu_modelassumption}, there exists a constant $c_0>0$ such that, for all sufficiently large $n$,
\begin{align}\label{eq_power_center_gap}
	|\vartheta_{r_0}^{\mathtt{MB}}-\mathsf E|
	\ge c_0\varsigma_{r_0}^2.
\end{align}

For a fresh draw $\lambda_{r_0,\mathrm{B}+1}$, it follows that 
\begin{align*}
	&|\mathsf{E}-(\lambda_{r_0,\mathrm{B}+1}+\Delta_{r_0})|=|\mathsf{E}-\vartheta_{r_0}^{\mathtt{MB}}+\vartheta_{r_0}^{\mathtt{MB}}-\lambda_{r_0,\mathrm{B}+1}-\Delta_{r_0}|\\
	&\ge |\vartheta_{r_0}^{\mathtt{MB}}-\mathsf{E}|
	-|\lambda_{r_0,\mathrm{B}+1}-\vartheta_{r_0}^{\mathtt{MB}}|
	-|\Delta_{r_0}|\\
	&\ge c_0\varsigma_{r_0}^2-
	\mathrm{o}_{\mathbb P(\cdot\mid X)}(\varsigma_{r_0}^2)=
	c_0\varsigma_{r_0}^2(1+\mathrm o_{\mathbb P(\cdot\mid X)}(1)).
\end{align*}
for some positive constant $c_0>0$. Similarly, \eqref{eq_power_width} gives $z_{\alpha/2}\mathfrak{s}=\mathrm{o}(\varsigma_{r_0}^2).$
Therefore the distance from $\mathsf E$ to the interval center dominates the interval half-width, and
\[
\mathbb P\big(\mathsf E\in[\widehat{\mathsf E}^-,\widehat{\mathsf E}^+]\mid X\big)=\mathrm{o}(1).
\]
This completes the proof of Theorem \ref{thm_power}.

\subsection{Proof of Corollary \ref{cor_spikeestimation}}

Let
\[
T_{r_0}
:=
\lambda_{r_0,\mathrm{B}+1}+\Delta_{r_0}+z_{\alpha/2}\mathfrak{s}
\]
be the threshold in \eqref{eq_estimator}. Recall $\{\mu_i\}$ are the eigenvalues of $Q_{\mathrm{Sam}}$ as in Table \ref{table_notations}. Since the sample eigenvalues are ordered decreasingly, we have
\[
\{\widehat r=r\}
=
\{\mu_r>T_{r_0}\}\cap\{\mu_{r+1}\le T_{r_0}\}.
\]
Thus it suffices to show
\[
\mathbb P(\mu_r>T_{r_0}\mid X)=1-\mathrm o(1),
\qquad
\mathbb P(\mu_{r+1}\le T_{r_0}\mid X)=1-\alpha/2+\mathrm o(1),
\]
on an event of probability at least $1-\mathrm o(1)$.

We first consider the non-spiked eigenvalue $\mu_{r+1}$. Since $r_0$ is fixed and $\mu_{r+1},\dots,\mu_{r_0}$ are the first finitely many non-spiked eigenvalues, Definition \ref{def_OmegaX} yields on $\Omega_X$ that
\[
\mu_{r+1}-\mathsf E=\mathrm O(n^{-2/3+c_{X,2}}).
\]
By Theorem \ref{thm_power}, the bootstrap fluctuation scale satisfies
\[
\mathfrak s^2=\frac{\mathsf v}{n}(1+\mathrm o_{\mathbb P(\cdot\mid X)}(1)),
\]
and hence, since $\mathsf v\gg n^{-1/3+2c_{X,2}}$,
\begin{align}\label{eq_cor_nonspike_small}
	\mu_{r+1}-\mathsf E=
	\mathrm o_{\mathbb P(\cdot\mid X)}(\mathfrak s)
\end{align}
on $\Omega_X$.
On the other hand, the null part of Theorem \ref{thm_power} gives
\begin{align}\label{eq_cor_studentized}
	\sup_{x\in\mathbb R}
	\left|
	\mathbb P\left(
	\mathfrak s^{-1}
	(\lambda_{r_0,\mathrm{B}+1}-\mathsf E+\Delta_{r_0})
	\le x
	\,\middle|\,X
	\right)-\Phi(x)\right|
	=
	\mathrm o(1).
\end{align}
Therefore,
\begin{align*}
	\mathbb P(\mu_{r+1}\le T_{r_0}\mid X)
	&=
	\mathbb P\left(
	\mathfrak s^{-1}
	(\lambda_{r_0,\mathrm{B}+1}-\mathsf E+\Delta_{r_0})
	\ge
	\mathfrak s^{-1}(\mu_{r+1}-\mathsf E)-z_{\alpha/2}
	\,\middle|\,X
	\right)\\
	&=
	1-\Phi\!\left(
	\mathfrak s^{-1}(\mu_{r+1}-\mathsf E)-z_{\alpha/2}
	\right)
	+
	\mathrm o(1).
\end{align*}
By \eqref{eq_cor_nonspike_small}, the argument of $\Phi$ equals $-z_{\alpha/2}+\mathrm o_{\mathbb P(\cdot\mid X)}(1)$. Hence,
\[
\mathbb P(\mu_{r+1}\le T_{r_0}\mid X)
=
1-\alpha/2+\mathrm o(1).
\]

We next consider the smallest spiked eigenvalue $\mu_r$. By Lemma \ref{lem_power_outlier_expansion} and Lemma \ref{lem_power_sample_localization}, on $\Omega_X$, for sufficiently large $n$,
\[
\mu_r-\mathsf E
=
(\vartheta_r^{\mathtt S}-\mathsf E)+(\mu_r-\vartheta_r^{\mathtt S})
\ge c_0\varsigma_r^2,
\]
for some constant $c_0>0$, where we recall that $\varsigma_r=\widetilde\sigma_r+\mathsf b^{-1}.$

On the other hand, by the null part of Theorem \ref{thm_power},
\[
T_{r_0}-\mathsf E
=
\mathrm O_{\mathbb P(\cdot\mid X)}(\mathfrak s)
=
\mathrm O_{\mathbb P(\cdot\mid X)}\!\left(\sqrt{\frac{\mathsf v}{n}}\right).
\]
Since Assumption \ref{assu_modelassumption}(iii) gives $\varsigma_r\gtrsim n^{-1/6+\kappa}$ with $\kappa>\delta/2$, we have
\[
\varsigma_r^2\gg \sqrt{\frac{\mathsf v}{n}}.
\]
Consequently,
\[
\mu_r-T_{r_0}
=
(\mu_r-\mathsf E)-(T_{r_0}-\mathsf E)
\ge
c_0\varsigma_r^2-\mathrm O_{\mathbb P(\cdot\mid X)}\!\left(\sqrt{\frac{\mathsf v}{n}}\right)
>
0
\]
with probability tending to one. Therefore,
\[
\mathbb P(\mu_r>T_{r_0}\mid X)=1-\mathrm o(1).
\]

Combining the above results, we obtain
\begin{align*}
	\mathbb P(\widehat r=r\mid X)
	&=
	\mathbb P(\mu_r>T_{r_0},\ \mu_{r+1}\le T_{r_0}\mid X)\\
	&=
	1-\alpha/2+\mathrm o(1).
\end{align*}
This proves the corollary.

\section{Proofs of some auxiliary lemmas}

\subsection{Proof of Lemma \ref{lem_edgeconvergence1}}\label{sec_prf_lem_edgeconvergence1}
In this subsection, we prove Lemma \ref{lem_edgeconvergence1}. We begin with several auxiliary results that will be used throughout the proof.

\begin{lemma}[Multiplier-side regulation]\label{lem_regulation}
	Let $x_0:=m_{1n,c}(\mathsf{E}_{\mathrm{MB}})$. Suppose Assumption \ref{assum_technique} holds. In particular, assume that there exists a constant $\tau_0>0$ such that
	\begin{align*}
		\inf_{t\in \operatorname{supp}(F_{\xi^2})}
		|1+t x_0|\geq \tau_0 .
	\end{align*}
	Assume further that the multiplier distribution has bounded support, namely
	\begin{align*}
		\operatorname{supp}(F_{\xi^2})\subset [0,\mathsf{C}_\xi]
	\end{align*}
	for some constant $\mathsf{C}_\xi$. Then there exist constants $\delta_0>0$ and $\tau>0$ such that
	\begin{align*}
		\inf_{|x-x_0|\leq \delta_0}
		\inf_{t\in \operatorname{supp}(F_{\xi^2})}
		|1+tx|\geq \tau .
	\end{align*}
	Consequently, for any realization $\{\xi_j^2\}_{j=1}^n$ drawn from
	$F_{\xi^2}$,
	\begin{align*}
		\inf_{|x-x_0|\le \delta_0} \min_{1\leq j\leq n} |1+\xi_j^2x|\geq \tau.
	\end{align*}
	In particular, if
	\begin{align*}
		|m_{1n}(\widehat{\mathsf{E}}_{\mathrm{MB}})-x_0|=\mathrm{o}(1),
	\end{align*}
	then, for all sufficiently large $n$,
	\begin{align*}
		\min_{1\le j\le n}|1+\xi_j^2m_{1n}(\widehat{\mathsf{E}}_{\mathrm{MB}})|\geq \tau .
	\end{align*}
\end{lemma}
\begin{proof}
	By Assumption \ref{assum_technique},
	\begin{align*}
		\inf_{t\in \operatorname{supp}(F_{\xi^2})}|1+t x_0|\geq \tau_0.
	\end{align*}
	Choose $0<\delta_0\leq \tau_0(2\mathsf{C}_\xi)^{-1}.$ Then, for any $|x-x_0|\le \delta_0$ and any $t\in \operatorname{supp}(F_{\xi^2})$,
	\begin{align*}
		|1+tx|&\ge |1+t x_0|-t|x-x_0|  \\
		&\ge \tau_0-\mathtt{C}_\xi\delta_0\ge \frac{\tau_0}{2}.
	\end{align*}
	Hence the first claim holds with $\tau=\tau_0/2$. Since each sampled $\xi_j^2$ belongs to $\operatorname{supp}(F_{\xi^2})$, the sample-level bound follows immediately. Finally, if
	\begin{align*}
		|m_{1n}(\widehat{\mathsf{E}}_{\mathrm{MB}})-x_0|=\mathrm{o}(1),
	\end{align*}
	then for all sufficiently large $n$,
	$|m_{1n}(\widehat{\mathsf{E}}_{\mathrm{MB}})-x_0|\le \delta_0.$ Substituting $x=m_{1n}(\widehat{\mathsf{E}}_{\mathrm{MB}})$ in the neighborhood bound gives the desired conclusion.
\end{proof}

Recall Definition \ref{def_OmegaM}.
\begin{lemma}[Uniform concentration of multiplier transforms]
	\label{lem_uniform_multiplier_concentration}
	Under the assumptions of Lemma \ref{lem_regulation}, for any  small constant
	$c_{\Xi,1}>0$, with probability $1-\mathrm{o}(1)$, the following estimates hold:
	\begin{align*}
		\sup_{|x-x_0|\le\delta_0}\left|\frac{1}{n}\sum_{j=1}^n \frac{\xi_j^2}{1+\xi_j^2x}-\int\frac{t}{1+tx}\,\mathrm{d}F_{\xi^2}(t) \right|\leq n^{-1/2+c_{\Xi,1}},
	\end{align*}
	\begin{align*}
		\sup_{|x-x_0|\le\delta_0}\left|\frac{1}{n}\sum_{j=1}^n\frac{\xi_j^4}{(1+\xi_j^2x)^2}-\int\frac{t^2}{(1+tx)^2}\,\mathrm{d}F_{\xi^2}(t)\right|\leq n^{-1/2+c_{\Xi,1}},
	\end{align*}
	and
	\begin{align*}
		\sup_{|x-x_0|\le\delta_0}\left|\frac{1}{n}\sum_{j=1}^n \frac{\xi_j^6}{(1+\xi_j^2x)^3}-\int\frac{t^3}{(1+tx)^3}\,\mathrm{d}F_{\xi^2}(t)\right|\leq n^{-1/2+c_{\Xi,1}}.
	\end{align*}
\end{lemma}

\begin{proof}
	We prove the first estimate; the other two follow by the same argument. Define
	\begin{align*}
		h(t,x):=\frac{t}{1+tx}.
	\end{align*}
	By Lemma \ref{lem_regulation}, for all $|x-x_0|\le \delta_0$ and all
	$t\in \operatorname{supp}(F_{\xi^2})$, we have $|1+tx|\ge \tau.$ Since $t\in[0,\mathsf{C}_\xi]$, both $h(t,x)$ and $\partial_xh(t,x)$ are uniformly bounded on this set. Hence $h(t,x)$ is uniformly Lipschitz in $x$.
	
	Let $\mathcal G_n\subset[x_0-\delta_0,x_0+\delta_0]$ be a grid with mesh size $n^{-2}$. Then $|\mathcal G_n|=\mathrm{O}(n^2)$. For each fixed
	$x\in\mathcal G_n$, the random variables $h(\xi_j^2,x)-\mathbb E h(\xi^2,x)$ are independent, centered, and uniformly bounded. By Hoeffding's inequality,
	\begin{align*}
		\mathbb P\left(\left|\frac1n\sum_{j=1}^n h(\xi_j^2,x)- \mathbb E h(\xi^2,x)\right|>n^{-1/2+c_{\Xi,1}}\right)\leq2\exp(-c n^{2c_{\Xi,1}})
	\end{align*}
	for some constant $c>0$. Taking a union bound over
	$x\in\mathcal G_n$, we get
	\begin{align*}
		\max_{x\in\mathcal G_n}\left|\frac{1}{n}\sum_{j=1}^n h(\xi_j^2,x)-\mathbb E h(\xi^2,x)\right|\leq n^{-1/2+c_{\Xi,1}}
	\end{align*}
	with probability $1-\mathrm{o}(1)$.
	
	For any $x\in[x_0-\delta_0,x_0+\delta_0]$, choose $x'\in\mathcal G_n$ such that $|x-x'|\le n^{-2}.$ By the uniform Lipschitz property,
	\begin{align*}
		\left|\frac{1}{n}\sum_{j=1}^n h(\xi_j^2,x)-\frac{1}{n}\sum_{j=1}^n h(\xi_j^2,x')\right|\leq Cn^{-2},
	\end{align*}
	and similarly,
	\begin{align*}
		\left|\mathbb E h(\xi^2,x)-\mathbb E h(\xi^2,x')\right|\le Cn^{-2}.
	\end{align*}
	This proves the first estimate. The same argument applies to
	\begin{align*}
		h(t,x)=\frac{t^2}{(1+tx)^2}
		\quad\text{and}\quad
		h(t,x)=\frac{t^3}{(1+tx)^3},
	\end{align*}
	because these functions and their $x$-derivatives are uniformly bounded under the same denominator separation. This proves the claim.
\end{proof}
\begin{remark}
	For the Chi-squared multiplier in Example \ref{example_one}, since $\chi_{\mathsf{N}}^2/\mathsf{N}$ has unbounded support, we use a truncated and recentered version in the theoretical analysis. Let
	\begin{align*}
		\zeta=\chi_{\mathsf{N}}^2/\mathsf{N},\qquad a_n=C\sqrt{\frac{\log n}{\mathsf{N}}},
	\end{align*}
	and define
	\begin{align*}
		\xi_{\operatorname{tr}}^2:=
		1+\zeta\mathbf 1_{\{|\zeta-1|\le a_n\}}-
		\mathbb E\left[\zeta\mathbf 1_{\{|\zeta-1|\le a_n\}}\right].
	\end{align*}
	If $\log n\ll \mathsf{N}\ll n^{1/3-2\delta_{*}}$, then
	\begin{align*}
		\mathbb E\xi_{\operatorname{tr}}^2=1,\qquad
		\operatorname{Var}(\xi_{\operatorname{tr}}^2)=\frac{2}{\mathsf{N}}(1+\mathrm{o}(1)),
	\end{align*}
	and
	\begin{align*}
		\operatorname{supp}(F_{\xi_{\operatorname{tr}}^2})
		\subset [1-\mathrm{o}(1),1+\mathrm{o}(1)].
	\end{align*}
	Hence, under $|1+m_{1n,c}(\mathsf{E}_{\mathrm{MB}})|\ge \tau_0$, the multiplier-side regularity condition in Assumption \ref{assum_technique} holds.
\end{remark}

We now turn to the main proof. Recall the definitions of the bivariate functions $F_n(x,y)$ and $F_{n,c}(x,y)$ in \eqref{eq_def_systemequation}. Also, recall the systems equations in \eqref{eq_systemequations_edge} of the pairs $(m_{1n,c}(\mathsf{E}_{\mathrm{MB}}),\mathsf{E}_{\mathrm{MB}})$ and $(m_{1n}(\widehat{\mathsf{E}}_{\mathrm{MB}}),\widehat{\mathsf{E}}_{\mathrm{MB}})$. To proceed with the proof, we need the following definition.
\begin{definition}[Regular-edge condition]\label{def_regularedge}
	Let $x_0:=m_{1n,c}(\mathsf E_{\mathrm{MB}})$ and $y_0:=\mathsf E_{\mathrm{MB}}$. We say that the edge system is regular if there exists a constant $c_0>0$ such that
	\begin{align}\label{eq_regular_edge_condition}
		|\partial_yF_{n,c}(x_0,y_0)|\ge c_0,
		\qquad
		|\partial_{xx}F_{n,c}(x_0,y_0)|\ge c_0.
	\end{align}
\end{definition}
We first record the following stability result.
\begin{lemma}[Stability of the edge system]
	\label{lem_stabilityargument}
	Let $x_0:=m_{1n,c}(\mathsf E_{\mathrm{MB}})$ and $y_0:=\mathsf E_{\mathrm{MB}}$. Suppose that the system is regular in the sense that Definition \ref{def_regularedge} holds.
	Then there exists a unique local pair $(x_1,y_1)$ such that
	\begin{align*}
		F_n(x_1,y_1)=0,
		\qquad
		\partial_xF_n(x_1,y_1)=0
	\end{align*}
	and
	\begin{align*}
		|x_1-x_0|+|y_1-y_0|=\mathrm{O}(n^{-1/2+c_{\Xi,1}}).
	\end{align*}
\end{lemma}

Given Lemma \ref{lem_stabilityargument}, the first statement of Lemma \ref{lem_edgeconvergence1} follows immediately. We therefore proceed to prove Lemma \ref{lem_stabilityargument}.

\begin{proof}
	Fix a realization of multipliers in $\Omega_{\Xi}$. The following argument is deterministic.
	Define the two-dimensional maps
	
	\begin{align*}
		G_n(x,y):=
		\begin{pmatrix}
			F_n(x,y)\\
			\partial_xF_n(x,y)
		\end{pmatrix},
		\qquad
		G_{n,c}(x,y):=
		\begin{pmatrix}
			F_{n,c}(x,y)\\
			\partial_xF_{n,c}(x,y)
		\end{pmatrix}.
	\end{align*}
	By the definition of $(x_0,y_0)$ in Lemma \ref{lem_stabilityargument}, we have $G_{n,c}(x_0,y_0)=0.$ Moreover,
	\begin{align*}
		DG_{n,c}(x_0,y_0)
		=
		\begin{pmatrix}
			\partial_xF_{n,c}(x_0,y_0)
			&
			\partial_yF_{n,c}(x_0,y_0)\\
			\partial_{xx}F_{n,c}(x_0,y_0)
			&
			\partial_{xy}F_{n,c}(x_0,y_0)
		\end{pmatrix}.
	\end{align*}
	Since $\partial_xF_{n,c}(x_0,y_0)=0,$
	we have
	\begin{align*}
		\det DG_{n,c}(x_0,y_0)
		=
		-\partial_yF_{n,c}(x_0,y_0)
		\partial_{xx}F_{n,c}(x_0,y_0).
	\end{align*}
	Therefore, by the non-degeneracy assumption,
	\begin{align*}
		|\det DG_{n,c}(x_0,y_0)|\ge c_0^2.
	\end{align*}
	Thus $DG_{n,c}(x_0,y_0)$ is invertible and its inverse is uniformly bounded.
	
	Let $\mathcal N_{\delta_0}:=\{(x,y): |x-x_0|+|y-y_0|\le \delta_0\}$ for some sufficiently small constant $\delta_0>0$. By Definition \ref{def_OmegaM} and Lemma \ref{lem_uniform_multiplier_concentration}, on $\Omega_{\Xi}$,
	\begin{align*}
		\sup_{(x,y)\in\mathcal N_{\delta_0}}
		\|G_n(x,y)-G_{n,c}(x,y)\|=\mathrm{O}(n^{-1/2+c_{\Xi,1}}),
	\end{align*} 
	and
	\begin{align*}
		\sup_{(x,y)\in\mathcal N_{\delta_0}}
		\|DG_n(x,y)-DG_{n,c}(x,y)\|=\mathrm{O}(n^{-1/2+c_{\Xi,1}}).
	\end{align*}
	In particular,
	\begin{align*}
		G_n(x_0,y_0)=G_n(x_0,y_0)-G_{n,c}(x_0,y_0)=\mathrm{O}(n^{-1/2+c_{\Xi,1}}),
	\end{align*}
	and
	\begin{align*}
		DG_n(x_0,y_0)= DG_{n,c}(x_0,y_0)+\mathrm{O}(n^{-1/2+c_{\Xi,1}}).
	\end{align*}
	Hence, for all sufficiently large $n$, $DG_n(x_0,y_0)$ is also invertible and
	\begin{align*}
		\|[DG_n(x_0,y_0)]^{-1}\|\le C
	\end{align*}
	for some constant $C>0$. We now solve $G_n(x,y)=0$ near $(x_0,y_0)$. Write
	\begin{align*}
		d:=
		\begin{pmatrix}
			x-x_0\\
			y-y_0
		\end{pmatrix},
		\qquad
		d_0:=
		-[DG_n(x_0,y_0)]^{-1}G_n(x_0,y_0).
	\end{align*}
	Then $\|d_0\|=\mathrm{O}(n^{-1/2+c_{\Xi,1}})$.
	For $\|d\|\le C_1n^{-1/2+c_{\Xi,1}}$ for some $C_1>0$, Taylor's expansion gives
	\begin{align*}
		G_n(x,y)=G_n(x_0,y_0)+DG_n(x_0,y_0)d+R_n(d),
	\end{align*}
	where, after possibly shrinking $\delta_0$, the second derivatives of $G_n$ are uniformly bounded on $\mathcal N_{\delta_0}$. Hence, we have $ \|R_n(d)\|\le C_2\|d\|^2,$ for some constant $C_2>0$. Equivalently, the equation $G_n(x,y)=0$ can be written as
	\begin{align*}
		d=-[DG_n(x_0,y_0)]^{-1}G_n(x_0,y_0)-
		[DG_n(x_0,y_0)]^{-1}R_n(d)
		=
		d_0-\,[DG_n(x_0,y_0)]^{-1}R_n(d).
	\end{align*}
	Define the map
	\begin{align*}
		\mathcal T(d)
		:=
		d_0-\,[DG_n(x_0,y_0)]^{-1}R_n(d).
	\end{align*}
	For $C_3>C$ large enough and $n$ sufficiently large, $\mathcal T$ maps the
	ball $\mathcal B_n:=\{d:\|d\|\le C_3n^{-1/2+c_{\Xi,1}}$
	into itself, because
	\begin{align*}
		\|\mathcal T(d)\|
		\le
		C n^{-1/2+c_{\Xi,1}}+CC_2\|d\|^2
		\le
		C n^{-1/2+c_{\Xi,1}}+C C_1^2 C_2 n^{-1+2c_{\Xi,1}}
		\le
		C_3n^{-1/2+c_{\Xi,1}}.
	\end{align*}
	Moreover, for $d,d^{\prime}\in\mathcal B_n$,
	\begin{align*}
		\|\mathcal T(d)-\mathcal T(d^{\prime})\|
		\le
		C(\|d\|+\|d^{\prime}\|)\|d-d^{\prime}\|
		\le
		C n^{-1/2+c_{\Xi,1}}\|d-d^{\prime}\|.
	\end{align*}
	Since the coefficient on the right-hand side is $\mathrm{o}(1)$, $\mathcal T$ is a contraction on $\mathcal B_n$ for all sufficiently large $n$. By the contraction mapping theorem, there exists a unique $d_*\in\mathcal B_n$ such that $\mathcal T(d_*)=d_*$. Setting $x_1:=x_0+d_{*,1},
	y_1:=y_0+d_{*,2},$ we obtain $G_n(x_1,y_1)=0;$ that is,
	\begin{align*}
		F_n(x_1,y_1)=0,
		\qquad
		\partial_xF_n(x_1,y_1)=0.
	\end{align*}
	Furthermore,
	\begin{align*}
		|x_1-x_0|+|y_1-y_0|
		\le C\|d_*\|
		=\mathrm{O}(n^{-1/2+c_{\Xi,1}}).
	\end{align*}
	
	The same contraction argument also gives uniqueness in the local ball $\mathcal B_n$. After shrinking $\delta_0$, this is the unique local solution in $\mathcal N_{\delta_0}$. The proof is complete.
\end{proof}
We now prove the second statement of Lemma \ref{lem_edgeconvergence1}. Using \eqref{eq_edgeequations1} and \eqref{eq_edgeequations2}, we obtain
\begin{align}\label{eq_generaldecomposition}
	&m_{1n,c}(\mathsf{E}_{\mathrm{MB}})-m_{1n}(\widehat{\mathsf{E}}_{\mathrm{MB}})\\
	&=\frac{1}{n}\sum_i\frac{\sigma_i}{\widehat{\mathsf{E}}_{\mathrm{MB}}-\frac{\sigma_i}{n}\sum_j\frac{\xi^2_j}{1+\xi^2_jm_{1n}(\widehat{\mathsf{E}}_{\mathrm{MB}})}}-\frac{1}{n}\sum_i\frac{\sigma_i}{\mathsf{E}_{\mathrm{MB}}-\frac{\sigma_i}{n}\sum_j\frac{\xi^2_j}{1+\xi^2_jm_{1n,c}(\mathsf{E}_{\mathrm{MB}})}} \nonumber \\
	&+\frac{1}{n}\sum_i\frac{-\sigma_i\int\frac{t}{1+tm_{1n,c}(\mathsf{E}_{\mathrm{MB}})}\mathrm{d}F_{\xi^2}(t)+\frac{\sigma_i}{n}\sum_j\frac{\xi^2_j}{1+\xi^2_jm_{1n,c}(\mathsf{E}_{\mathrm{MB}})}}{(\mathsf{E}_{\mathrm{MB}}-\frac{\sigma_i}{n}\sum_j\frac{\xi^2_j}{1+\xi^2_jm_{1n,c}(\mathsf{E}_{\mathrm{MB}})})(\mathsf{E}_{\mathrm{MB}}-\sigma_i\int\frac{t}{1+tm_{1n,c}(\mathsf{E}_{\mathrm{MB}})}\mathrm{d}F_{\xi^2}(t))} \nonumber \\
	&=\frac{1}{n}\sum_i\frac{\sigma_i(\mathsf{E}_{\mathrm{MB}}-\widehat{\mathsf{E}}_{\mathrm{MB}})}{(\widehat{\mathsf{E}}_{\mathrm{MB}}-\frac{\sigma_i}{n}\sum_j\frac{\xi^2_j}{1+\xi^2_jm_{1n,c}(\mathsf{E}_{\mathrm{MB}})})(\mathsf{E}_{\mathrm{MB}}-\frac{\sigma_i}{n}\sum_j\frac{\xi^2_j}{1+\xi^2_jm_{1n,c}(\mathsf{E}_{\mathrm{MB}})})} \nonumber \\
	&+\frac{1}{n}\sum_i\frac{-\frac{\sigma_i^2}{n}\sum_j\frac{\xi^4_j(m_{1n}(\widehat{\mathsf{E}}_{\mathrm{MB}})-m_{1n,c}(\mathsf{E}_{\mathrm{MB}}))}{(1+\xi^2_jm_{1n}(\widehat{\mathsf{E}}_{\mathrm{MB}}))(1+\xi^2_jm_{1n,c}(\mathsf{E}_{\mathrm{MB}}))}}{(\widehat{\mathsf{E}}_{\mathrm{MB}}-\frac{\sigma_i}{n}\sum_j\frac{\xi^2_j}{1+\xi^2_jm_{1n}(\widehat{\mathsf{E}}_{\mathrm{MB}})})(\widehat{\mathsf{E}}_{\mathrm{MB}}-\frac{\sigma_i}{n}\sum_j\frac{\xi^2_j}{1+\xi^2_jm_{1n,c}(\mathsf{E}_{\mathrm{MB}})})} \nonumber \\
	&+\frac{1}{n}\sum_i\frac{-\sigma^2_i\int\frac{t}{1+tm_{1n,c}(\mathsf{E}_{\mathrm{MB}})}\mathrm{d}F_{\xi^2}(t)+\frac{\sigma^2_i}{n}\sum_j\frac{\xi^2_j}{1+\xi^2_jm_{1n,c}(\mathsf{E}_{\mathrm{MB}})}}{(\mathsf{E}_{\mathrm{MB}}-\frac{\sigma_i}{n}\sum_j\frac{\xi^2_j}{1+\xi^2_jm_{1n,c}(\mathsf{E}_{\mathrm{MB}})})(\mathsf{E}_{\mathrm{MB}}-\sigma_i\int\frac{t}{1+tm_{1n,c}(\mathsf{E}_{\mathrm{MB}})}\mathrm{d}F_{\xi^2}(t))} \nonumber \\
	&:=\mathsf{T}_1+\mathsf{T}_2+\mathsf{T}_3. 
\end{align}
By the first statement of Lemma \ref{lem_edgeconvergence1}, together with Assumption \ref{assum_technique} and Definition \ref{def_OmegaM}, on the event $\Omega_{\Xi}$ we have
\begin{equation}\label{eq_T1control}
	\mathsf{T}_1=\mathsf{C}_1(\mathsf{E}_{\mathrm{MB}}-\widehat{\mathsf{E}}_{\mathrm{MB}})+\mathrm{O}(n^{-1+2c_{\Xi,1}}).
\end{equation}

For the term $\mathsf{T}_2,$ we see that on the event $\Omega_{\Xi}$,

\begin{align}\label{eq_T2control}
	&\mathsf{T}_2=\frac{1}{n}\sum_i\frac{-\frac{\sigma_i^2}{n}\sum_j\frac{\xi^4_j(m_{1n}(\widehat{\mathsf{E}}_{\mathrm{MB}})-m_{1n,c}(\mathsf{E}_{\mathrm{MB}}))}{(1+\xi^2_jm_{1n}(\widehat{\mathsf{E}}_{\mathrm{MB}}))(1+\xi^2_jm_{1n,c}(L_{+}))}}{(\widehat{\mathsf{E}}_{\mathrm{MB}}-\frac{\sigma_i}{n}\sum_j\frac{\xi^2_j}{1+\xi^2_jm_{1n}(\widehat{\mathsf{E}}_{\mathrm{MB}})})^2} \nonumber \\
	&+\frac{1}{n}\sum_i\frac{\big(\frac{\sigma_i^2}{n}\sum_j\frac{\xi^4_j}{(1+\xi^2_jm_{1n}(\widehat{\mathsf{E}}_{\mathrm{MB}}))(1+\xi^2_jm_{1n,c}(L_{+}))}\big)^2(m_{1n}(\widehat{\mathsf{E}}_{\mathrm{MB}})-m_{1n,c}(\mathsf{E}_{\mathrm{MB}}))^2}{(\widehat{\mathsf{E}}_{\mathrm{MB}}-\frac{\sigma_i}{n}\sum_j\frac{\xi^2_j}{1+\xi^2_jm_{1n}(\widehat{\mathsf{E}}_{\mathrm{MB}})})^2(\widehat{\mathsf{E}}_{\mathrm{MB}}-\frac{\sigma_i}{n}\sum_j\frac{\xi^2_j}{1+\xi^2_jm_{1n,c}(\mathsf{E}_{\mathrm{MB}})})} \nonumber \\
	&=\frac{1}{n}\sum_i\frac{-\frac{\sigma_i^2}{n}\sum_j\frac{\xi^4_j(m_{1n}(\widehat{\mathsf{E}}_{\mathrm{MB}})-m_{1n,c}(\mathsf{E}_{\mathrm{MB}}))}{(1+\xi^2_jm_{1n}(\widehat{\mathsf{E}}_{\mathrm{MB}}))^2}}{(\widehat{\mathsf{E}}_{\mathrm{MB}}-\frac{\sigma_i}{n}\sum_j\frac{\xi^2_j}{1+\xi^2_jm_{1n}(\widehat{\mathsf{E}}_{\mathrm{MB}})})^2}+\frac{1}{n}\sum_i\frac{-\frac{\sigma_i^2}{n}\sum_j\frac{\xi^4_j(m_{1n}(\widehat{\mathsf{E}}_{\mathrm{MB}})-m_{1n,c}(\mathsf{E}_{\mathrm{MB}}))^2}{(1+\xi^2_jm_{1n}(\widehat{\mathsf{E}}_{\mathrm{MB}}))^2(1+\xi^2_jm_{1n,c}(\mathsf{E}_{\mathrm{MB}}))}}{(\widehat{\mathsf{E}}_{\mathrm{MB}}-\frac{\sigma_i}{n}\sum_j\frac{\xi^2_j}{1+\xi^2_jm_{1n,c}(\mathsf{E}_{\mathrm{MB}})})^2} \nonumber \\
	&+\frac{1}{n}\sum_i\frac{\big(\frac{\sigma_i^2}{n}\sum_j\frac{\xi^4_j}{(1+\xi^2_jm_{1n}(\widehat{\mathsf{E}}_{\mathrm{MB}}))(1+\xi^2_jm_{1n,c}(\mathsf{E}_{\mathrm{MB}}))}\big)^2(m_{1n}(\widehat{\mathsf{E}}_{\mathrm{MB}})-m_{1n,c}(\mathsf{E}_{\mathrm{MB}}))^2}{(\widehat{\mathsf{E}}_{\mathrm{MB}}-\frac{\sigma_i}{n}\sum_j\frac{\xi^2_j}{1+\xi^2_jm_{1n}(\widehat{\mathsf{E}}_{\mathrm{MB}})})(\widehat{\mathsf{E}}_{\mathrm{MB}}-\frac{\sigma_i}{n}\sum_j\frac{\xi^2_j}{1+\xi^2_jm_{1n,c}(\mathsf{E}_{\mathrm{MB}})})^2}  \\
	&=-(m_{1n}(\widehat{\mathsf{E}}_{\mathrm{MB}})-m_{1n,c}(\mathsf{E}_{\mathrm{MB}}))+\mathrm{O}(n^{-1+2c_{\Xi,1}}), \nonumber
\end{align}
where we used \eqref{eq_edgeequations1}, \eqref{eq_edgeconvergence_firststatement}, Definition \ref{def_OmegaM}, and Assumption \ref{assum_technique}. For $\mathsf{T}_3$, we have
\begin{equation}\label{eq_T3control}
	\mathsf{T}_3=\mathsf{C}_2 \mathcal{X}+\mathrm{O}(n^{-1+2c_{\Xi,1}}). 
\end{equation}
Combining \eqref{eq_T1control}, \eqref{eq_T2control}, and \eqref{eq_T3control} with \eqref{eq_generaldecomposition} completes the proof of Lemma \ref{lem_edgeconvergence1}.

\subsection{Proof of Lemma \ref{lem_probabilitycontrol_M}}\label{sec_prf_Omega_M}
In this subsection, we prove Lemma \ref{lem_probabilitycontrol_M}. We focus on the feasible multiplier construction introduced in Example \ref{example_one} under the constraints of Definition \ref{defn_feasiblemultiplier}.

The first and third statements follow directly from Chebyshev's inequality, while the fourth follows directly from Markov's inequality. We omit the routine details. 

For the second statement, define the centered random variables $\mathbf{b}_{\xi^2_i}$, $1\leq i\leq n$, by
\begin{gather*}
	\mathbf{b}_{\xi^2_i}:=\frac{\xi^2_i}{1+\xi^2_im_{1n,c}(\mathsf{E}_{\mathrm{MB}})}-\int\frac{t}{1+tm_{1n,c}(\mathsf{E}_{\mathrm{MB}})}\mathrm{d}F_{\xi^2}(t).
\end{gather*}
By construction, $\mathbb{E}\mathbf{b}_{\xi^2_i}=0$ for $1\le i\le n$. On the other hand, since
\begin{align*}
	\frac{1}{n} \sum_{i=1}^p \frac{\sigma_i^2 \int \frac{t^2}{|1+t m_{1n,c}(\mathsf{E}_{\mathrm{MB}})|^2} \mathrm{d} F_{\xi^2}(t) }{|-\mathsf{E}_{\mathrm{MB}}+\sigma_i \int \frac{t}{1+t m_{1n,c}(\mathsf{E}_{\mathrm{MB}})} \mathrm{d} F_{\xi^2}(t)|^2}=1,
\end{align*}
by Assumption \ref{assum_technique}, we have
\begin{align*}
	\int \frac{t^2}{|1+t m_{1n,c}(\mathsf{E}_{\mathrm{MB}})|^2} \mathrm{d} F_{\xi^2}(t)\le C,
\end{align*}
for some constant $C>0$. Consequently, by the Cauchy--Schwarz inequality, we obtain constants $C_1, C_2>0$ such that
\begin{gather*}
	\mathbb{E}|\mathbf{b}_{\xi^2}|^2\leq C_1  \int\frac{t^2}{|1+tm_{1n,c}(\mathsf{E}_{\mathrm{MB}})|^2}\mathrm{d}F_{\xi^2}(t) < C_2<\infty. 
\end{gather*}
Since the variables $\mathbf{b}_{\xi_i^2}$, $1 \leq i \leq n$, are independent, the proof is completed by another application of Markov's inequality.

\subsection{Proof of Lemma \ref{lem_power_bootstrap_gap}}\label{sec_proof_lemma_bootstrapgap}

Fix $1\le i\le r$. Write
\begin{align*}
	x_i:=m_{1n,c}(\vartheta_i^{\mathtt{MB}}).
\end{align*}
By \eqref{eq_def_vartheta_MB},
\[
1+\widetilde{\sigma}_i m_{2n,c}(\vartheta_i^{\mathtt{MB}})=0,
\]
we have
\begin{align}\label{eq_power_gap_m1}
	x_i
	=
	\frac{1}{n}\sum_{a=1}^p\frac{\sigma_a}{-\vartheta_i^{\mathtt{MB}}(1-\sigma_a/\widetilde{\sigma}_i)}
	=
	\frac{\widetilde{\sigma}_i}{\vartheta_i^{\mathtt{MB}}}\cdot
	\frac{1}{n}\sum_{a=1}^p\frac{\sigma_a}{\sigma_a-\widetilde{\sigma}_i}.
\end{align}
On the other hand,
\begin{align}\label{eq_power_gap_m2}
	m_{2n,c}(\vartheta_i^{\mathtt{MB}})
	=
	\int\frac{t}{-\vartheta_i^{\mathtt{MB}}(1+t x_i)}\,\mathrm{d}F_{\xi^2}(t)
	=
	-\widetilde{\sigma}_i^{-1},
\end{align}
which is equivalent to
\begin{align}\label{eq_power_gap_theta_rep}
	\vartheta_i^{\mathtt{MB}}
	=
	\widetilde{\sigma}_i\int\frac{t}{1+t x_i}\,\mathrm{d}F_{\xi^2}(t).
\end{align}
Combining \eqref{eq_power_gap_m1} and \eqref{eq_power_gap_theta_rep}, we obtain
\begin{align}\label{eq_power_gap_root}
	x_i\int\frac{t}{1+t x_i}\,\mathrm{d}F_{\xi^2}(t)
	=
	\frac{1}{n}\sum_{a=1}^p\frac{\sigma_a}{\sigma_a-\widetilde{\sigma}_i}.
\end{align}

Let
\begin{align*}
	\kappa_i:=\frac{1}{n}\sum_{a=1}^p\frac{\sigma_a}{\sigma_a-\widetilde{\sigma}_i},
	\qquad
	T_i(x):=
	x\int\frac{t}{1+t x}\,\mathrm{d}F_{\xi^2}(t)-\kappa_i.
\end{align*}
Then \eqref{eq_power_gap_root} implies that $T_i(x_i)=0$.

Parallelly, we can introduce the following equation associated with $Q_{\text{sam}}$ without the multiplier
\begin{align*}
	T_i^{(0)}(x):=\frac{x}{1+x}-\kappa_i.
\end{align*}
Let $x_i^{(0)}$ be the unique solution of $T_i^{(0)}(x)=0$. Since
\begin{align*}
	\frac{x_i^{(0)}}{1+x_i^{(0)}}=\kappa_i,
\end{align*}
we have
\begin{align*}
	\frac{\widetilde{\sigma}_i}{1+x_i^{(0)}}
	=
	\widetilde{\sigma}_i(1-\kappa_i)
	=
	\widetilde{\sigma}_i
	-
	\widetilde{\sigma}_i\frac{1}{n}\sum_{a=1}^p\frac{\sigma_a}{\sigma_a-\widetilde{\sigma}_i}.
\end{align*}
On the other hand, by the definition \eqref{eq_def_vartheta_S},
\begin{align*}
	\vartheta_i^{\mathtt{S}}
	=
	f(-\widetilde{\sigma}_i^{-1})
	=
	\widetilde{\sigma}_i
	+
	\frac{1}{n}\sum_{a=1}^p\frac{1}{-\widetilde{\sigma}_i^{-1}+\sigma_a^{-1}}
	=
	\widetilde{\sigma}_i
	-
	\widetilde{\sigma}_i\frac{1}{n}\sum_{a=1}^p\frac{\sigma_a}{\sigma_a-\widetilde{\sigma}_i}.
\end{align*}
Hence,  for $i=1,\dots, r$
\begin{align}\label{eq_power_gap_thetaS_rep}
	\vartheta_i^{\mathtt{S}}
	=
	\frac{\widetilde{\sigma}_i}{1+x_i^{(0)}}.
\end{align}

Let $\mathcal N_i$ be a fixed neighborhood of $x_i^{(0)}$. By the truncation remark after Lemma \ref{lem_uniform_multiplier_concentration} and Assumption \ref{assum_technique}, after shrinking $\mathcal N_i$ if necessary, there exists a constant $c>0$ such that
\begin{align}\label{eq_power_gap_den}
	|1+t x|\ge c,\qquad x\in\mathcal N_i,\ t\in \operatorname{supp}(F_{\xi^2}).
\end{align}
Therefore, for
\begin{align*}
	h_x(t):=\frac{t}{1+t x},
\end{align*}
the derivatives $\partial_t h_x(t)$ and $\partial_{tt} h_x(t)$ are uniformly bounded on $\mathcal N_i\times \operatorname{supp}(F_{\xi^2})$.
Since $\mathbb E\xi^2=1$, Taylor expansion at $t=1$ yields, uniformly for $x\in\mathcal N_i$,
\begin{align*}
	h_x(t)
	=
	h_x(1)+h_x'(1)(t-1)+\mathrm{O}((t-1)^2).
\end{align*}
Integrating with respect to $F_{\xi^2}$ and using Definition \ref{defn_feasiblemultiplier}, we obtain
\begin{align}\label{eq_power_gap_H}
	\int\frac{t}{1+t x}\,\mathrm{d}F_{\xi^2}(t)
	=
	\frac{1}{1+x}
	+
	\mathrm{O}\big(\operatorname{Var}(\xi^2)\big),
\end{align}
uniformly for $x\in\mathcal N_i$. Consequently,
\begin{align}\label{eq_power_gap_T}
	T_i(x)-T_i^{(0)}(x)
	=
	\mathrm{O}\big(\operatorname{Var}(\xi^2)\big),
\end{align}
uniformly for $x\in\mathcal N_i$.

Moreover,
\begin{align*}
	\big(T_i^{(0)}\big)'(x)=\frac{1}{(1+x)^2},
\end{align*}
in particular
\begin{align}\label{eq_power_gap_derivative}
	\big(T_i^{(0)}\big)'(x_i^{(0)})=\frac{1}{(1+x_i^{(0)})^2}\asymp 1.
\end{align}
Therefore, by \eqref{eq_power_gap_T}, \eqref{eq_power_gap_derivative}, and the stability of simple roots,
\begin{align}\label{eq_power_gap_x}
	x_i-x_i^{(0)}
	=
	\mathrm{O}\big(\operatorname{Var}(\xi^2)\big).
\end{align}

Finally, combining \eqref{eq_power_gap_theta_rep}, \eqref{eq_power_gap_thetaS_rep}, \eqref{eq_power_gap_H}, and \eqref{eq_power_gap_x}, we get
\begin{align*}
	\vartheta_i^{\mathtt{MB}}-\vartheta_i^{\mathtt{S}}
	&=
	\widetilde{\sigma}_i
	\left(
	\int\frac{t}{1+t x_i}\,\mathrm{d}F_{\xi^2}(t)-\frac{1}{1+x_i^{(0)}}
	\right)\\
	&=
	\widetilde{\sigma}_i
	\left(
	\int\frac{t}{1+t x_i}\,\mathrm{d}F_{\xi^2}(t)-\frac{1}{1+x_i}
	\right)
	+
	\widetilde{\sigma}_i
	\left(
	\frac{1}{1+x_i}-\frac{1}{1+x_i^{(0)}}
	\right)\\
	&=
	\mathrm{O}\big(\operatorname{Var}(\xi^2)\big),
\end{align*}
which proves \eqref{eq_power_gap}. Combining \eqref{eq_power_gap} with Lemma \ref{lem_power_outlier_expansion}, we obtain
\[
\vartheta_i^{\mathtt{MB}}-\mathsf E
=
\frac{1}{2}f''(\mathsf b)\mathsf b^4\varsigma_i^2
+
\mathrm{O}\big(\varsigma_i^3+\operatorname{Var}(\xi^2)\big).
\]
We then complete the proof of \eqref{eq_power_theta_c_gap}.

\subsection{Proof of Lemma \ref{lem_power_bootstrap_localization}}\label{sec_proof_lemma_bootstraplocation}

Fix $1\le i\le r$. Write
\begin{align*}
	\widehat{x}_i:=m_{1n}(\widehat{\vartheta}_i^{\mathtt{MB}}),
	\qquad
	x_i:=m_{1n,c}(\vartheta_i^{\mathtt{MB}}).
\end{align*}
As in the proof of Lemma \ref{lem_power_bootstrap_gap}, the equations
\[
1+\widetilde{\sigma}_im_{2n}(\widehat{\vartheta}^{\mathtt{MB}}_i)=0,
\qquad
1+\widetilde{\sigma}_im_{2n,c}(\vartheta^{\mathtt{MB}}_i)=0
\]
imply
\begin{align*}
	\widehat{\vartheta}^{\mathtt{MB}}_i
	=
	\widetilde{\sigma}_i\frac{1}{n}\sum_{j=1}^n\frac{\xi_j^2}{1+\xi_j^2\widehat{x}_i},
	\qquad
	\vartheta^{\mathtt{MB}}_i
	=
	\widetilde{\sigma}_i\int\frac{t}{1+t x_i}\,\mathrm{d}F_{\xi^2}(t),
\end{align*}
and
\begin{align*}
	\widehat{x}_i\frac{\widehat{\vartheta}^{\mathtt{MB}}_i}{\widetilde{\sigma}_i}
	=
	\kappa_i,
	\qquad
	x_i\frac{\vartheta^{\mathtt{MB}}_i}{\widetilde{\sigma}_i}
	=
	\kappa_i,
\end{align*}
where
\[
\kappa_i=\frac{1}{n}\sum_{a=1}^p\frac{\sigma_a}{\sigma_a-\widetilde{\sigma}_i}.
\]
Equivalently,
\begin{align}\label{eq_power_loc_root}
	\widehat{T}_i(\widehat{x}_i)=0,
	\qquad
	T_i(x_i)=0,
\end{align}
where
\begin{align*}
	\widehat{T}_i(x):=
	x\frac{1}{n}\sum_{j=1}^n\frac{\xi_j^2}{1+\xi_j^2x}-\kappa_i,
	\qquad
	T_i(x):=
	x\int\frac{t}{1+t x}\,\mathrm{d}F_{\xi^2}(t)-\kappa_i.
\end{align*}

By Lemma \ref{lem_power_bootstrap_gap},
\[
\vartheta_i^{\mathtt{MB}}-\vartheta_i^{\mathtt{S}}
=
\mathrm{O}\big(\operatorname{Var}(\xi^2)\big).
\]
Hence $x_i$ lies in a fixed neighborhood $\mathcal N_i$ of the reference root $x_i^{(0)}$ introduced in the proof of Lemma \ref{lem_power_bootstrap_gap}. On this neighborhood, the denominator separation \eqref{eq_power_gap_den} holds. In particular,
\begin{align*}
	T_i'(x)
	=
	\int\frac{t}{(1+t x)^2}\,\mathrm{d}F_{\xi^2}(t)
\end{align*}
is continuous and strictly positive on $\mathcal N_i$, so there exist constants $0<c<C<\infty$ such that
\begin{align}\label{eq_power_loc_derivative_det}
	c\le T_i'(x)\le C,\qquad x\in\mathcal N_i.
\end{align}

Moreover,
\begin{align*}
	\widehat{T}_i'(x)-T_i'(x)
	=
	\frac{1}{n}\sum_{j=1}^n\frac{\xi_j^2}{(1+\xi_j^2x)^2}
	-
	\int\frac{t}{(1+t x)^2}\,\mathrm{d}F_{\xi^2}(t).
\end{align*}
By the same grid argument as in Lemma \ref{lem_uniform_multiplier_concentration}, applied to $h(t,x)=t(1+t x)^{-2}$, we have on the event $\Omega_{\Xi}$
\begin{align}\label{eq_power_loc_derivative_rand}
	\sup_{x\in\mathcal N_i}
	|\widehat{T}_i'(x)-T_i'(x)|
	=
	\mathrm{O}(n^{-1/2+c_{\Xi,1}}).
\end{align}
Combining \eqref{eq_power_loc_derivative_det} and \eqref{eq_power_loc_derivative_rand}, we may assume on $\Omega_{\Xi}$ that
\begin{align}\label{eq_power_loc_derivative_hat}
	c/2\le \widehat{T}_i'(x)\le 2C,\qquad x\in\mathcal N_i.
\end{align}

We next estimate $\widehat{T}_i(x_i)$. Let
\begin{align*}
	q_i(t):=\frac{t}{1+t x_i}.
\end{align*}
Then
\begin{align}\label{eq_power_loc_Txi}
	\widehat{T}_i(x_i)
	=
	x_i\left(
	\frac{1}{n}\sum_{j=1}^n q_i(\xi_j^2)-\int q_i(t)\,\mathrm{d}F_{\xi^2}(t)
	\right).
\end{align}
Since $x_i\in\mathcal N_i$ and \eqref{eq_power_gap_den} hold, the derivative
\[
q_i'(t)=\frac{1}{(1+t x_i)^2}
\]
is uniformly bounded on $\operatorname{supp}(F_{\xi^2})$. Therefore, if $\xi_*^2$ is an independent copy of $\xi^2$,
\begin{align*}
	\operatorname{Var}(q_i(\xi^2))
	=
	\frac{1}{2}\mathbb E\big(q_i(\xi^2)-q_i(\xi_*^2)\big)^2
	\le
	C\frac{1}{2}\mathbb E(\xi^2-\xi_*^2)^2
	=
	C\operatorname{Var}(\xi^2).
\end{align*}
Because the multiplier law is bounded after truncation, $q_i(\xi^2)$ is uniformly bounded. Hence Bernstein's inequality yields

\begin{align}\label{eq_power_loc_q}
	\frac{1}{n}\sum_{j=1}^n q_i(\xi_j^2)-\int q_i(t)\,\mathrm{d}F_{\xi^2}(t)
	=
	\mathrm{O}\!\left(n^{-1/2+c_{\Xi,1}}\sqrt{\operatorname{Var}(\xi^2)}\right)
\end{align}
with probability $1-\mathrm{o}(1)$. Then on the event $\Omega_{\Xi}$, \eqref{eq_power_loc_Txi} becomes
\begin{align}\label{eq_power_loc_Tsmall}
	\widehat{T}_i(x_i)
	=
	\mathrm{O}\!\left(n^{-1/2+c_{\Xi,1}}\sqrt{\operatorname{Var}(\xi^2)}\right).
\end{align}

Since $\widehat{T}_i(\widehat{x}_i)=0$, the mean-value theorem and \eqref{eq_power_loc_derivative_hat} imply
\begin{align}\label{eq_power_loc_x}
	\widehat{x}_i-x_i
	=
	\mathrm{O}\!\left(n^{-1/2+c_{\Xi,1}}\sqrt{\operatorname{Var}(\xi^2)}\right).
\end{align}

Finally,
\begin{align*}
	\widehat{\vartheta}_i^{\mathtt{MB}}-\vartheta_i^{\mathtt{MB}}
	&=
	\widetilde{\sigma}_i
	\left(
	\frac{1}{n}\sum_{j=1}^n\frac{\xi_j^2}{1+\xi_j^2\widehat{x}_i}
	-
	\int\frac{t}{1+t x_i}\,\mathrm{d}F_{\xi^2}(t)
	\right)\\
	&=
	\widetilde{\sigma}_i
	\left(
	\frac{1}{n}\sum_{j=1}^n\frac{\xi_j^2}{1+\xi_j^2\widehat{x}_i}
	-
	\frac{1}{n}\sum_{j=1}^n\frac{\xi_j^2}{1+\xi_j^2x_i}
	\right)
	+
	\widetilde{\sigma}_i
	\left(
	\frac{1}{n}\sum_{j=1}^n\frac{\xi_j^2}{1+\xi_j^2x_i}
	-
	\int\frac{t}{1+t x_i}\,\mathrm{d}F_{\xi^2}(t)
	\right).
\end{align*}
The second term is exactly the fluctuation in \eqref{eq_power_loc_q}. For the first term, the map
\[
x\mapsto \frac{1}{n}\sum_{j=1}^n\frac{\xi_j^2}{1+\xi_j^2x}
\]
is uniformly Lipschitz on $\mathcal N_i$ by \eqref{eq_power_gap_den}. Therefore, using \eqref{eq_power_loc_x},
\begin{align*}
	\frac{1}{n}\sum_{j=1}^n\frac{\xi_j^2}{1+\xi_j^2\widehat{x}_i}
	-
	\frac{1}{n}\sum_{j=1}^n\frac{\xi_j^2}{1+\xi_j^2x_i}
	=
	\mathrm{O}\!\left(n^{-1/2+c_{\Xi,1}}\sqrt{\operatorname{Var}(\xi^2)}\right).
\end{align*}
Combining the last two displays proves
\[
\widehat{\vartheta}_i^{\mathtt{MB}}-\vartheta_i^{\mathtt{MB}}
=
\mathrm{O}\!\left(n^{-1/2+c_{\Xi,1}}\sqrt{\operatorname{Var}(\xi^2)}\right).
\]

\bibliographystyle{chicago}

\bibliography{ref}

\end{document}